\documentclass[10pt,journal,compsoc]{IEEEtran}
\IEEEoverridecommandlockouts
\usepackage{cite}
\usepackage{amsmath,amssymb,amsfonts}
\usepackage{algorithm}
\usepackage{algpseudocode}
\usepackage{graphicx}
\usepackage{textcomp}
\usepackage{comment}
\usepackage{booktabs}
\usepackage{amsthm}
\usepackage{amsmath}
\usepackage{tabularx}
\usepackage{threeparttable}
\usepackage{breakurl}
\usepackage[normalem]{ulem}
\usepackage{url}
\usepackage[nolist]{acronym}
\usepackage[official]{eurosym}

\usepackage{subfigure}

\usepackage{textcomp}

\usepackage{multirow}

\usepackage{colortbl}

\usepackage{color}

\usepackage{varwidth}
\usepackage{rotating}

\usepackage{nameref}
\usepackage{zref-xr,zref-user}

\definecolor{Linen}{rgb}{0.9803,0.9411,0.9019} 
\definecolor{White}{rgb}{1,1,1}
\definecolor{Coral}{rgb}{1,0.4980,0.3137}
\definecolor{Grayblue}{rgb}{0.9411,0.9411,0.9803}
\definecolor{DarkLinen}{rgb}{0.729,0.7176,0.635}

\definecolor{Linen}{rgb}{0.9803,0.9411,0.9019}
\definecolor{White}{rgb}{1,1,1}
\definecolor{Coral}{rgb}{1,0.4980,0.3137}
\definecolor{Grayblue}{rgb}{0.9411,0.9411,0.9803}
\definecolor{DarkLinen}{rgb}{0.729,0.7176,0.635}

\makeatletter
\newcommand*{\Strut}[1][0.1em]{\vrule\@width\z@\@height#1\@depth\z@\relax}
\makeatother

\def\BibTeX{{\rm B\kern-.05em{\sc i\kern-.025em b}\kern-.08em
    T\kern-.1667em\lower.7ex\hbox{E}\kern-.125emX}}

\newtheorem{theorem}{Theorem}

\newtheorem{Proposition}[theorem]{Proposition}

\makeatletter
\def\Cline#1#2{\@Cline#1#2\@nil}
\def\@Cline#1-#2#3\@nil{%
  \omit
  \@multicnt#1%
  \advance\@multispan\m@ne
  \ifnum\@multicnt=\@ne\@firstofone{&\omit}\fi
  \@multicnt#2%
  \advance\@multicnt-#1%
  \advance\@multispan\@ne
  \leaders\hrule\@height#3\hfill
  \cr}
\makeatother

\begin{document}

\begin{acronym}

\acro{5G-NR}{5G New Radio}
\acro{3GPP}{3rd Generation Partnership Project}
\acro{AC}{address coding}
\acro{ACF}{autocorrelation function}
\acro{ACR}{autocorrelation receiver}
\acro{ADC}{analog-to-digital converter}
\acrodef{aic}[AIC]{Analog-to-Information Converter}     
\acro{AIC}[AIC]{Akaike information criterion}
\acro{aric}[ARIC]{asymmetric restricted isometry constant}
\acro{arip}[ARIP]{asymmetric restricted isometry property}

\acro{ARQ}{Automatic Repeat Request}
\acro{AUB}{asymptotic union bound}
\acrodef{awgn}[AWGN]{Additive White Gaussian Noise}     
\acro{AWGN}{additive white Gaussian noise}

\acro{APSK}[PSK]{asymmetric PSK} 

\acro{waric}[AWRICs]{asymmetric weak restricted isometry constants}
\acro{warip}[AWRIP]{asymmetric weak restricted isometry property}
\acro{BCH}{Bose, Chaudhuri, and Hocquenghem}        
\acro{BCHC}[BCHSC]{BCH based source coding}
\acro{BEP}{bit error probability}
\acro{BFC}{block fading channel}
\acro{BG}[BG]{Bernoulli-Gaussian}
\acro{BGG}{Bernoulli-Generalized Gaussian}
\acro{BPAM}{binary pulse amplitude modulation}
\acro{BPDN}{Basis Pursuit Denoising}
\acro{BPPM}{binary pulse position modulation}
\acro{BPSK}{Binary Phase Shift Keying}
\acro{BPZF}{bandpass zonal filter}
\acro{BSC}{binary symmetric channels}              
\acro{BU}[BU]{Bernoulli-uniform}
\acro{BER}{bit error rate}
\acro{BS}{base station}

\acro{CAPEX}{CAPital EXpenditures}
\acro{CP}{Cyclic Prefix}
\acrodef{cdf}[CDF]{cumulative distribution function}   
\acro{CDF}{Cumulative Distribution Function}
\acrodef{c.d.f.}[CDF]{cumulative distribution function}
\acro{CCDF}{complementary cumulative distribution function}
\acrodef{ccdf}[CCDF]{complementary CDF}               
\acrodef{c.c.d.f.}[CCDF]{complementary cumulative distribution function}
\acro{CD}{cooperative diversity}

\acro{CDMA}{Code Division Multiple Access}
\acro{ch.f.}{characteristic function}
\acro{CIR}{channel impulse response}
\acro{cosamp}[CoSaMP]{compressive sampling matching pursuit}
\acro{CR}{cognitive radio}
\acro{cs}[CS]{compressed sensing}                   
\acrodef{cscapital}[CS]{Compressed sensing} 
\acrodef{CS}[CS]{compressed sensing}
\acro{CSI}{channel state information}
\acro{CCSDS}{consultative committee for space data systems}
\acro{CC}{convolutional coding}
\acro{Covid19}[COVID-19]{Coronavirus disease}

\acro{DAA}{detect and avoid}
\acro{DAB}{digital audio broadcasting}
\acro{DCT}{discrete cosine transform}
\acro{dft}[DFT]{discrete Fourier transform}
\acro{DR}{distortion-rate}
\acro{DS}{direct sequence}
\acro{DS-SS}{direct-sequence spread-spectrum}
\acro{DTR}{differential transmitted-reference}
\acro{DVB-H}{digital video broadcasting\,--\,handheld}
\acro{DVB-T}{digital video broadcasting\,--\,terrestrial}
\acro{DL}{DownLink}
\acro{DSSS}{Direct Sequence Spread Spectrum}
\acro{DFT-s-OFDM}{Discrete Fourier Transform-spread-Orthogonal Frequency Division Multiplexing}
\acro{DAS}{Distributed Antenna System}
\acro{DNA}{DeoxyriboNucleic Acid}

\acro{EC}{European Commission}
\acro{EED}[EED]{exact eigenvalues distribution}
\acro{EIRP}{Equivalent Isotropically Radiated Power}
\acro{ELP}{equivalent low-pass}
\acro{eMBB}{enhanced Mobile BroadBand}
\acro{EMF}{ElectroMagnetic Field}
\acro{EU}{European union}

\acro{FC}[FC]{fusion center}
\acro{FCC}{Federal Communications Commission}
\acro{FEC}{forward error correction}
\acro{FFT}{fast Fourier transform}
\acro{FH}{frequency-hopping}
\acro{FH-SS}{frequency-hopping spread-spectrum}
\acrodef{FS}{Frame synchronization}
\acro{FSsmall}[FS]{frame synchronization}  
\acro{FDMA}{Frequency Division Multiple Access}

\acro{GA}{Gaussian approximation}
\acro{GF}{Galois field }
\acro{GG}{Generalized-Gaussian}
\acro{GIC}[GIC]{generalized information criterion}
\acro{GLRT}{generalized likelihood ratio test}
\acro{GPS}{Global Positioning System}
\acro{GMSK}{Gaussian Minimum Shift Keying}
\acro{GSMA}{Global System for Mobile communications Association}

\acro{HAP}{high altitude platform}

\acro{IDR}{information distortion-rate}
\acro{IFFT}{inverse fast Fourier transform}
\acro{iht}[IHT]{iterative hard thresholding}
\acro{i.i.d.}{independent, identically distributed}
\acro{IoT}{Internet of Things}                      
\acro{IR}{impulse radio}
\acro{lric}[LRIC]{lower restricted isometry constant}
\acro{lrict}[LRICt]{lower restricted isometry constant threshold}
\acro{ISI}{intersymbol interference}
\acro{ITU}{International Telecommunication Union}
\acro{ICNIRP}{International Commission on Non-Ionizing Radiation Protection}
\acro{IEEE}{Institute of Electrical and Electronics Engineers}
\acro{ICES}{IEEE international committee on electromagnetic safety}
\acro{IEC}{International Electrotechnical Commission}
\acro{IARC}{International Agency on Research on Cancer}
\acro{IS-95}{Interim Standard 95}

\acro{LEO}{low earth orbit}
\acro{LF}{likelihood function}
\acro{LLF}{log-likelihood function}
\acro{LLR}{log-likelihood ratio}
\acro{LLRT}{log-likelihood ratio test}
\acro{LOS}{Line-of-Sight}
\acro{LRT}{likelihood ratio test}
\acro{wlric}[LWRIC]{lower weak restricted isometry constant}
\acro{wlrict}[LWRICt]{LWRIC threshold}
\acro{LPWAN}{Low Power Wide Area Network}
\acro{LoRaWAN}{Low power long Range Wide Area Network}
\acro{NLOS}{Non-Line-of-Sight}
\acro{LiFi}[Li-Fi]{light-fidelity}
 \acro{LED}{light emitting diode}

\acro{MB}{multiband}
\acro{MC}{multicarrier}
\acro{MDS}{mixed distributed source}
\acro{MF}{matched filter}
\acro{m.g.f.}{moment generating function}
\acro{MI}{mutual information}
\acro{MIMO}{Multiple-Input Multiple-Output}
\acro{MILP}{Mixed Integer Linear Programming}
\acro{MINLP}{Mixed Integer Non-Linear Programming}
\acro{MISO}{multiple-input single-output}
\acrodef{maxs}[MJSO]{maximum joint support cardinality}                       
\acro{ML}[ML]{maximum likelihood}
\acro{MMSE}{minimum mean-square error}
\acro{MMV}{multiple measurement vectors}
\acrodef{MOS}{model order selection}
\acro{M-PSK}[${M}$-PSK]{$M$-ary phase shift keying}                       
\acro{M-APSK}[${M}$-PSK]{$M$-ary asymmetric PSK} 

\acro{M-QAM}[$M$-QAM]{$M$-ary quadrature amplitude modulation}
\acro{MRC}{maximal ratio combiner}                  
\acro{maxs}[MSO]{maximum sparsity order}                                      
\acro{M2M}{Machine-to-Machine}                                                
\acro{MUI}{multi-user interference}
\acro{mMTC}{massive Machine Type Communications}      
\acro{mm-Wave}{millimeter-wave}
\acro{MP}{mobile phone}
\acro{MPE}{maximum permissible exposure}
\acro{MAC}{media access control}
\acro{NB}{narrowband}
\acro{NBI}{narrowband interference}
\acro{NLA}{nonlinear sparse approximation}
\acro{NLOS}{Non-Line of Sight}
\acro{NTIA}{National Telecommunications and Information Administration}
\acro{NTP}{National Toxicology Program}
\acro{NHS}{National Health Service}

\acro{OC}{optimum combining}                             
\acro{OC}{optimum combining}
\acro{ODE}{operational distortion-energy}
\acro{ODR}{operational distortion-rate}
\acro{OFDM}{Orthogonal Frequency-Division Multiplexing}
\acro{omp}[OMP]{orthogonal matching pursuit}
\acro{OSMP}[OSMP]{orthogonal subspace matching pursuit}
\acro{OQAM}{offset quadrature amplitude modulation}
\acro{OQPSK}{offset QPSK}
\acro{OFDMA}{Orthogonal Frequency-division Multiple Access}
\acro{OPEX}{Operating Expenditures}
\acro{OQPSK/PM}{OQPSK with phase modulation}

\acro{PAM}{pulse amplitude modulation}
\acro{PAR}{peak-to-average ratio}
\acrodef{pdf}[PDF]{probability density function}                      
\acro{PDF}{probability density function}
\acrodef{p.d.f.}[PDF]{probability distribution function}
\acro{PDP}{power dispersion profile}

\acro{PMF}{probability mass function}                             
\acrodef{p.m.f.}[PMF]{probability mass function}
\acro{PN}{pseudo-noise}
\acro{PPM}{pulse position modulation}
\acro{PRake}{Partial Rake}
\acro{PSD}{power spectral density}
\acro{PSEP}{pairwise synchronization error probability}
\acro{PSK}{phase shift keying}
\acro{PD}{power density}
\acro{8-PSK}[$8$-PSK]{$8$-phase shift keying}

\acro{FSK}{Frequency Shift Keying}

\acro{QAM}{Quadrature Amplitude Modulation}
\acro{QPSK}{Quadrature Phase Shift Keying}
\acro{OQPSK/PM}{OQPSK with phase modulator }

\acro{RD}[RD]{raw data}
\acro{RDL}{"random data limit"}
\acro{ric}[RIC]{restricted isometry constant}
\acro{rict}[RICt]{restricted isometry constant threshold}
\acro{rip}[RIP]{restricted isometry property}
\acro{ROC}{receiver operating characteristic}
\acro{rq}[RQ]{Raleigh quotient}
\acro{RS}[RS]{Reed-Solomon}
\acro{RSC}[RSSC]{RS based source coding}
\acro{r.v.}{random variable}                               
\acro{R.V.}{random vector}
\acro{RMS}{root mean square}
\acro{RFR}{radiofrequency radiation}
\acro{RIS}{Reconfigurable Intelligent Surface}
\acro{RNA}{RiboNucleic Acid}

\acro{SA}[SA-Music]{subspace-augmented MUSIC with OSMP}
\acro{SCBSES}[SCBSES]{Source Compression Based Syndrome Encoding Scheme}
\acro{SCM}{sample covariance matrix}
\acro{SEP}{symbol error probability}
\acro{SG}[SG]{sparse-land Gaussian model}
\acro{SIMO}{single-input multiple-output}
\acro{SINR}{signal-to-interference plus noise ratio}
\acro{SIR}{signal-to-interference ratio}
\acro{SISO}{Single-Input Single-Output}
\acro{SMV}{single measurement vector}
\acro{SNR}[\textrm{SNR}]{signal-to-noise ratio} 
\acro{sp}[SP]{subspace pursuit}
\acro{SS}{spread spectrum}
\acro{SW}{sync word}
\acro{SAR}{specific absorption rate}
\acro{SSB}{synchronization signal block}

\acro{TH}{time-hopping}
\acro{ToA}{time-of-arrival}
\acro{TR}{transmitted-reference}
\acro{TW}{Tracy-Widom}
\acro{TWDT}{TW Distribution Tail}
\acro{TCM}{trellis coded modulation}
\acro{TDD}{Time-Division Duplexing}
\acro{TDMA}{Time Division Multiple Access}
\acro{TMC}{Torrino MezzoCammino}

\acro{UAV}{Unmanned Aerial Vehicle}
\acro{uric}[URIC]{upper restricted isometry constant}
\acro{urict}[URICt]{upper restricted isometry constant threshold}
\acro{UWB}{ultrawide band}
\acro{UWBcap}[UWB]{Ultrawide band}   
\acro{URLLC}{Ultra Reliable Low Latency Communications}
         
\acro{wuric}[UWRIC]{upper weak restricted isometry constant}
\acro{wurict}[UWRICt]{UWRIC threshold}                
\acro{UE}{User Equipment}
\acro{UL}{UpLink}

\acro{WiM}[WiM]{weigh-in-motion}
\acro{WLAN}{wireless local area network}
\acro{wm}[WM]{Wishart matrix}                               
\acroplural{wm}[WM]{Wishart matrices}
\acro{WMAN}{wireless metropolitan area network}
\acro{WPAN}{wireless personal area network}
\acro{wric}[WRIC]{weak restricted isometry constant}
\acro{wrict}[WRICt]{weak restricted isometry constant thresholds}
\acro{wrip}[WRIP]{weak restricted isometry property}
\acro{WSN}{wireless sensor network}                        
\acro{WSS}{Wide-Sense Stationary}
\acro{WHO}{World Health Organization}
\acro{Wi-Fi}{Wireless Fidelity}

\acro{sss}[SpaSoSEnc]{sparse source syndrome encoding}

\acro{VLC}{Visible Light Communication}
\acro{VPN}{Virtual Private Network} 
\acro{RF}{Radio Frequency}
\acro{FSO}{Free Space Optics}
\acro{IoST}{Internet of Space Things}

\acro{GSM}{Global System for Mobile Communications}
\acro{2G}{Second-generation cellular network}
\acro{3G}{Third-generation cellular network}
\acro{4G}{Fourth-generation cellular network}
\acro{5G}{Fifth-generation cellular network}	
\acro{gNB}{next-generation Node-B}
\acro{NR}{New Radio}
\acro{UMTS}{Universal Mobile Telecommunications Service}
\acro{LTE}{Long Term Evolution}

\acro{QoS}{Quality of Service}
\end{acronym}

\theoremstyle{remark}
\theoremstyle{definition}
\newtheorem{defin}{Definition}
\newtheorem{assum}{Assumption}
\newtheorem{rem}{Remark}
\newtheorem{ins}{Insight}

\title{5G Network Planning under Service and EMF Constraints: Formulation and Solutions}
\author{Luca Chiaraviglio$^{(1,2)}$, Cristian Di Paolo,$^{(2)}$, Nicola Blefari-Melazzi$^{(1,2)}$\\
1) Department of Electronic Engineering, University of Rome Tor Vergata, Rome, Italy,\\ email: luca.chiaraviglio@uniroma2.it, crisdp95@gmail.com, blefari@uniroma2.it\\
2) Consorzio Nazionale Interuniversitario per le Telecomunicazioni (CNIT), Italy}

\IEEEcompsoctitleabstractindextext{
\begin{abstract}
We target the planning of a 5G cellular network under 5G service and ElectroMagnetic Fields (EMFs) constraints. We initially model the problem with a Mixed Integer Linear Programming (MILP) formulation. The pursued objective is a weighed function of  {next-generation Node-B} (gNB) installation costs and 5G service coverage level  {from a massive Multiple Input Multiple Output (MIMO) system}. In addition, we precisely model restrictive EMF constraints and we integrate scaling parameters to estimate the power radiated by 5G gNBs. Since the considered planning problem is NP-Hard, and therefore very challenging to be solved even for small problem instances, we design an efficient heuristic, called \textsc{PLATEA}, to practically solve it. Results, obtained over a realistic scenario that includes EMF exposure from pre-5G technologies (e.g., 2G, 3G, 4G), prove that  {the cellular planning selected by \textsc{PLATEA}} ensures 5G service and restrictive EMF constraints. However, we demonstrate that the results are strongly affected by: \textit{i}) the relative weight between gNB installation costs and 5G service coverage level, \textit{ii}) the scaling parameters to estimate the exposure generated by 5G gNBs, \textit{iii}) the amount of exposure from pre-5G technologies  {and \textit{iv}) the adopted frequency reuse scheme}.  
\end{abstract}

\begin{IEEEkeywords}
5G Mobile Networks, 5G Network Planning, Base Station Deployment, Service and EMF constraints, EMF regulations
\end{IEEEkeywords}}

\maketitle

\IEEEdisplaynontitleabstractindextext

\IEEEpeerreviewmaketitle

\IEEEraisesectionheading{\section{Introduction}}
The provisioning of the 5G service inevitably requires the installation of new 5G equipment, called \ac{gNB}, over the territory. The task of selecting and configuring the set of sites hosting 5G equipment is often referred as 5G \text{cellular planning} \cite{chiaraviglio2018planning}, a complex problem that involves costs, service coverage and \ac{EMF} constraints. In general, the planning of a cellular network is a critical step that has a huge impact on the \ac{CAPEX} costs incurred by the operator \cite{OugKatEne:19}, as well as on the \ac{QoS} perceived by users \cite{amaldi2003planning,mishra2004fundamentals}. From an operator perspective, network planning should minimize the costs for deploying new 5G sites and installing 5G equipment. In addition, the operator aims at maximizing the performance (e.g., throughput, delay) that is experienced by 5G \ac{UE}. 

 {Apart from economic and service goals, another important aspect that should be considered during the planning phase is the \ac{EMF} exposure from 5G \acp{gNB}. More in depth, the radiation from 5G \acp{gNB} is a matter of strong debates among the population, including a supposed increase of exposure w.r.t. legacy technologies, as well as possible health effects associated with adoption of a wide range of frequencies (including mm-Waves). In this context, a clear causal-correlation between exposure levels from 5G \acp{gNB} operating below the limits defined by laws and emergence of health diseases has not been scientifically proven so far} \cite{chiaraviglio2020health} {. Therefore, it is of utmost importance to plan the network by ensuring \ac{EMF} exposure constraints.}

Under realistic settings, the planning problem is strongly affected by the regulations governing the \ac{EMF} levels radiated by 5G equipment \cite{chiaraviglio2018not}. In general, many countries in the world ensure that the \ac{EMF} levels radiated by 5G \acp{gNB} are lower than maximum values (often referred as \ac{EMF} limits) \cite{madjar2016human}, which depend on the frequency exploited by the 5G \ac{gNB}. Traditionally, international/federal bodies like \ac{ICNIRP} and \ac{FCC} define \ac{EMF} limits for all the cellular frequencies, including the ones used by 5G equipment \cite{icnirp,FCC:97}.  
 As a result, the constraints introduced by \ac{EMF} regulations have to be carefully taken into account during the installation and then the operation of 5G equipment. Intuitively, the \ac{EMF} constraints tend to limit the number of 5G sites installed over the territory and/or the amount of radiated power by each 5G \ac{gNB}. Therefore, the \ac{EMF} regulations have a large impact on the 5G \ac{gNB} installation costs and the 5G service received by \ac{UE} \cite{itutksupplement14,chiaraviglio2018planning,chiaraviglio2018not}.

The picture is further complicated in different countries (such as Italy) \cite{madjar2016human}, which introduce \ac{EMF} regulations more restrictive that the ones defined by \ac{ICNIRP} and \ac{FCC}, on the basis of the application of a precautionary principle, in order to preserve the population from (still unknown) long-term health effects triggered by \ac{EMF} exposure.  {Such restrictive} rules include \ac{EMF} limits strongly lower that the ones defined by \ac{ICNIRP}/\ac{FCC} \cite{madjar2016human,italianreg}.  {For example, \ac{ICNIRP} regulations allow a maximum \ac{EMF} exposure up to 61~[V/m] for frequencies above 2~[GHz], while this limit is decreased to only 6~[V/m] in the Italian \ac{EMF} regulations that are applied in residential areas. In addition, other stringent rules include} minimum distances between sensitive places (e.g., schools, hospitals public parks) and the installed 5G sites \cite{romereg}.  {For example, a minimum distance of 100~[m] is applied} in the city of Rome, an area of 1287 square kilometers inhabited by almost 3 million people. 

As sketched in Fig.~\ref{fig:sensitive_areas_and_saturation_levels}, the introduction of strict \ac{EMF} limits strongly affects  {5G network planning}. For example, the installation of new 5G sites is prevented within a minimum distance from the center of the sensitive area (e.g., the school in the figure). In addition, the enforcement of very low \ac{EMF} limits tends to generate \ac{EMF} saturation areas (e.g., the one shown on top right of the figure), where the \ac{EMF} levels from pre-5G technologies are already close to the maximum limits. In such zones, therefore, the installation of new 5G sites is denied. As a result, the 5G sites have to be installed in other locations  {(i.e., not the optimal ones)}, thus further impacting the installation costs and the service level offered by the 5G network.

\begin{figure}[t]
	\centering
	\includegraphics[width=8cm]{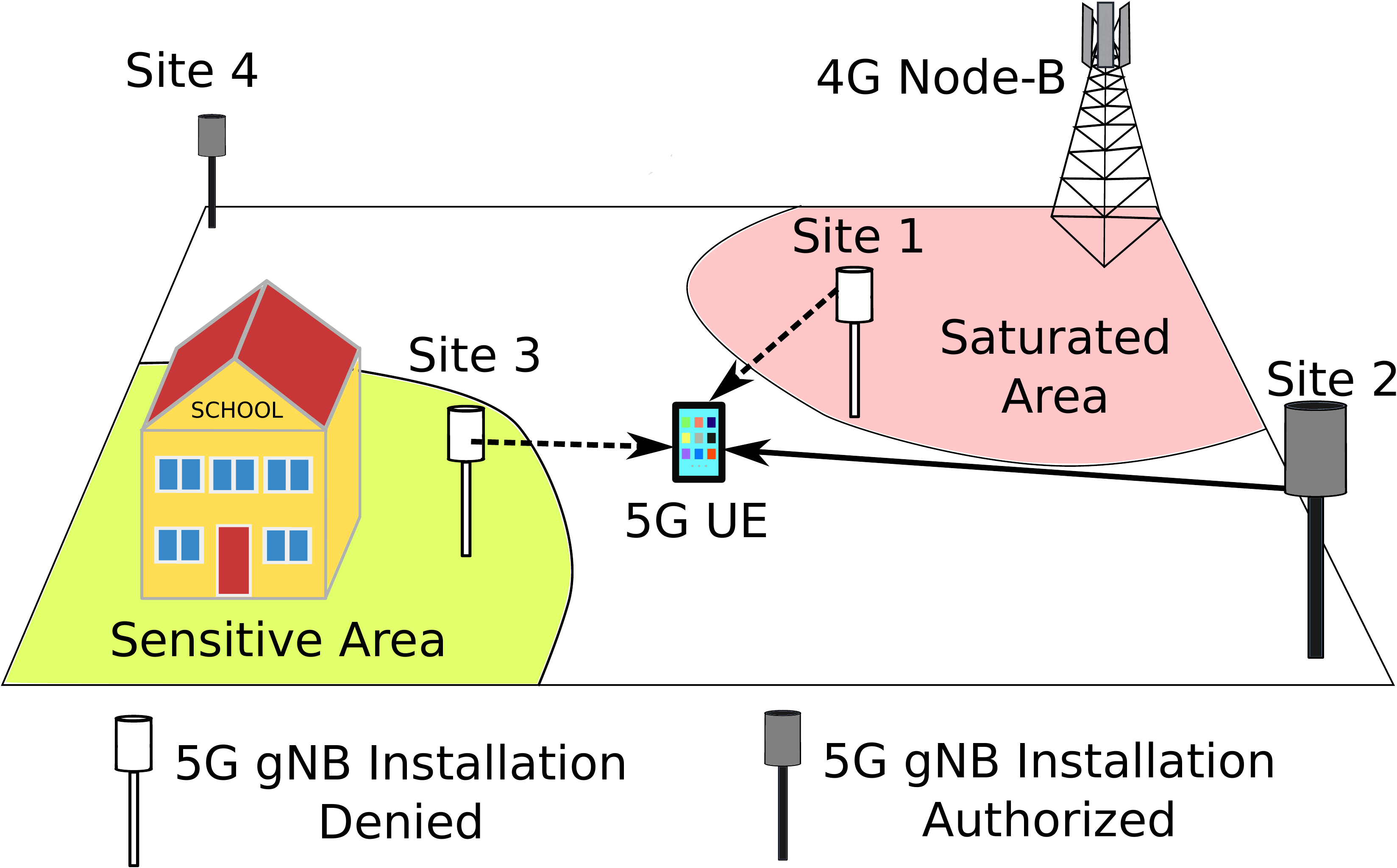}
	\caption{The presence of sensitive areas and \ac{EMF} saturation zones heavily impacts 5G cellular planning.}
         \label{fig:sensitive_areas_and_saturation_levels}
\end{figure}

In this context, a natural question emerges: Is it possible to deploy a heterogeneous 5G network, while ensuring 5G service and restrictive \ac{EMF} constraints? The ambitions goal of the paper is to tackle such interesting - and challenging - problem. Our innovative contributions can be summarized as follows. First, we take into account restrictive \ac{EMF} regulations affecting the 5G planning phase, i.e., \ac{EMF} limits stricter than \ac{ICNIRP}/\ac{FCC} ones, as well as enforcement of minimum distances between 5G sites and sensitive places. Second, we optimally model the 5G planning problem under service and \ac{EMF} constraints {, by assuming a 5G system based on massive \ac{MIMO} functionalities.} The presented problem integrates the widely used model of Marzetta \cite{marzetta} to compute the service level from a massive \ac{MIMO} system, as well as the EMF point source described in \ac{ITU} K.70 recommendation \cite{itutk70}, which is enriched by a set of scaling parameters to take into account temporal and statistical variations of radiated power from 5G \ac{MIMO} systems with beamforming capabilities. We also show that the complete formulation falls within the class of \ac{MILP} problems and it is NP-Hard. Third, we design \textsc{PLanning Algorithm Towards EMF Emissions Assessment (PLATEA)}, a novel heuristic that is able to efficiently solve the 5G planning problem while ensuring adherence to strict \ac{EMF} limits and 5G service for the set of pixels belonging to the area under consideration. Fourth, we evaluate \textsc{PLATEA} and two reference algorithms in a realistic scenario, whose parameters have been measured on the field (e.g., the \ac{EMF} levels radiated by pre-5G sites already deployed in the scenario). Results prove that \textsc{PLATEA} outperforms the reference algorithms, by efficiently balancing between 5G gNB installation costs and amount of 5G service coverage. In addition, we demonstrate that the scaling parameters used to compute the power radiated by 5G \acp{gNB} play a critical role in determining the selected planning and the \ac{EMF} levels over the territory. Eventually, we show that the exposure levels generated by pre-5G technologies  {and the adopted frequency reuse factor} have an impact on 5G planning.

To the best of our knowledge, previous works in the literature are focused on orthogonal aspects w.r.t the ones investigated in this paper. For example, Oughton \textit{et al.} \cite{OugKatEne:19} target the solution of the 5G planning problem by means of techno-economic approaches, with little emphasis on the impact of \ac{EMF} constraints. On the other hand, Matalatala \textit{et al.} \cite{MatDerTanJoseph:18,MatMarSer:19} design heuristics targeting the reduction of the radiated power, \ac{EMF} and/or \ac{SAR}, without considering: \textit{i)} the linearization of the problem constraints,\footnote{The authors of \cite{MatMarSer:19} introduce an optimal formulation, which is however not linear w.r.t. the \ac{SNR} computation, the electric field computation and the \ac{SAR} computation.}  \textit{ii)} the impact of the variation of the scaling parameters to compute the \ac{EMF} levels from 5G \acp{gNB}, \textit{iii)} the introduction of constraints to ensure a minimum distance between sensitive places and 5G gNBs, \textit{iv)}  {the \ac{EMF} and throughput evaluation over the whole territory (i.e., not only on single users)}. In this work, we show that both \textit{ii)} and \textit{iii)} are fundamental to determine the actual planning. In addition, we  {tackle both \textit{i}) and \textit{iv}) by adopting a pixel tessellation and a set of linearized constraints, which are included in an innovative formulation and a new algorithm}.

The rest of the paper is organized as follows. Sec.~\ref{sec:rel_works} reviews the related work. The main building blocks of the considered 5G framework are highlighted in Sec.~\ref{sec:building_blocks}. Sec.~\ref{sec:planning} reports the problem formulation. The \textsc{PLATEA} algorithm is thoroughly described in Sec.~\ref{sec:algorithm}. The scenario under consideration is detailed in Sec.~\ref{sec:scenario}. Results are analyzed in Sec.~\ref{sec:results}. Finally, Sec.~\ref{sec:conclusions} concludes our work.

\begin{table*}[t]
    \caption{Work positioning w.r.t. the related literature.}
    \label{tab:positioning}
    \scriptsize
    \centering
    \begin{tabular}{|p{0.25cm}|c|p{1.5cm}|c|c|c|c|c|}
\hline
\rowcolor{Coral}  & & \textbf{5G Equip.} & \textbf{5G Service} & & & & \\[-0.05em]
\rowcolor{Coral} \multirow{-2}{*}{\textbf{Ref}} & \multirow{-2}{*}{\textbf{Goal}} & \textbf{Type} & \textbf{Metrics} & \multirow{-2}{*}{\textbf{EMF Features}} & \multirow{-2}{*}{\textbf{EMF Regulations}} & \multirow{-2}{*}{\textbf{Methodology}} & \multirow{-2}{*}{\textbf{Scenario Complexity}} \\
\hline

\hline
 & & & & & & & \\[-0.8em]
\cite{OugKatEne:19}  & \begin{minipage}{2cm}Capacity and costs assessment\end{minipage}  & \begin{minipage}{1.5cm}Micro \& macro 5G \acp{gNB}\end{minipage} & \begin{minipage}{2cm}\ac{SINR}, network spectral efficiency\end{minipage} & - & - & \begin{minipage}{2cm}Model assessment\end{minipage}& \begin{minipage}{3cm}Hexagonal coverage layouts with 7 candidate sites, service area of few square kilometers, evaluation done on pixels.\end{minipage} \\[-0.8em]
 & & & & & & & \\
\hline
\rowcolor{Linen} & & & & & & & \\[-0.8em]
\rowcolor{Linen} \cite{MatDerTanJoseph:18}  & \begin{minipage}{2cm}Power consumption reduction, \ac{EMF} exposure reduction\end{minipage} & \begin{minipage}{1.5cm}Generic \acp{gNB} operating at 3.7~[GHz]\end{minipage} & \begin{minipage}{2cm}Throughput, \ac{SINR}\end{minipage} & \begin{minipage}{2cm}Presence of exclusion zones\end{minipage}  & \begin{minipage}{2cm}Strict \ac{EMF} limits\end{minipage} & \begin{minipage}{2cm}Heuristic\end{minipage} & \begin{minipage}{3cm}Dozens of candidate sites with irregular coverage layout, service area of several square kilometers, evaluation done on users (not on a pixel base).\end{minipage} \\ 
\hline
 & & & & & & & \\[-0.8em]
 \cite{MatMarSer:19} & \begin{minipage}{2cm}Power consumption reduction, \ac{EMF} exposure reduction, \ac{SAR} exposure reduction, dose exposure reduction \end{minipage}  & \begin{minipage}{1.5cm}Generic \acp{gNB} operating at 3.7~[GHz]\end{minipage} & \begin{minipage}{2cm}Throughput, \ac{SINR}\end{minipage} & \begin{minipage}{2cm}Statistical models (with fixed parameters), presence of exclusion zones\end{minipage} & \begin{minipage}{2cm}\ac{ICNIRP}-based \ac{EMF} limits\end{minipage} & \begin{minipage}{2cm}\ac{MINLP} optimization model, heuristic\end{minipage} &  \begin{minipage}{3cm}Dozens of candidate sites with irregular coverage layout, service area of several square kilometers, evaluation done on users (not on a pixel base).\end{minipage} \\[-0.8em] 
 & & & & & & & \\
\hline
\rowcolor{Linen}  & & & & & & & \\[-0.8em]
\rowcolor{Linen} \multirow{-2}{*}{\begin{sideways}This work\end{sideways}} & \begin{minipage}{2cm}Installation costs reduction, maximization of the number of served pixels\end{minipage} & \begin{minipage}{1.5cm}Micro \& macro 5G \acp{gNB}\end{minipage} & \begin{minipage}{2cm}Throughput, minimum \ac{SIR} threshold, maximum coverage distance\end{minipage} & \begin{minipage}{2cm} {Temporal} and statistical models (with variation of parameters), presence of exclusion zones\end{minipage} & \begin{minipage}{2cm}Strict \ac{EMF} limits, minimum site distance from sensitive places\end{minipage} & \begin{minipage}{2cm}\ac{MILP} optimization model, heuristic\end{minipage} & \begin{minipage}{3cm}Dozens of candidate sites with irregular coverage layout, service area of different square kilometers, evaluation done on pixels.\end{minipage}\\
\hline
\end{tabular}
\end{table*}

\section{Related Works}
\label{sec:rel_works}

 {In the literature, different works investigate the cellular network planning problem under various angles, including e.g., application of Artificial Intelligence (AI)-based tools} \cite{perez2016knowledge} {, identification of opportunities and challenges} \cite{taufique2017planning} {, optimal design response in smart grids} \cite{saxena2017efficient} {, service-based network dimensioning} \cite{khan2020service} {. Despite we recognize the importance of such previous works, none of them investigate the impact of \ac{EMF} regulations on the planning.}

Tab~\ref{tab:positioning} reports the positioning of this paper w.r.t. \cite{OugKatEne:19,MatDerTanJoseph:18,MatMarSer:19} {, which we believe are the closest contributions to our work}. More in depth, we consider the following features to classify \cite{OugKatEne:19,MatDerTanJoseph:18,MatMarSer:19}  {and our work}: \textit{i)} pursued goal(s) (e.g., cost reduction, power consumption reduction), \textit{ii)} targeted 5G equipment type (e.g., generic \acp{gNB}, micro and macro \acp{gNB}), \textit{iii)} modeled 5G service (e.g., \ac{SINR}, network spectral efficiency, throughput, maximum coverage distance), \textit{iv)} \ac{EMF} features (e.g., temporal and statistical models, presence of exclusion zones in proximity to the \ac{gNB}), \textit{v)} considered \ac{EMF} regulations (e.g., \ac{ICNIRP}-based \ac{EMF} limits, strict \ac{EMF} limits, minimum site distance from sensitive places), \textit{vi)} pursued methodology (e.g., model assessment, optimal formulation, heuristic) and \textit{vii)} scenario complexity (e.g., number of candidate sites, regular or irregular coverage layout, size of the service area, pixel or user evaluation). 

Compared to \cite{OugKatEne:19,MatDerTanJoseph:18,MatMarSer:19}, our work moves one step further by: \textit{i)} explicitly targeting a weighed function of installation costs and service coverage, \textit{ii)} considering a heterogeneous 5G network composed of micro and macro \acp{gNB},  {in order to control the service level provided by micro \acp{gNB} and the one by macro \acp{gNB}
}, \textit{iii)} precisely modeling multiple service metrics, including throughput, minimum \ac{SIR} and maximum coverage distance, \textit{iv)} performing the variation of the scaling parameters to compute the \ac{EMF}, as well as including exclusion zones in proximity to the \acp{gNB}, \textit{v)} integrating \ac{EMF} regulations more restrictive than \ac{ICNIRP}/\ac{FCC}, both in terms of maximum limits and in terms of minimum distance between a 5G site and a sensitive place, \textit{vi)} defining a linear formulation (\ac{MILP}) and exploiting the linearized constraints to design the heuristic,  \textit{vii)} analyzing a large scenario composed of dozens of candidate sites with irregular coverage layouts, a service area in the order of different square kilometers, and service/\ac{EMF} evaluations performed in each pixel of the territory.

\section{Building Blocks}
\label{sec:building_blocks}

 {The goal of the planning problem considered in this work is to install a 5G network of Non Standalone type for a single operator, in order to provide an \ac{eMBB} 5G service} \cite{3gppservice}.  {More specifically, the operator can install \acp{gNB} operating up to the 6~[GHz] band, which is the currently available option for implementing 5G in many countries in the world (including Italy).}  {Basic coverage is provided by macro \acp{gNB} operating on sub-GHz frequencies, while micro \acp{gNB} operate on mid-band frequencies to provide hot-spot capacity.} \footnote{ {5G frequencies also include mm-Waves, which are however not being used in Italy at the time of preparing our study. Actually, the \acp{gNB} under installation include, in fact, the same equipment types assumed in this work (i.e., operating on sub-Ghz and mid-band). Clearly, the adoption of equipment operating on mm-Waves may require service and/or \ac{EMF} models potentially different than those ones employed in our work. This aspect is left as a future research activity, in parallel with the deployment of mm-Waves \acp{gNB}.}}   {Clearly, micro and macro \acp{gNB} operate on different portions of the spectrum, which are also separated by those ones used by previous technologies and/or other operators from the same country.}\footnote{ {The investigation of spectrum sharing approaches, including e.g., carrier aggregation, channel bonding, licensed shared access, licensed assisted access, and/or location-based licensing} \cite{patwary2020potential} {, is left for future work.}}

 {We then present the main building blocks} that are integrated in our 5G framework, namely: \textit{i)} the model to assess 5G performance, \textit{ii}) the model to estimate \ac{EMF} radiated by a set of \acp{gNB} and \textit{iii}) the \ac{EMF} regulations for the installation of 5G sites. In the following, we provide more details about each building block.

\subsection{5G Performance Model}

We adopt the widely known \ac{MIMO} model of Marzetta \cite{marzetta} to evaluate the 5G performance for a set of installed \acp{gNB}. We refer to \cite{marzetta} for the details, while we report  {below} the salient features. In brief, the model assumes that each site is equipped with arrays composed of a very large number of antenna elements.  {Moreover, the same power is radiated
by all arrays operating at the same frequency.}\footnote{ {In case the arrays operate with multiple power levels (e.g., due to different power budget settings), a different model has to be used. We leave the investigation of this aspect as future work.}} In the work of Marzetta \cite{marzetta}, the number of antennas is higher than the number of users. Since the number of antennas is very large, the downlink \ac{SINR} is dominated by the interference from neighboring \acp{gNB} rather than by the noise floor. More formally, the \ac{SIR} of the $k$-th user served by $l$-th \ac{gNB} operating on frequency $f$ is defined as:\footnote{In the original model of \cite{marzetta} a set of \acp{gNB} operating at the same central frequency is assumed. In this work, instead, we consider a heterogeneous network composed of multiple tiers of \acp{gNB} operating at different frequencies. However, the extension of the model of Marzetta  \cite{marzetta} to the multiple frequencies case is straightforward, as only \acp{gNB} operating on the same frequency have to be counted in the \ac{SIR} computation of Eq.~(\ref{eq:sir}).  {Moreover, when spectrum-based approaches are introduced, the denominator has to include the interference terms from \acp{gNB} that compete on the same spectrum used by $f$.}}

\begin{equation}
\label{eq:sir}
\text{S}_{(k,l,f)} = \frac{\beta_{(l,k,l,f)}^2}{\sum_{j \neq l}{\beta_{(l,k,j,f)}^2}}
\end{equation}

In the previous equation, the $\beta$ terms are expressed as: 

\begin{equation}
\beta_{(l,k,j,f)} = \frac{z_{(l,k,j,f)}}{D_{(k,l,f)}^{\gamma_f}}
\label{eq:beta}
\end{equation}

where $D_{(k,l,f)}$ is the distance between 5G \ac{UE} $k$ and \ac{gNB} operating on frequency $f$ and installed at location $l$, $\gamma_f$ is the path-loss exponent for frequency $f$ and $z_{(l,k,j,f)}$ is a log-normal random variable  {modeling the shadow fading component}, i.e., the quantity $10 \cdot \text{log}_{10}(z_{(l,k,j,f)})$  {(i.e., the shadow fading term in dB scale)} is a distributed zero-mean Gaussian with a standard deviation $\sigma^{\text{SHAD}}_f$ \cite{marzetta}. 

\begin{figure}[t]
	\centering
	\includegraphics[width=7cm]{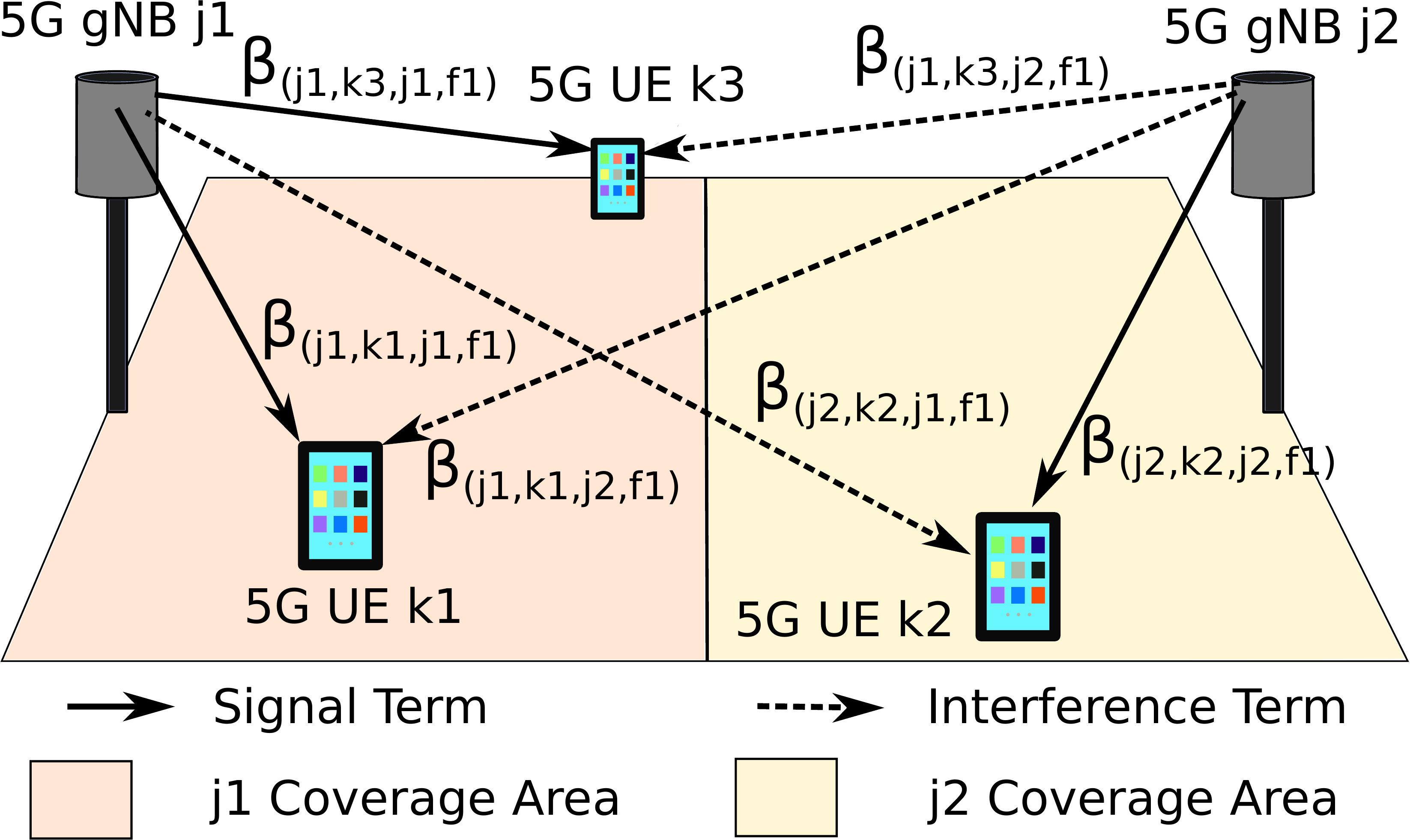}
	\caption{Signal and interference terms in the 5G service model of \cite{marzetta} for a toy case scenario with two \acp{gNB} and three \acp{UE}.}
	\label{fig:signal_interference_terms}
\end{figure}

The $\beta$ terms appearing in Eq.~(\ref{eq:sir}) are also sketched in the toy-case scenario of Fig.~\ref{fig:signal_interference_terms}, which is composed by two \acp{gNB} and three \ac{UE}. Intuitively, each 5G UE is served by a single \ac{gNB}, while the other \ac{gNB} contributes to the interference experienced by the 5G UE. It is important to remark that, in the downlink direction, the contributions of interference are solely due to the neighboring \ac{gNB}, and not to the simultaneous transmissions to other \ac{UE} in the same cell (e.g., terminals $k1$ and $k3$ in the figure), due to the fact that an \ac{OFDM} technology is assumed.

The downlink throughput received by user $k$ from \ac{gNB} installed at location $l$ and operating on frequency $f$ is then expressed as:
\begin{equation}
\label{eq:capacity}
T_{(k,l,f)}  =  \frac{B_f \cdot \Gamma_f}{\epsilon^{\text{F-REUSE}}_f}  \cdot   \log_2 \left(1 + \text{S}_{(k,l,f)} \right)  
\end{equation}
where: $B_f$ is the \ac{gNB} bandwidth, $\epsilon^{\text{F-REUSE}}_f\leq1$ is a parameter  {governing the frequency reuse factor} over frequency $f$ (equal to 1 when  {a unity frequency reuse is adopted over frequency $f$}), $\Gamma_f$ is a shaping factor, formally expressed as: 
\begin{equation}
\label{eq:gamma_term}
\Gamma_f= \frac{(\tau^{\text{SLOT}}_f - \tau^{\text{PILOT}}_f) \cdot \tau^{\text{USEFUL}}_f}{\tau^{\text{SLOT}}_f \cdot \tau^{\text{SYMBOL}}_f}
\end{equation}
where $\tau^{\text{SLOT}}_f$ is the slot duration over $f$, $\tau^{\text{PILOT}}_f$ is the pilot duration over $f$, $\tau^{\text{SYMBOL}}_f$ is the symbol interval over $f$ and $\tau^{\text{USEFUL}}_f$ is the useful symbol duration over $f$. The expressions for $\tau^{\text{SLOT}}_f$, $\tau^{\text{PILOT}}_f$, $\tau^{\text{SYMBOL}}_f$ and $\tau^{\text{USEFUL}}_f$ are reported in Tab.~\ref{tab:marzetta_parameters}, where $N^{\text{OFDM}}_f$ is the number of \ac{OFDM} symbols over $f$, $\Delta_f$ is the subcarrier spacing over $f$, $\delta_f$ is the cyclic prefix duration over $f$, $N^{\text{OFDM-PILOT}}_f$ is the number of \ac{OFDM} symbols used for pilots over $f$ and $\tau^{\text{COHERENCE}}_f$ is the coherence time over $f$.

\renewcommand*{\arraystretch}{1}
\begin{table}[t]
    \caption{Expressions for the time-related parameters of \cite{marzetta}.}
    \label{tab:marzetta_parameters}
    \scriptsize
    \centering
    \begin{tabular}{|c|c|}
\hline
 \rowcolor{Coral} \textbf{Parameter} & \textbf{Expression}\\
\hline
& \\[-0.75em]
    $\tau^{\text{SLOT}}_f$ & $\left[N^{\text{OFDM}}_f\cdot(\frac{1}{\Delta_f})\right] + \delta_f$\\[0.6em]
\hline
 \rowcolor{Linen}& \\[-0.75em]
 \rowcolor{Linen}   $\tau^{\text{PILOT}}_f$ & $N^{\text{OFDM-PILOT}}_f\cdot \tau^{\text{SYMBOL}}_f$\\[0.6em]
\hline
& \\[-0.75em]
    $\tau^{\text{SYMBOL}}_f$ & $\frac{\tau^{\text{COHERENCE}}_f}{N^{\text{OFDM}}_f}$\\[0.6em]
\hline
 \rowcolor{Linen} & \\[-0.75em]
 \rowcolor{Linen}   $\tau^{\text{USEFUL}}_f$ & $\frac{1}{\Delta_f}$\\[0.6em]
\hline
    \end{tabular}
\end{table}
\renewcommand*{\arraystretch}{1}

Summarizing, the model of Marzetta \cite{marzetta} allows to compute the \ac{SIR} and the maximum throughput for each \ac{UE}, given the set of installed 5G \acp{gNB} and the \ac{UE}-\ac{gNB} association. 

\subsection{5G EMF Model}

The second building block that is instrumental to our framework is the computation of the \ac{EMF}  that is received by a 5G \ac{UE} from a 5G \ac{gNB}. In the literature, different  \ac{EMF} models have been proposed to this purpose. We refer the interested reader to ITU-T K.70 recommendation \cite{itutk70} for an overview about the \ac{EMF} models. In brief, the available options include point source models, synthetic models and full-wave models. In this work, we select the point source model, due to the following key properties: \textit{i}) the actual \ac{EMF} levels that are measured over the territory in the far-field region are typically lower than the ones estimated through the point source model of \cite{itutk70}. Therefore, when the \ac{EMF} is computed through this model and the obtained level is below the maximum limit, the adherence to the limit is always guaranteed; 
\textit{ii}) a linear set of constraints to compute the total \ac{EMF} levels can be built when the model is integrated in our framework. Ensuring the linearity of the constraints is a desirable property, which, in fact, allows to reduce the complexity of both the optimal formulation and the designed heuristic.  {We also refer the interested reader to Appendix~\ref{app:point-source} for more considerations about the point source model.}

 {More formally}, the point source model allows to compute the power density $P_{(k,l,f)}$ that is received by UE $k$ from a \ac{gNB} located at site $l$ and operating on frequency $f$.\footnote{The power density metric is commonly used to characterize the level of exposure. Other exposure metrics include electric field, magnetic field and \ac{SAR}. We refer the interested reader to \cite{icnirp} for an overview and comparison of the different exposure metrics.} Clearly, the distance $D_{(k,l,f)}$ between \ac{gNB} $l$ and \ac{UE} $k$ is assumed to be in the far-field region \cite{itutk70}. More formally, $P_{(k,l,f)}$ is expressed as:
\begin{equation}
\label{eq:power_density_original}
P_{(k,l,f)} = \frac{\text{EIRP}_{(l,f)}}{4 \pi \cdot D_{(k,l,f)}^2} \cdot F_{(k,l,f)}
\end{equation}
where $\text{EIRP}_{(l,f)}$ is the \ac{EIRP} from \ac{gNB} operating on frequency $f$ and located at site $l$ and $F_{(k,l,f)}\leq 1$ is the normalized numeric gain over user $k$ from an antenna installed at location $l$ operating on frequency $f$. More in depth, $\text{EIRP}_{(l,f)}$ is formally expressed as:
\begin{equation}
\label{eq:eirp}
\text{EIRP}_{(l,f)} = O^{\text{MAX}}_{(l,f)} \cdot \frac{\eta^{\text{GAIN}}_{f}}{\eta^{\text{LOSS}}_{f}}
\end{equation}
where $O^{\text{MAX}}_{(l,f)}$ is the maximum output power for a \ac{gNB} operating on frequency $f$ and located at site $l$, $\eta^{\text{GAIN}}_{f}$ is the transmission gain on frequency $f$, and $\eta^{\text{LOSS}}_{f}$ is the transmission loss on frequency $f$.

\begin{figure}[t]
	\centering
	\includegraphics[width=7cm]{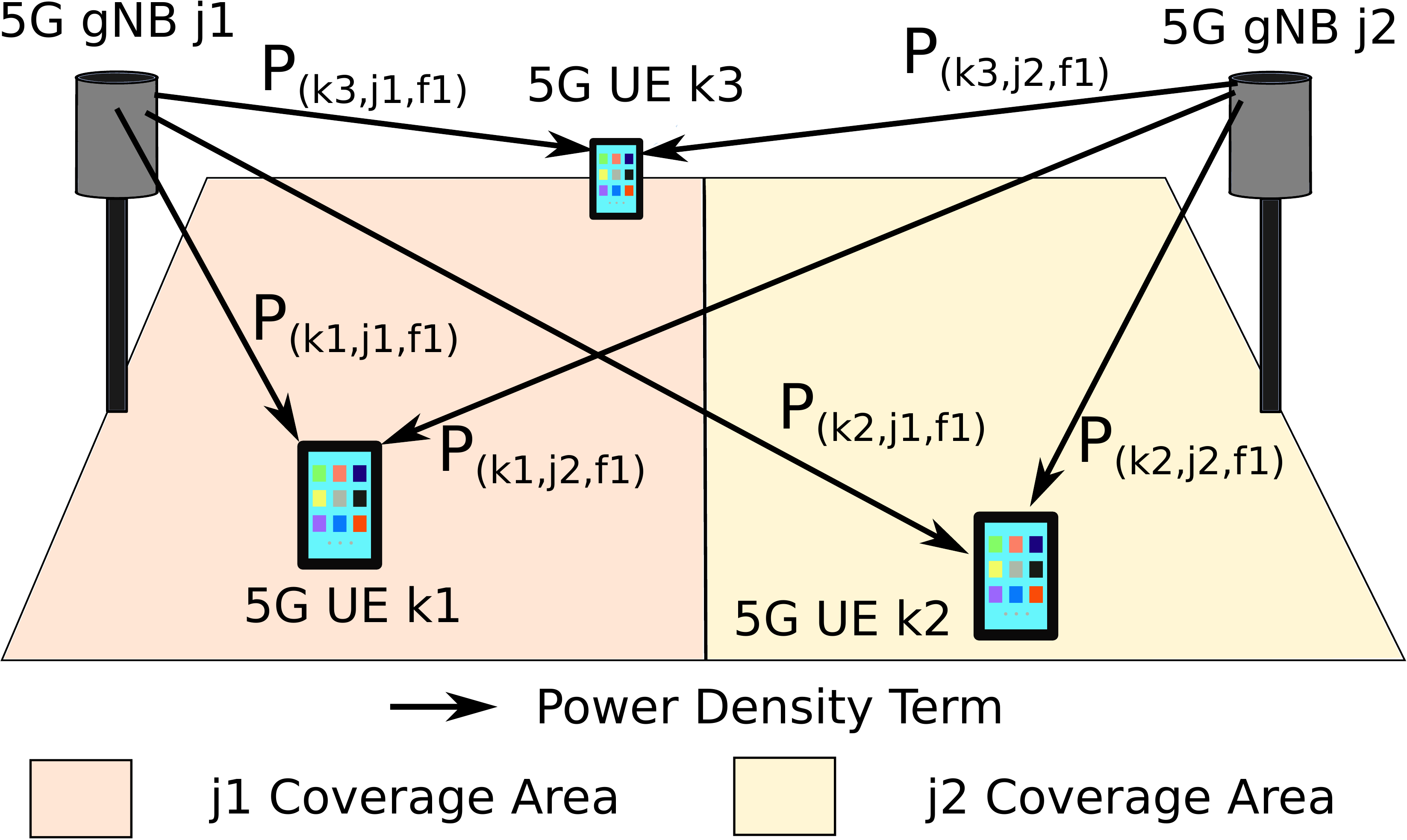}
	\caption{Power density terms in the same toy case scenario of Fig.~\ref{fig:signal_interference_terms}.}
	\label{fig:power_density_term}
\end{figure}

\begin{table*}[t]
\centering
\caption{Comparison across \ac{ICNIRP} guidelines \cite{icnirp,ICNIRPGuidelines:98}, Italian regulations \cite{italianreg,scale4g} and Rome regulations \cite{italianreg,scale4g,romereg}.$^{(a),(b),(c)}$}
\label{tab:regulation_comparison}
     \begin{threeparttable}
\scriptsize
	\begin{tabular}{|c|c|c|c|c|c|c|}
	\hline
\rowcolor{Coral}	& & & & & & \\[-0.8em]
\rowcolor{Coral}  & & \textbf{Frequency} & \textbf{Max. incident} & \textbf{Max. incident} & \textbf{Averaging} &  \textbf{Min. distance $D^{\text{MIN}}$}   \\[-0.05em] 
\rowcolor{Coral} \multirow{-2}{*}{\textbf{ID}}	 & \multirow{-2}{*}{\textbf{Name}}	& \textbf{range} & \textbf{electric field} & \textbf{power density $L_f$} & \textbf{time interval} &  \textbf{from sensitive places}   \\[0.2em]
	\hline 
	& & & & & & \\[-0.8em]
		& & 400 - 2000~[MHz] & 1.375$\cdot f^{0.5}$~[V/m] & $f$/200~[W/m$^2$] & &  \\
	\multirow{-2}{*}{\textit{R1}}& \multirow{-2}{*}{ICNIRP 1998 Guidelines \cite{ICNIRPGuidelines:98}}		& 2 - 300~[GHz] & 61~[V/m] & 10~[W/m$^2$] & \multirow{-2}{*}{6~[min] - up to 10~[GHz]} & \multirow{-2}{*}{-}  \\
		\hline
\rowcolor{Linen} & & & & & & \\[-0.8em]
\rowcolor{Linen}	&	 & 400 - 2000~[MHz] & 1.375$\cdot f^{0.5}$~[V/m] & $f$/200~[W/m$^2$] & &  \\
\rowcolor{Linen}	\multirow{-2}{*}{\textit{R2}} &	\multirow{-2}{*}{ICNIRP 2020 Guidelines \cite{icnirp}}	& 2 - 300~[GHz] & N/A & 10~[W/m$^2$] & \multirow{-2}{*}{30~[min]} &  \multirow{-2}{*}{-} \\ [0.25ex]		
		\hline
& & & & & & \\[-0.8em]
		& Italian Regulation \cite{italianreg,scale4g}  &  3 - 3000~[MHz]  & 20~[V/m] & 1~[W/m$^2$] &  & \\ [0.25ex]
		\multirow{-2}{*}{\textit{R3}} & (General Public Areas)		& 3 - 300~[GHz] & 40~[V/m] & 4~[W/m$^2$] & \multirow{-2}{*}{6~[min]} & \multirow{-2}{*}{-} \\ [0.25ex]
		\hline
\rowcolor{Linen} & & & & & & \\[-0.8em]
\rowcolor{Linen}		& Italian Regulation \cite{italianreg,scale4g} &   &  &  && \\ [0.25ex]
\rowcolor{Linen}		\multirow{-2}{*}{\textit{R4}} & (Residential Areas)		& \multirow{-2}{*}{0.1~[MHz] - 300~[GHz]} & \multirow{-2}{*}{6~[V/m]} & \multirow{-2}{*}{0.1~[W/m$^2$]} & \multirow{-2}{*}{24~[h]} & \multirow{-2}{*}{-}  \\ [0.25ex]
		\hline
& & & & & & \\[-0.8em]
			&	Rome regulation \cite{italianreg,scale4g,romereg} &  3 - 3000~[MHz]  &  20~[V/m] & 1~[W/m$^2$] & &  \\ [0.25ex]
		\multirow{-2}{*}{\textit{R5}} & (General Public Areas)		& 3 - 300~[GHz] & 40~[V/m] & 4~[W/m$^2$] & \multirow{-2}{*}{6~[min]} &  \multirow{-2}{*}{100~[m]} \\ [0.25ex]
\hline
\rowcolor{Linen}			&	Rome regulation \cite{italianreg,scale4g,romereg} &   &  & & &  \\ [0.25ex]
\rowcolor{Linen}		\multirow{-2}{*}{\textit{R6}} & (Residential Areas)		& \multirow{-2}{*}{0.1~[MHz] - 300~[GHz]} & \multirow{-2}{*}{6~[V/m]} & \multirow{-2}{*}{0.1~[W/m$^2$]} & \multirow{-2}{*}{24~[h]} &  \multirow{-2}{*}{100~[m]} \\ [0.25ex]
	\hline
	\end{tabular}
        \begin{tablenotes}
            \item[(a)] N/A stands for Not Applicable, meaning that the related quantity does not have to be taken into account when dealing with a compliance assessment.
            \item[(b)]  {The averaging time interval is the required amount of time to compute the average exposure levels that need to be compared against the maximum limit values}. 
	    \item[(c)] $f$ is the used frequency in MHz.
        \end{tablenotes}
\vspace{-1mm}
     \end{threeparttable}
\end{table*}

By assuming the maximum achievable numeric gain $F_{(k,l,f)}=1$ (as in \cite{itutk70}), Eq.~(\ref{eq:power_density_original}) is simplified into:
\begin{equation}
\label{eq:power_density_new}
P_{(k,l,f)} = \frac{\text{EIRP}_{(l,f)}}{4 \pi \cdot D_{(k,l,f)}^2} 
\end{equation}

Given $P_{(k,l,f)}$, the electric field value $E_{(k,l,f)}$ is then expressed as:
\begin{equation}
\label{eq:tot_emf}
E_{(k,l,f)}=\sqrt{P_{(k,l,f)} \cdot Z_0}
\end{equation}
where $Z_0$ denotes the free space wave impedance.

Fig.~\ref{fig:power_density_term} reports the power density terms $P_{(k,l,f)}$ in the same toy-case scenario of Fig.~\ref{fig:signal_interference_terms}, which is composed by two \acp{gNB} operating at the same frequency $f$ and three \acp{UE}. Actually, the total power density that is received by each \ac{UE} is a linear combination of the single terms radiated by the two \acp{gNB}. For example, the total power density radiated over 5G \ac{UE} $k1$ is simply equal to $P_{(k1,j1,f1)}+P_{(k1,j2,f1)}$. Note that, when considering the total electric field from multiple \acp{gNB}, a root mean square operator has to be applied. By considering the previous example, the total electric field that is radiated over \ac{UE} $k1$ is equal to $\sqrt{E_{(k1,j1,f1)}^2+E_{(k1,j2,f1)}^2}$. This operation introduces a non-linearity in the computation of the total exposure, which also generates non-linear constraints when binary variables are employed to select the set of sites that have to be installed out of the candidate ones. To overcome this issue, in this work we consider the computation of the total exposure through the power density metric, which instead allows to preserve the linearity of the constraints.

When computing the exposure from a 5G \ac{gNB}, a key role is played by the \ac{EIRP} value appearing in Eq.~(\ref {eq:power_density_new}). Clearly, the higher is the \ac{EIRP}, the larger will be also the received power density $P_{(k,l,f)}$.  This fact imposes to precisely estimate \ac{EIRP} values that match real 5G \ac{gNB} exposure patterns. In this context, two important elements that affect the \ac{EIRP} values of a 5G \ac{gNB} are the temporal variation and the statistical variation of the radiated power. We refer the interested reader to the \ac{IEC} standards \cite{iec1,iec2} for an overview about these aspects. In brief, the temporal variation is due to the fact that the number of 5G \ac{UE} (and their  traffic over the cellular network) exhibits a day-night pattern. Actually, the actual output  power levels of the 5G \ac{gNB} match this variation, with radiated power higher during the day and clearly lower during the night.  On the other hand, the application of \ac{MIMO} with beamforming features introduces strong variations in the radiated power over the territory, resulting in an exposure that is concentrated only to the zones where the served users are located. This issue is typically taken into account by solving statistical exposure models, which allow to compute spatially averaged radiated power values, as a consequence of the actual user distribution over the territory. 

In this work, we take into account the  {aforementioned} temporal and the statistical variability of the \ac{EIRP} from 5G \ac{gNB} by introducing two scaling parameters, denoted as $R^{\text{TIME}}_{(l,f)} \in (0,1]$ and $R^{\text{STAT}}_{(l,f)} \in (0,1]$, respectively.  {More in depth, $R^{\text{TIME}}_{(l,f)}$ captures the temporal variation of the radiated power (typically on a daily base), while $R^{\text{STAT}}_{(l,f)}$ models the statistical variation of exposure over the coverage area.}\footnote{ {The application} of  {scaling parameters to compute realistic \ac{EIRP} values} is in line  {with the relevant literature involved in the 5G exposure modeling} \cite{iec2,ThoFurTor:17,ecoscienza,scalingparam}.}  {As a result,} the scaled \ac{EIRP} from \ac{gNB} installed at site $l$ and operating on frequency $f$ is computed as:
\begin{equation}
\label{eq:eirp_ts}
\text{EIRP}^\text{TS}_{(l,f)}=\text{EIRP}_{(l,f)} \cdot R^{\text{TIME}}_{(l,f)} \cdot R^{\text{STAT}}_{(l,f)}
\end{equation}

In addition, let us introduce the power density $P^{\text{TS}}_{(k,l,f)}$ computed from $\text{EIRP}^\text{TS}_{(l,f)}$. By adopting Eq.~(\ref{eq:eirp_ts}) and the left-hand side of Eq.~(\ref{eq:power_density_new}), $P^{\text{TS}}_{(k,l,f)}$ is formally expressed as:
\begin{equation}
\label{eq:power_density}
P^\text{TS}_{(k,l,f)} = P_{(k,l,f)} \cdot R^{\text{TIME}}_{(l,f)} \cdot R^{\text{STAT}}_{(l,f)}
\end{equation}

In this work, we use Eq.~(\ref{eq:power_density_new}),(\ref{eq:power_density}) to characterize the level of exposure from a 5G \ac{gNB} located at site $l$, operating on frequency $f$ and radiating over \ac{UE} $k$. Moreover, we demonstrate that the actual values of $R^{\text{TIME}}_{(l,f)}$ and $R^{\text{STAT}}_{(l,f)}$ have a crucial role in determining the level of exposure and consequently the set of \acp{gNB} that are  installed over the territory. 

\subsection{5G EMF Constraints}

We then consider the third building block of our framework, i.e., the integration of the 5G \ac{EMF} constraints defined in the regulations. To this aim, Tab.~\ref{tab:regulation_comparison} reports the set of regulations \textit{R1})-\textit{R6}),  which include \ac{ICNIRP} guidelines (\textit{R1}-\textit{R2}), Italian national regulations \textit{R3})-\textit{R4}), and the local \ac{EMF} regulations enforced in the city of Rome \textit{R5})-\textit{R6}). For each regulation/guideline, the table reports: \textit{i}) the frequency range relevant to 5G, \textit{ii}) the maximum electric field limit for each frequency $f$, \textit{iii}) the maximum power density limit $L_{f}$ for each $f$, \textit{iv}) the time interval to compute the average \ac{EMF} that has to be compared against the limit value and \textit{v}) the (eventual) minimum distance constraints that have to be ensured between the 5G installations and the sensitive places.
As a side comment, we include in Tab.~\ref{tab:regulation_comparison} the \ac{ICNIRP} 1998 guidelines \cite{ICNIRPGuidelines:98} and \ac{ICNIRP} 2020 ones \cite{icnirp}, due to the fact that the formers are still adopted in many countries in the world, while the latters are the up-to-date regulations which are going to be adopted in the forthcoming month/years, and hence in parallel with the deployment of 5G networks. 

Several considerations hold by analyzing Tab.~\ref{tab:regulation_comparison}. First of all, \textit{R2}) defines a power density limit and not a limit based on electric field strength, for all the frequencies between 2~[GHz] and 300~[GHz]. This fact further corroborates our choice for selecting the power density as the reference metric when performing the compliance assessment against the maximum limits. Second, the Italian regulations \textit{R3})-\textit{R4}) are in general stricter than \textit{R1})-\textit{R2}), both in terms of electric field and in terms of power density. Third, the Italian regulations in \textit{R3})  and \textit{R4}) further differentiate between general public areas (e.g., zones of the territory where the population is not continuously living) and residential areas (e.g., zones where people tend to live and/or work), respectively. Interestingly, \textit{R4}) regulations are more restrictive than \textit{R3}). Fourth, the city of Rome applies a minimum distance $D^{\text{MIN}}$ from sensitive places in addition to the strict \ac{EMF} limits defined in \textit{R3})-\textit{R4}). Therefore, the regulations \textit{R5})-\textit{R6}) further restrict \textit{R3})-\textit{R4}). Fifth, the averaging time interval strongly varies across the different regulations, ranging from  values of few minutes to 24 hours. This interval plays a crucial role in estimating the average \ac{EMF} that has to be compared against the limit thresholds. Clearly, the lower is the time interval, the higher will be the influence of (possible) spikes on the average \ac{EMF}. On the other hand, the higher is the time interval, the lower will be the impact of spikes on the average \ac{EMF}. As a side comment, the instantaneous \ac{EMF} field \textit{can} be higher than the thresholds reported in Tab.~\ref{tab:regulation_comparison}. The actual metric that is meaningful for comparison against the limit is in fact the average \ac{EMF} over the time interval defined in each regulation.

After analyzing the \ac{EMF} regulations, a natural question emerges: How to perform the compliance assessment w.r.t. the maximum limits when multiple Base Stations operating at different frequencies radiate the same area of territory? To answer this question, let us denote with $\sum_{l \in \mathcal{L}} P_{(k,l,f)}$ the composite power density that is radiated over \ac{UE} $k$ by all the Base Stations operating on frequency $f \in \mathcal{F}$, where $\mathcal{F}$ is the set of frequencies in use. The compliance w.r.t. the limits is ensured over $k$ if the following condition holds:
\begin{equation}
\label{eq:compliance}
\sum_{f \in \mathcal{F}} \frac{\sum_{l \in \mathcal{L}} P_{(k,l,f)}}{L_f} \leq 1
\end{equation}
Clearly, the power density terms $P_{(k,l,f)}$ of Eq.~(\ref{eq:compliance}) have to be computed as average values over the time intervals reported in Tab.~\ref{tab:regulation_comparison}. In addition, the actual \ac{EMF} metric that is measured under practical conditions is the electric field strength $E_{(k,l,f)}$, which is then translated into power density $P_{(k,l,f)}$ by applying Eq.~(\ref{eq:tot_emf}).

\subsection{Summary and Next Steps}

The model of Marzetta \cite{marzetta} is used to control the \ac{SIR} and consequently the maximum downlink throughput provided to the \ac{UE}. The point source model of ITU-T K.70 \cite{itutk70}, integrated with scaling parameters that characterize the exposure from 5G \ac{gNB}, is instead used to compute  {conservative estimation of} power density. Finally, the limits defined by international/national bodies and local municipalities are used to ensure that the composite power density is lower than the thresholds. In addition, a minimum distance rule from sensitive places is ensured in accordance to the local regulation. In the next section, we join together these building blocks, in order to build an innovative formulation able to balance between \acp{gNB} installation costs and 5G service coverage level, while ensuring \ac{QoS} and strict \ac{EMF} constraints.

\section{Optimal 5G Planning Formulation}
\label{sec:planning}

We divide our formulation in the following steps: \textit{i}) preliminaries, \textit{ii}) set definition, \textit{iii}) constraint, variables and input parameters, \textit{iv}) overall formulation.

\subsection{Preliminaries}

In the previous section we have provided the models to compute the service coverage and the power density for each \ac{UE} in the scenario under consideration. In this section, we generalize these models by extending the evaluations from a sparse set of \ac{UE} to a tessellation of non-overlapping squared pixels that fully cover the area under interest.  {More formally, the pixel is a small area of territory, in which similar propagation conditions are experienced. We refer the reader to Appendix~\ref{app:pixel} for a detailed discussion about the pixel tessellation.}

Focusing then on the modelling of the \ac{EMF} regulations, we assume to enforce the most restrictive ones, namely \textit{R5})-\textit{R6}) of Tab.~\ref{tab:regulation_comparison}. Therefore, we distinguish between general public areas, residential areas and zones within the minimum distance from sensitive places. However, we point out that the other guidelines presented in Tab.~\ref{tab:regulation_comparison} can be easily implemented in our framework by applying different limit thresholds and/or by setting the minimum distance $D^{\text{MIN}}$ to zero.

\subsection{Set Definition}

Let us denote with $\mathcal{P}$ the set of pixels under consideration. $\mathcal{P}^{\text{RES}} \subset \mathcal{P}$ and  $\mathcal{P}^{\text{GEN}} \subset \mathcal{P}$ are the subsets of pixels in residential areas and in general public areas, respectively. Moreover, $\mathcal{P}^{\text{SENS}} \subset \mathcal{P}$ is the subset of pixels in sensitive areas. In addition, let $\mathcal{L}$ be the set of candidate locations (sites) that can host 5G \ac{gNB} equipment. Eventually, let $\mathcal{F}$ be the set of frequencies that can be exploited by 5G \acp{gNB}.

\subsection{Constraints, Variables and Input Parameters}

We then detail constraints, variables and input parameters to our problem by adopting a step-by-step approach. We also refer the reader to Tab.~\ref{tab:notation}  {of Appendix~\ref{app:notation}} for the main notation that is adopted throughout the section.

\textbf{5G Coverage and Service Constraints.} We initially model the constraint that a pixel $p \in \mathcal{P}$ can be covered by a 5G \ac{gNB} located in $l$ only if the distance $D_{(p,l,f)}$ between the pixel and the installed \ac{gNB} is lower than a maximum one, denoted with $D^{\text{MAX}}_f$, where $f$ is the operating frequency of the 5G \ac{gNB} installed in $l$. More formally, we have:
\begin{equation}
\label{eq:cov_pixel}
D_{(p,l,f)} \cdot x_{(p,l,f)}  \leq D^{\text{MAX}}_f \cdot y_{(l,f)}, \quad \forall p \in \mathcal{P}, l \in \mathcal{L}, f \in \mathcal{F}
\end{equation}
where  $x_{(p,l,f)}$ is a binary variable, set to 1 if $p$ is served by \ac{gNB} operating on frequency $f$ and located at $l$ (0 otherwise). Moreover, $y_{(l,f)}$ is another binary variable, set to 1 if 5G \ac{gNB} operating on frequency $f$ is installed at location $l$ (0 otherwise).

We then impose the constraint that each pixel $p$ can be served by at most $N^{\text{SER}}\geq1$ \acp{gNB} at the same time:\footnote{Although multiple coverage from different \acp{gNB} is a desirable condition, the increase in the number of \acp{gNB} covering the same pixel may introduce side effects, like an increase in the handover rates for \ac{UE}, which may dramatically decrease the perceived \ac{QoS} \cite{chiaraviglio2018not}. Therefore, we introduce a constraint to control the number of \acp{gNB} serving the same pixel. } 
\begin{equation}
\label{eq:maximum_number_of_cells}
\sum_{l \in \mathcal{L}} \sum_{f \in \mathcal{F}} x_{(p,l,f)} \leq N^{\text{SER}}, \quad \forall p \in \mathcal{P}
\end{equation}

In the following, we impose that the \ac{SIR} value in each pixel $p$ that is served by a 5G \ac{gNB} operating on frequency $f$ has to be higher than a minimum value $S^{\text{MIN}}_{f}$. By adopting the \ac{SIR} computation already introduced in Eq.~(\ref{eq:sir}),\footnote{Since we have extended the evaluation from the single \ac{UE} to the whole set of pixels, the $k$ index of Eq.~(\ref{eq:sir}) is replaced with $p \in \mathcal{P}$.} we have:\footnote{ {The adopted \ac{SIR} model assumes a flat terrain, i.e., without considering possible hills/valleys that can affect the \ac{SIR} computation. Intuitively, a possible way to take into account the terrain elevation changes would be to exploit a gray-scale pixel map} \cite{kamar2010optimized} {, with different shades representing different elevation value. We leave the investigation of this aspect as future work.}}
\begin{eqnarray}
\label{eq:sir_not_linear}
\underbrace{\frac{\beta^2_{(l,p,l,f)}\cdot y_{(l,f)}}{\sum_{l_2 \neq l \in \mathcal{L}}\beta^2_{(l,p,l_2,f)}\cdot y_{(l_2,f)}}}_{\text{SIR} \quad S(p,l,f)} \geq \underbrace{S^{\text{MIN}}_{f} \cdot x_{(p,l,f)}}_{\text{Min. SIR Threshold}}, \nonumber \\ \quad \forall p \in \mathcal{P}, l \in \mathcal{L}, f \in \mathcal{F}
\end{eqnarray} 

Clearly, the previous constraint is not linear, due to the optimization variables $y_{(l,f)}$ that appear on both the numerator and the denominator of the left-hand side, coupled with the presence of the $x_{(p,l,f)}$ variables on the right-hand side of the constraint.  {To solve this issue, we apply a linearization procedure. More in depth, we refer the reader to Appendix~\ref{app:linearization} for the details, while we report below the final outcomes. In brief, we replace Eq.~(\ref{eq:sir_not_linear}) with the following linear constraints:}
\begin{equation}
\label{eq:sir_aux_1}
v_{(l,p,l_2,f)} \leq x_{(p,l,f)}, \quad \forall p \in \mathcal{P}, l \in \mathcal{L}, l_2 \in \mathcal{L},  f \in \mathcal{F} 
\end{equation}

\begin{equation}
\label{eq:sir_aux_2}
v_{(l,p,l_2,f)} \leq y_{(l_2,f)}, \quad \forall p \in \mathcal{P}, l \in \mathcal{L}, l_2 \in \mathcal{L},  f \in \mathcal{F} 
\end{equation}

\begin{eqnarray}
\label{eq:sir_aux_3}
v_{(l,p,l_2,f)} \geq x_{(p,l,f)} + y_{(l_2,f)} - 1, \quad \quad \quad \quad \quad \quad \nonumber \\ \quad \forall p \in \mathcal{P}, l \in \mathcal{L}, l_2 \in \mathcal{L},  f \in \mathcal{F} 
\end{eqnarray}

\begin{eqnarray}
\label{eq:sir_linear}
S^{\text{MIN}}_{f} \cdot \left[\sum_{l_2 \in \mathcal{L}}\frac{\beta^2_{(l,p,l_2,f)}}{\beta^2_{(l,p,l,f)}}\cdot v_{(l,p,l_2,f)} - x_{(p,l,f)} \right] \leq 1 \nonumber \\ \quad \forall p \in \mathcal{P}, l \in \mathcal{L}, f \in \mathcal{F} 
\end{eqnarray}

where $v_{(l,p,l_2,f)} \in \{0,1\}$ is a binary auxiliary variable. 

\textbf{Power Density Limits.} We initially select the pixels that fall in the exclusion zones of the installed 5G \acp{gNB} and therefore are not subject to the \ac{EMF} limits defined for the general public. More formally, we introduce the binary variable $w_p$, set to 1 if $p$ is located inside an exclusion zone of an installed \ac{gNB} (0 otherwise). In addition, input parameter $I^{\text{ZONE}}_{(p,l,f)}$ takes value 1 if pixel $p$ is inside the exclusion zone of \ac{gNB} operating on frequency $f$ and located at $l$ (0 otherwise). The value of $w_p$ is then set through the following set of constraints:
\begin{equation}
\label{eq:lower_bound_exclusion}
w_p \geq I^{\text{ZONE}}_{(p,l,f)} \cdot y_{(l,f)}, \quad \forall p \in \mathcal{P}, l \in \mathcal{L}, f \in \mathcal{F}
\end{equation}
\begin{equation}
\label{eq:upper_bound_exclusion}
w_p \leq \sum_{l \in \mathcal{L}} \sum_{f \in \mathcal{F}} I^{\text{ZONE}}_{(p,l,f)} \cdot y_{(l,f)}, \quad \forall p \in \mathcal{P}
\end{equation}
More in depth, constraint (\ref{eq:lower_bound_exclusion}) activates $w_p$ if $p$ is inside at least one exclusion zone of an installed \ac{gNB}. On the other hand, constraint (\ref{eq:upper_bound_exclusion}) forces $w_p$ to 0 if $p$ is outside the exclusion zones for all the installed \acp{gNB}.

In the following, we introduce the constraints to compute the power density received by pixel $p$ over frequency $f$. Let us denote with input parameter $P^{\text{ADD}}_{(p,l,f)}$ the additional power density that is received by pixel $p$ when a \ac{gNB} operating on frequency $f$ is installed in $l$. Let us denote with $P^{\text{ADD-TS}}_{(p,f)}$ the variable storing the additional power density for pixel $p$ over $f$, which is computed from $P^{\text{ADD}}_{(p,l,f)}$ by applying the scaling factors $R^{\text{TIME}}_{(l,f)}$ and $R^{\text{STAT}}_{(l,f)}$. More formally, we include Eq.~(\ref{eq:power_density}) to our problem, thus yielding:
\begin{eqnarray}
\label{eq:p_tot_def}
P^{\text{ADD-TS}}_{(p,f)} = (1-w_p) \sum_{l \in \mathcal{L}} P^{\text{ADD}}_{(p,l,f)} \cdot R^{\text{TIME}}_{(l,f)} \cdot R^{\text{STAT}}_{(l,f)} \cdot y_{(l,f)} &&  \nonumber \\
  \forall p \in \mathcal{P}, f \in \mathcal{F} &&
\end{eqnarray}
In the previous constraint, the term $(1 - w_p)$ ensures that a pixel falling inside the exclusion zone of an installed 5G \ac{gNB} is not considered when the power density is evaluated against the limits.  {However, the presence of the term $(1 - w_p) \cdot y_{(l,f)}$ makes Eq.~(\ref{eq:p_tot_def}) not linear. We then linearize it by:} \textit{i}) introducing the auxiliary variable $z_{(p,l,f)} \in \{0,1\}$, and \textit{ii}) replacing Eq.~(\ref{eq:p_tot_def}) with the following set of constraints:
\begin{equation}
\label{eq:lin1_res}
z_{(p,l,f)} \leq (1-w_p), \quad \forall p \in \mathcal{P}, l \in \mathcal{L}, f \in \mathcal{F}
\end{equation}
\begin{equation}
\label{eq:lin2_res}
z_{(p,l,f)} \leq y_{(l,f)}, \quad \forall p \in \mathcal{P}, l \in \mathcal{L}, f \in \mathcal{F}
\end{equation}
\begin{equation}
\label{eq:lin3_res}
z_{(p,l,f)} \geq y_{(l,f)} - w_p, \quad \forall p \in \mathcal{P}, l \in \mathcal{L}, f \in \mathcal{F}
\end{equation}

\begin{eqnarray}
\label{eq:p_tot_def_lin}
P^{\text{ADD-TS}}_{(p,f)}= \sum_{l \in \mathcal{L}} P^{\text{ADD}}_{(p,l,f)} \cdot R^{\text{TIME}}_{(l,f)} \cdot R^{\text{STAT}}_{(l,f)} \cdot z_{(p,l,f)}\nonumber \\ \quad \forall p \in \mathcal{P}, f \in \mathcal{F}
\end{eqnarray}

 {To give more insights, we demonstrate with a simple example that Eq.~(\ref{eq:lin3_res}) is essential to correctly set $z_{(p,l,f)}$. Let us consider the following case: \textit{i}) the pixel $p$ is outside the exclusion zone, and hence $w_p=0$, and \textit{ii}) there is a gNB installed at location $l$ operating at frequency $f$, and therefore $y_{(l,f)}=1$. When only (\ref{eq:lin1_res})-(\ref{eq:lin2_res}) are considered (without (\ref{eq:lin3_res})), a feasible solution in this case would be to set $z_{(p,l,f)}=0$, which is not correct because the product $(1-w_p)\cdot y_{(l,f)}$ in the non-linear constraint (\ref{eq:p_tot_def}) is equal to one in the same case. On the other hand, when (\ref{eq:lin3_res}) is introduced, $z_{(p,l,f)}$ is (correctly) set to 1.} 

In a similar way, we compute the additional total power density $P^{\text{ADD-NOTS}}_{(p,f)}$ that is received by pixel $p$ on frequency $f$, without applying the scaling factors $R^{\text{TIME}}_{(l,f)}$, $R^{\text{STAT}}_{(l,f)}$. $P^{\text{ADD-NOTS}}_{(p,f)}$ is meaningful when $p$ belongs to a general public area (e.g., \textit{R5} of Tab.~\ref{tab:regulation_comparison}). In this case, in fact, the scaling parameters are not applied.\footnote{A revision in the regulations may be introduced in the future in order to introduce scaling parameters also for general public areas.} Therefore, we have:
\begin{eqnarray}
\label{eq:p_tot_def_lin_2}
P^{\text{ADD-NOTS}}_{(p,f)}=\sum_{l \in \mathcal{L}} P^{\text{ADD}}_{(p,l,f)}  \cdot z_{(p,l,f)}, \forall p \in \mathcal{P}, f \in \mathcal{F}
\end{eqnarray}


\begin{figure}[t]
	\centering
 	\subfigure[Pixel outside the exclusion zone]
	{
		\includegraphics[width=7cm]{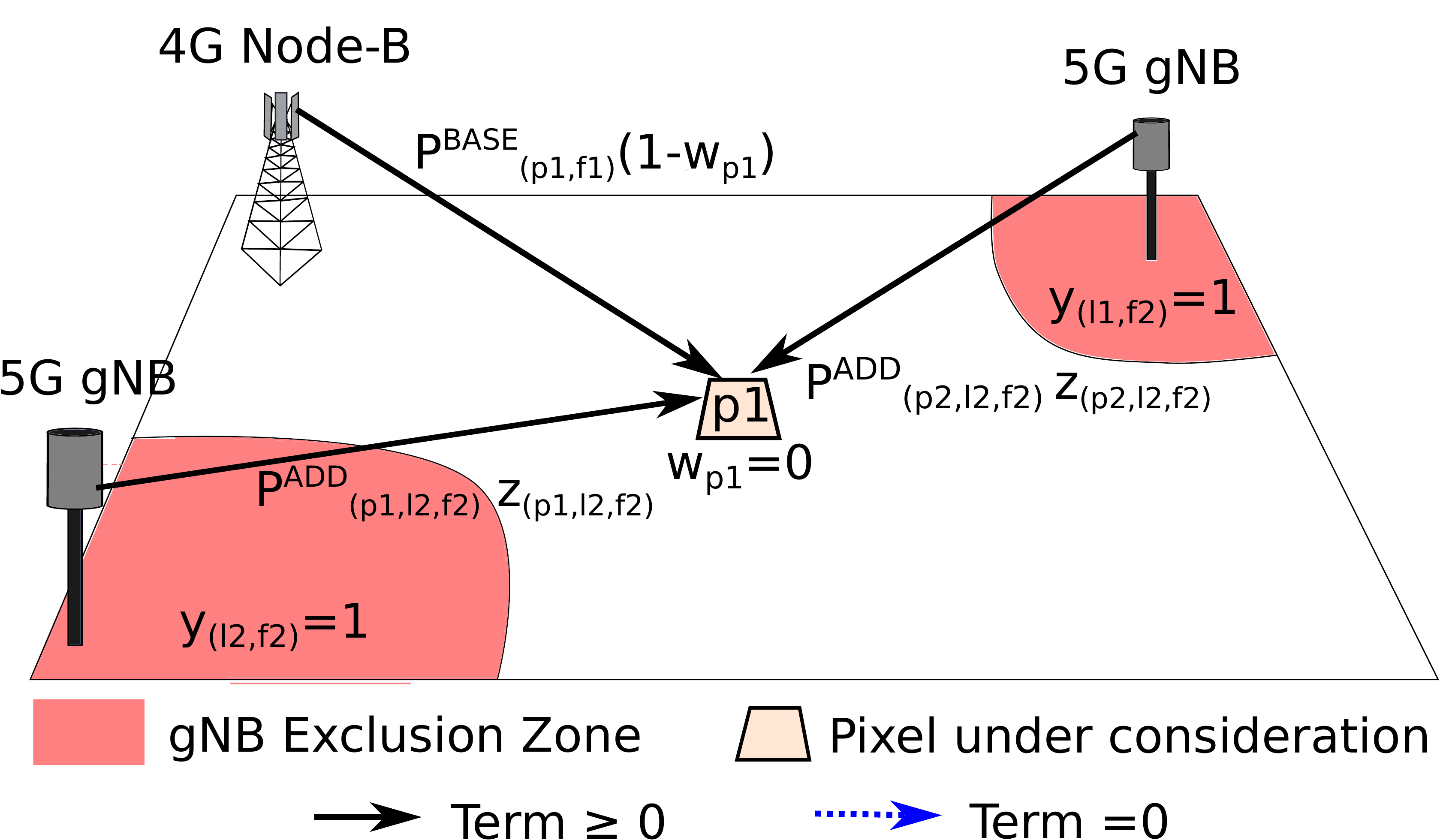}
		\label{fig:power_density_pixel_outside_exclusion_zone}

	}
 	\subfigure[Pixel inside the exclusion zone]
	{
		\includegraphics[width=7cm]{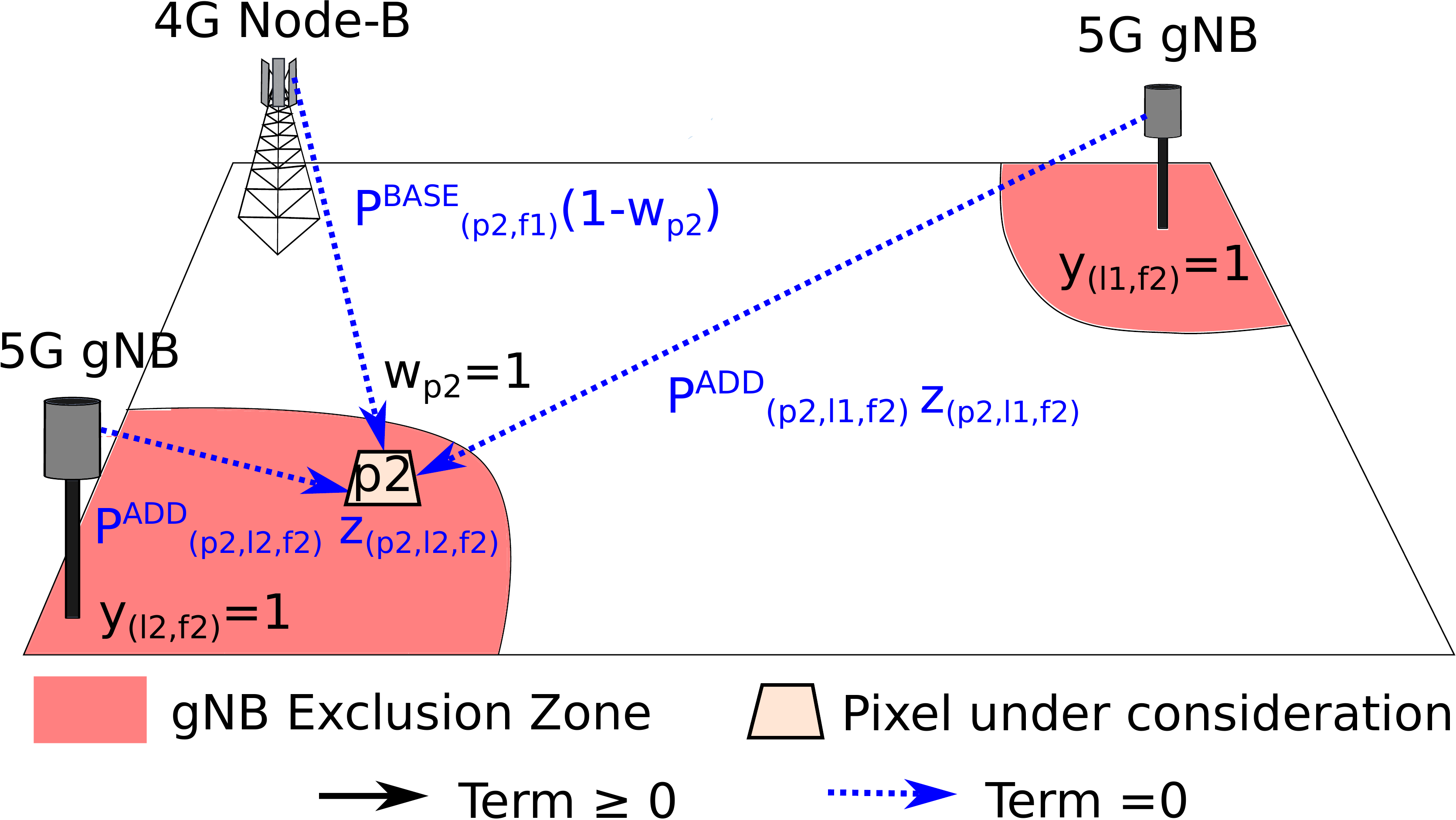}
		\label{fig:power_density_pixel_inside_exclusion_zone}
	}
	\caption{Computation of the power density terms in a toy-case scenario that includes exclusion zones from the newly installed \acp{gNB}.}
	\label{fig:power_density_pixel_computation}
\end{figure}

We then impose the power density limit on residential areas, which has to be ensured for each pixel $p \in \mathcal{P}^{\text{RES}}$.  More technically, we include the compliance assessment model of Eq.~(\ref{eq:compliance}) in our problem, thus obtaining:
\begin{equation}
\label{eq:power_density_limit_res}
\sum_{f \in \mathcal{F}} \frac{P^{\text{BASE}}_{(p,f)} \cdot (1-w_p) + P^{\text{ADD-TS}}_{(p,f)}}{L^{\text{RES}}_f}  \leq 1, \quad \forall p \in  \mathcal{P}^{\text{RES}} 
\end{equation}
where $P^{\text{BASE}}_{(p,f)}$ is the baseline power density over $p$ from all the radio-frequency sources operating of frequency $f$ and already installed in the scenario under consideration.

 {Finally}, we impose the power density limit on general public areas, which has to be ensured for each pixel $p \in \mathcal{P}^{\text{GEN}}$, by introducing the following constraint:
\begin{equation}
\label{eq:power_density_limit_gen}
\sum_{f \in \mathcal{F}} \frac{P^{\text{BASE}}_{(p,f)}\cdot (1-w_p)  + P^{\text{ADD-NOTS}}_{(p,f)}}{L^{\text{GEN}}_f}  \leq 1, \quad \forall p \in  \mathcal{P}^{\text{GEN}} 
\end{equation}

In order to clarify how the computation of the power density is governed by the optimization variables modelling the exclusion zones, Fig.~\ref{fig:power_density_pixel_computation} shows a graphical representation of $P^{\text{ADD}}_{(p,l,f)}\cdot z_{(p,l,f)}$ terms that appear in Eq.~(\ref{eq:p_tot_def_lin})-(\ref{eq:p_tot_def_lin_2}) as well as $P^{\text{BASE}}_{(p,f)} \cdot (1 - w_p)$ that are included in Eq.~(\ref{eq:power_density_limit_res})-(\ref{eq:power_density_limit_gen}). More in depth, the considered toy-case scenario includes one legacy 4G Node-B already installed over the territory and two newly installed \acp{gNB}. Fig.~\ref{fig:power_density_pixel_outside_exclusion_zone} focuses on a pixel $p1$ outside the exclusion zones of the \acp{gNB}. By applying Eq.~(\ref{eq:lower_bound_exclusion})-(\ref{eq:upper_bound_exclusion}), it holds that $w_{p1}=0$. Then, by applying constraints (\ref{eq:lin1_res})-(\ref{eq:lin3_res}), it holds that: $z_{(p1,l1,f2)}=1$, $z_{(p1,l2,f2)}=1$. Consequently, the computation of the total power density in (\ref{eq:power_density_limit_res})-(\ref{eq:power_density_limit_gen}) will include the contributions from the newly installed \acp{gNB} as well as the already installed 4G Node-B. Therefore, the installation of the newly installed \ac{gNB} is possible only if the \ac{EMF} compliance assessment constraints (\ref{eq:power_density_limit_res})-(\ref{eq:power_density_limit_gen}) are ensured. On the other hand, Fig.~\ref{fig:power_density_pixel_inside_exclusion_zone} reports the power density terms when the considered pixel $p2$ falls inside the exclusion zone of a \ac{gNB}. By applying Eq.~(\ref{eq:lower_bound_exclusion})-(\ref{eq:upper_bound_exclusion}),(\ref{eq:lin1_res})-(\ref{eq:lin3_res}) it holds that: $w_{p2}=1$, $z_{(p2,l1,f2)}=0$, $z_{(p2,l2,f2)}=0$. As a result, the power density terms $P^{\text{ADD}}_{(p2,l,f)}\cdot z_{(p2,l,f)}$ and $P^{\text{BASE}}_{(p2,f1)} \cdot (1 - w_{p2})$ are now set to zero. Therefore, the \ac{EMF} compliance assessment constraints (\ref{eq:power_density_limit_res})-(\ref{eq:power_density_limit_gen}) are always ensured for $p2$.

\textbf{Minimum Distance from Sensitive Places.} We then introduce the distance constraints that are included in regulations \textit{R5})-\textit{R6}) of Tab.~\ref{tab:regulation_comparison}. We remind that these constraints define a minimum distance between each installed 5G \ac{gNB} and a sensitive place. More formally, we have:
\begin{equation}
\label{eq:min_distance}
 D_{(p,l,f)} \cdot y_{(l,f)} \geq D^{\text{MIN}}, \quad \forall p \in \mathcal{P}^{\text{SENS}}, l \in \mathcal{L}, f \in \mathcal{F}
\end{equation}


\textbf{Site Constraints.} In the following, we impose that each site location can host up to $N^{\text{MAX}}$ \ac{gNB} types operating at different frequencies. More formally, we have:
\begin{equation}
\label{eq:max_equipment}
\sum_{f \in \mathcal{F}} y_{(l,f)} \leq N^{\text{MAX}}, \quad \forall l \in \mathcal{L}
\end{equation}

In addition, we introduce the indicator parameter $I^{\text{FREQ}}_{(l,f)}$, taking value 1 if \ac{gNB} of type $f$ can be hosted at location $l$, 0 otherwise. Clearly, a \ac{gNB} operating on frequency $f$ can be installed at $l$ only if the indicator parameter is 1. More formally, we have:
\begin{equation}
\label{eq:indicator_parameter}
y_{(l,f)} \leq I^{\text{FREQ}}_{(l,f)}, \quad \forall l \in \mathcal{L}, f \in \mathcal{F}
\end{equation}
 
\begin{algorithm*}[t]
        \small
        \caption{Pseudo-Code of \textsc{PLATEA} algorithm}
        \label{alg:platea}
	\textbf{Input:} Parameters and sets defined in Tab.~\ref{tab:notation}  {of Appendix~\ref{app:notation}}, assumed to be available through global variables\\
	\textbf{Output:} Variables $y$, $x$, $P^{\text{ADD-TS}}$, $P^{\text{ADD-NOTS}}$, $w$, $C^{\text{TOT}}$ for the best solution found
        \begin{algorithmic}[1]	
	\State{\texttt{// Step 1: Initialization}}
        \State{num\_f1\_max=$\sum_{l \in \mathcal{L}} I^{\text{FREQ}}_{(f1,l)}$; \texttt{// Max. number of installable $f1$ \acp{gNB}}}
	\State{num\_f2\_max=$\sum_{l \in \mathcal{L}} I^{\text{FREQ}}_{(f2,l)}$; \texttt{// Max. number of installable $f2$ \acp{gNB}}}
	\State{best\_obj=inf; \texttt{// Best objective initialization}}
	\State{[$y$, $x$, $P^{\text{ADD-TS}}$, $P^{\text{ADD-NOTS}}$, $w$, $C^{\text{TOT}}$]=\textsc{initial\_sol}(); \texttt{// Initialization of solution variables}}
	\State{\texttt{// Step 2: iteration over candidate deployments with frequency $f1$}}
	\For{num\_f1=1:num\_f1\_max}
		\State{[flag\_end x\_curr, y\_curr, pd\_curr]=\textsc{select\_best\_set\_f1}(num\_f1); \texttt{// Selection of the best deployment with num\_f1 \acp{gNB}}}
		\State{\texttt{// Step 3: iteration over candidate deployments with frequency $f2$}}
		\For{num\_f2=1:num\_f2\_max}
			\If{flag\_end==false}
				\State{sites\_f2\_comb=\textsc{extract\_sites}(num\_f2, f2); \texttt{// Extraction of the candidate deployments with frequency $f2$}}
				\For{sites\_f2 in sites\_f2\_comb}
					\State{[flag\_check, y\_curr, pd\_curr]=\textsc{install\_check}(sites\_f2, y\_curr, pd\_curr); \texttt{// Based on Eq.~(\ref{eq:lower_bound_exclusion}), (\ref{eq:upper_bound_exclusion}), (\ref{eq:lin1_res})-(\ref{eq:indicator_parameter})}}
					\If{flag\_check==true}
						\State{[x\_curr]=\textsc{associate\_pixels}(y\_curr, x\_curr); \texttt{// Based on Eq.~(\ref{eq:cov_pixel}),(\ref{eq:maximum_number_of_cells}),(\ref{eq:sir_aux_1})-(\ref{eq:sir_linear})}}
						\State{curr\_obj=\textsc{compute\_obj}(x\_curr, y\_curr); \texttt{// Based on Eq.~(\ref{tot_cost_bs_installed_computation}), (\ref{eq:tot_obj})}}
						\If{curr\_obj $<$ best\_obj}
							\State{best\_obj=curr\_obj;}
							\State{[$y$, $x$, $P^{\text{ADD-TS}}$, $P^{\text{ADD-NOTS}}$, $w$, $C^{\text{TOT}}$]=\textsc{save\_sol}(y\_curr, x\_curr, pd\_curr); \texttt{// Best Solution Saving}}
						\EndIf
						\If{(\textsc{all\_served}(x\_curr)==true)}
							\State{flag\_end=true; \texttt{// All pixels served}}						
						\EndIf
						\State{[x\_curr, y\_curr, pd\_curr]=\textsc{uninstall}($f2$, x\_curr, y\_curr, pd\_curr); \texttt{// Revert changes for $f2$}}
					\EndIf
				\EndFor
			\EndIf
		\EndFor
		\State{[x\_curr, y\_curr, pd\_curr]=\textsc{uninstall}($f1$, x\_curr, y\_curr, pd\_curr); \texttt{// Revert changes for $f1$}}
	\EndFor
	\end{algorithmic}
\end{algorithm*}

\textbf{Total Cost Computation.} Finally, we compute the total costs for installing the 5G \acp{gNB}. To this aim, let us denote with parameter $C^{\text{EQUIP}}_{f}$ the monetary costs of a 5G \ac{gNB} equipment operating on frequency $f$. In addition, let us denote with parameter $C_{(l,f)}^{\text{SITE}}$ the site installation cost for a 5G \ac{gNB} operating on frequency $f$ and installed at location $l$. The total costs $C^{\text{TOT}}$ for installing the new 5G \acp{gNB} are formally expressed as:
\begin{equation}
\label{tot_cost_bs_installed_computation}
C^{\text{TOT}} = \sum_{l \in \mathcal{L}} \sum_{f \in \mathcal{F}} \left (C^{\text{EQUIP}}_{f} + C_{(l,f)}^{\text{SITE}} \right) \cdot y_{(l,f)}
\end{equation}

\subsection{Objective Function and Overall Formulation}

 {We consider a multi-objective function that combines the total costs for installing the 5G \acp{gNB} (stored in variable $C^{\text{TOT}}$) and the number of pixels that are served by the installed 5G \acp{gNB} (stored in variables $x_{(p,l,f)}$).} The two terms are properly taken into account by the weight factor $\alpha_{(l,f)}$~[ {EUR}], which  {represents the per-pixel revenue by \ac{gNB} operating on frequency $f$ and located in $l$. More in depth, the $\alpha_{(l,f)}$ parameters have to be provided in input to our framework, and they can be estimated e.g., from the expected lifetime duration of a \ac{gNB} installation, the zones that are served by the \ac{gNB}, the projected increase of users over the years, and/or a weighed function of the aforementioned terms.}  {Interestingly, by varying the values of $\alpha_{(l,f)}$ over the set of locations $\mathcal{L}$ and over the set of frequencies $\mathcal{F}$, the operator can control the installation costs and the coverage level over the territory.}\footnote{Other alternative formulations commonly adopted during the cellular planning phase include the minimization of the \ac{CAPEX} costs under a given percentage of service coverage. However, this goal does not always guarantee problem feasibility. To overcome this issue, in this work we keep the service coverage in the objective function. As a result, our choice allows to preserve the problem feasibility on one side and to explore the impact of $\alpha_{(l,f)}$ on the obtained planning on the other one.}


The complete \textsc{Optimal Planning for 5G Networks under Service and Strict EMF Constraints (OPTPLAN-5G)} is formally expressed as:
\begin{equation}
\label{eq:tot_obj}
\text{min} \left(C^{\text{TOT}} - \sum_{p \in \mathcal{P}}\sum_{l \in \mathcal{L}}\sum_{f \in \mathcal{F}}\alpha_{(l,f)} \cdot x_{(p,l,f)}\right)
\end{equation}
subject to:
\begin{equation}
\label{eq:tot_constraints_eqn}
\begin{array}{ll}
\text{5G Coverage and Service:} & \text{Eq.}~(\ref{eq:cov_pixel},\ref{eq:maximum_number_of_cells}),(\ref{eq:sir_aux_1})-(\ref{eq:sir_linear})\\
\text{Power Density Limits:} &  \text{Eq.}~(\ref{eq:lower_bound_exclusion},\ref{eq:upper_bound_exclusion}),(\ref{eq:lin1_res})-(\ref{eq:power_density_limit_gen})\\
\text{Min. Distance Constraint:} &  \text{Eq.}~(\ref{eq:min_distance})\\
\text{Site Constraints:} &  \text{Eq.}~(\ref{eq:max_equipment}),(\ref{eq:indicator_parameter})\\
\text{Total Cost Computation:} &  \text{Eq.}~(\ref{tot_cost_bs_installed_computation})\\
\end{array}
\end{equation}

under variables: $C^{\text{TOT}} \geq 0$, $x_{(p,l,f)} \in \{0,1\}$, $y_{(l,f)} \in \{0,1\}$, $w_p \in \{0,1\}$, $v_{(p,l,l_2,f)} \in \{0,1\}$, $z_{(p,l,f)} \in \{0,1\}$.


\begin{Proposition}The \textsc{OPTPLAN-5G} problem is NP-Hard.\label{prop:NPHard}\end{Proposition}
\begin{proof}
Let us consider a special case of the problem, where a single pixel is evaluated. Moreover, let us assume that this single pixel is covered if the \ac{gNB} operating on frequency $f$ is installed in $l$, i.e., $x_{(p,l,f)}=y_{(l,f)}, \forall l \in \mathcal{L}, f \in \mathcal{F}$.\footnote{In this way, we assume that the pixel is within $D^{\text{MAX}}$ distance from all the \acp{gNB} and that the minimum value of service throughput is equal to 0.} Let us also consider the possibility to install up to one \ac{gNB} in each site, i.e., $N^{\text{MAX}}=1$. Consequently, constraint (\ref{eq:max_equipment}) can be rewritten as:
\begin{equation}
\label{eq:one_installation}
\sum_{f \in \mathcal{F}} y_{(l,f)} \leq 1, \quad \forall l \in \mathcal{L}
\end{equation}
Moreover, let us assume that the considered pixel is outside the exclusion zone of each installed \ac{gNB}, i.e., $w_p=0$. Consequently, $z_{(p,l,f)} = y_{(l,f)}, \forall l \in \mathcal{L}, f \in \mathcal{F}$. Moreover, we consider: \textit{i}) the application of general public limits, i.e., the scaling parameters are not applied and \textit{ii}) a relaxation of the power density constraints in (\ref{eq:power_density_limit_gen}) with no background power density (i.e., $P^{\text{BASE}}_{(p,f)}=0,\forall f \in \mathcal{F}$) and the limit verification for each frequency in isolation w.r.t. the other frequencies. More formally, constraint (\ref{eq:power_density_limit_gen}) is replaced with the following one:
\begin{equation}
\label{eq:power_density_comp_new}
\sum_{l \in \mathcal{L}} P^{\text{ADD}}_{(p,l,f)} \cdot y_{(l,f)} \leq L_f^{\text{GEN}}, \quad \forall f \in \mathcal{F} 
\end{equation}  
We then assume the maximization of the service coverage, which in our problem is equivalent to the maximization of the number of installed \acp{gNB}, weighted by $\alpha_{(l,f)}$. More formally, we have:
\begin{equation}
\text{max} \sum_{l \in \mathcal{L}}\sum_{f \in \mathcal{F}}\alpha_{(l,f)} \cdot y_{(l,f)}
\end{equation}
subject to: Eq.~(\ref{eq:one_installation}), Eq.~(\ref{eq:power_density_comp_new}); under variables: $y_{(l,f)} \in \{0,1\}$. It is therefore trivial to note that the aforementioned formulation is the well-known \textsc{Generalized Assignment Problem (GAP)}, which is NP-Hard \cite{martello1990knapsack}. Since \textsc{GAP} is a special case of our problem, we can conclude that also \textsc{OPTPLAN-5G} is NP-Hard. 
\end{proof}

\section{\textsc{PLATEA} algorithm}
\label{sec:algorithm}

Since the \text{OPTPLAN-5G} is NP-Hard, and therefore very challenging to be solved even for small problem instances, we design an efficient algorithm, called \textsc{PLanning Algorithm Towards \ac{EMF} Emission Assessment (PLATEA)} to practically solve it. We base our solution on the following intuitions:
\begin{enumerate}
\item we apply a \textit{divide et impera} approach, in which the complex planning problem is split into sets of subproblems. More in depth, since the different frequencies used in 5G have in general different goals (e.g., throughput maximization and/or coverage maximization), we exploit the \ac{gNB} operating frequency as the main metric to split the original problem into smaller subproblems;
\item we restrict the exploration of the solution space by evaluating  {subsets among all the possible combinations of} candidate deployments. However, we introduce a parameter to control the exploration level  {of the combinations set};
\item we exploit the linear constraints introduced in the previous section to limit the computational complexity of \textsc{PLATEA}.
\end{enumerate}

Alg.~\ref{alg:platea} reports the high-level pseudo-code of \textsc{PLATEA}.  {The source code of the algorithm is also available for download} \cite{platea}.  {More in detail, we design \textsc{PLATEA} by assuming that two distinct} frequencies $f1$ and $f2$ are exploited, with $f1$ targeting throughput maximization and $f2$ targeting coverage maximization.\footnote{In our case, $f1$ is a mid-band frequency, while $f2$ is a sub-GHz frequency. These two sets of frequencies are the ones currently in use by 5G, while the exploitation of frequencies in the mm-Wave band is still at the early stage in many countries in the world. However, \textsc{PLATEA}  {may} be easily generalized also to the case in which three types of frequencies (i.e., sub-GHz, mid-band, mm-Waves) are employed. We leave this aspect as future work.} In order to ease the presented pseudo-codes, we adopt the following guidelines: \textit{i}) the input parameters and sets defined in Tab.~\ref{tab:notation}  {of Appendix~\ref{app:notation}} are assumed to be available through global variables, and \textit{ii}) the subscripts appearing in the parameters/variables are hindered. The algorithm then produces as output the selected deployment $y$, the pixel to \ac{gNB} association $x$, the power density variables $P^{\text{ADD-TS}}$, $P^{\text{ADD-NOTS}}$, the exclusion zone variable $w$ and the total installation costs $C^{\text{TOT}}$ for the selected deployment.

\begin{algorithm}[t]
\small
\caption{Pseudo-Code of the \textsc{select\_best\_set\_f1} function}
\label{alg:selectbestsetf1}
	\textbf{Input:} num\_f1 deployed \acp{gNB}  with frequency $f1$\\  
	\textbf{Output:} flag\_end flag with installation status (false = installation successful, true = installation unsuccessful), temporary variables x\_best, y\_best, pd\_best
       \begin{algorithmic}[1]	
	\State{best\_obj=inf;}
	\State{flag\_end=true;}
	\State{sites\_f1\_comb=\textsc{extract\_sites}(num\_f1,f1); \texttt{// Extraction of the candidate deployments with frequency $f1$}}
	\For{sites\_f1 in sites\_f1\_comb}
		\State{[x\_temp y\_temp pd\_temp]=\textsc{initialize}();}
		\State{[flag\_check, pd\_temp]=\textsc{install\_check}(sites\_f1, y\_temp, pd\_temp); \texttt{// Based on Eq.~(\ref{eq:lower_bound_exclusion}), (\ref{eq:upper_bound_exclusion}), (\ref{eq:lin1_res})-(\ref{eq:indicator_parameter}) with frequency $f1$}}
		\If{flag\_check==true}
			\State{flag\_end=false;}
			\State{[x\_temp]=\textsc{associate\_pixels}(y\_temp, x\_temp); \texttt{// Based on Eq.~(\ref{eq:cov_pixel}),(\ref{eq:maximum_number_of_cells}),(\ref{eq:sir_aux_1})-(\ref{eq:sir_linear}) with frequency $f1$}}
			\State{temp\_obj=\textsc{compute\_obj}(x\_temp, y\_temp); \texttt{// Based on Eq.~(\ref{tot_cost_bs_installed_computation}), (\ref{eq:tot_obj}) with frequency $f1$}}
			\If{temp\_obj $<$ best\_obj}
				\State{best\_obj=temp\_obj;}
				\State{x\_best=x\_temp;}
				\State{y\_best=y\_temp;}
				\State{pd\_best=pd\_temp;}
			\EndIf
		\EndIf
	\EndFor
	\end{algorithmic}
\end{algorithm}

We then describe the operations performed by \textsc{PLATEA}. Initially, the maximum number of installable \ac{gNB} is computed (lines 2-3). In the following, all the variables are initialized to zero values by the \textsc{initial\_sol} function (line 5). \textsc{PLATEA} then iterates over the possible candidate deployments with frequency $f1$ (lines 7-31). In particular, the \textsc{select\_best\_set\_f1} function in line 8 retrieves the best solution found for each number of installable $f1$ \acp{gNB}, starting from one up to the maximum number (line 7). 

In the following, we provide more details about the \textsc{select\_best\_set\_f1} function, which is expanded in Alg.~\ref{alg:selectbestsetf1}. The function requires as input the number of targeted \acp{gNB} to be installed, denoted as \texttt{num\_f1}. Then, the function produces as output a flag (indicating if a feasible deployment has been found), as well as temporary variables storing the current set of installed $f1$ \acp{gNB}, the current pixel to \ac{gNB} association, and the current power density over the set of pixels. After  {initializing} the routine variables (line 1-2), the function retrieves the possible  {combinations} of $f1$ \acp{gNB}, by running the \textsc{extract\_sites} routine (line 3). Since enumerating all the possible  {combinations} is a challenging step in terms of computational requirements, we control the amount of generated  {combinations} by assuming that up to \texttt{num\_f1} candidate deployments are randomly generated. Intuitively, when \texttt{num\_f1} is low, it is not meaningful to explore the whole space of  {combinations}, since the number of served pixels will be in any case rather limited. On the other hand, we consider more  {combinations} as \texttt{num\_f1} increases, since the impact on service coverage may be not negligible in this case. 

In the following (lines 4-18), the \textsc{select\_best\_set\_f1} function iterates over the set of selected  {combinations}. In particular, the constraints about power density limits in Eq.~(\ref{eq:lower_bound_exclusion}), (\ref{eq:upper_bound_exclusion}), (\ref{eq:lin1_res})-(\ref{eq:power_density_limit_gen}), minimum distance from sensitive places in Eq.~(\ref{eq:min_distance}) and site constraints in Eq.~(\ref{eq:max_equipment})-(\ref{eq:indicator_parameter}) are checked over frequency $f1$. If the previous constraints are all met, the pixels are associated to the installed \acp{gNB} (line 9) and the objective function is evaluated (line 10). More in depth, the association of pixels in line 9 is performed by sequentially analyzing the set of installed \acp{gNB} and by associating each pixel while ensuring the 5G coverage and service constraints of Eq.~(\ref{eq:cov_pixel}),(\ref{eq:maximum_number_of_cells}),(\ref{eq:sir_aux_1})-(\ref{eq:sir_linear}) with frequency $f1$. Clearly, if the previous constraints are not met, the current pixel is not associated to the \ac{gNB} under consideration. In addition, the computation of the objective function in line 10 exploits constraints Eq.~(\ref{tot_cost_bs_installed_computation}),(\ref{eq:tot_obj}) with frequency $f1$. Eventually, the best solution is updated in lines 11-16.

When \textsc{select\_best\_set\_f1} is terminated, \textsc{PLATEA} performs lines 9-31 of Alg.~\ref{alg:platea}. In particular, the algorithm iterates over the candidate deployments on frequency $f2$ (line 10). Clearly, this step is performed only if a feasible candidate deployment over frequency $f1$ has been found (line 11). In the following, the algorithm generates \texttt{num\_f2}  {combinations} of $f2$ \acp{gNB} (line 12), and then iterates over each candidate deployment (lines 13-27) in order to verify the constraints (line 14) and eventually to perform the \ac{gNB}-pixel association (lines 15-16). The functions used in these steps are exactly the same adopted in Alg.~\ref{alg:selectbestsetf1}, except from the adopted  {frequency}, which  {is now set} to $f2$. In the following, the objective function is evaluated (line 17), and the best solution is eventually updated (lines 18-21). The algorithm then stops evaluating further deployments if all the pixels have been served  (line 22-24). Clearly, when passing between the evaluation of one deployment to the following one, the changes operated on the temporary variables are reverted to the previous state (lines 25,30). 

 {Finally, we compute the computational complexity of \textsc{PLATEA}. The detailed breakdown of the complexity for the single functions is reported in Appendix.~\ref{app:complexity}, while here we report the salient features. In brief, the} overall complexity of \textsc{PLATEA} is in the order of $\mathcal{O}(|\mathcal{P}|\times |\mathcal{L}|^5 \times |\mathcal{F}|)$.

\section{Scenario Description}
\label{sec:scenario}

\begin{figure}[t]
	\centering
	\includegraphics[width=7.5cm]{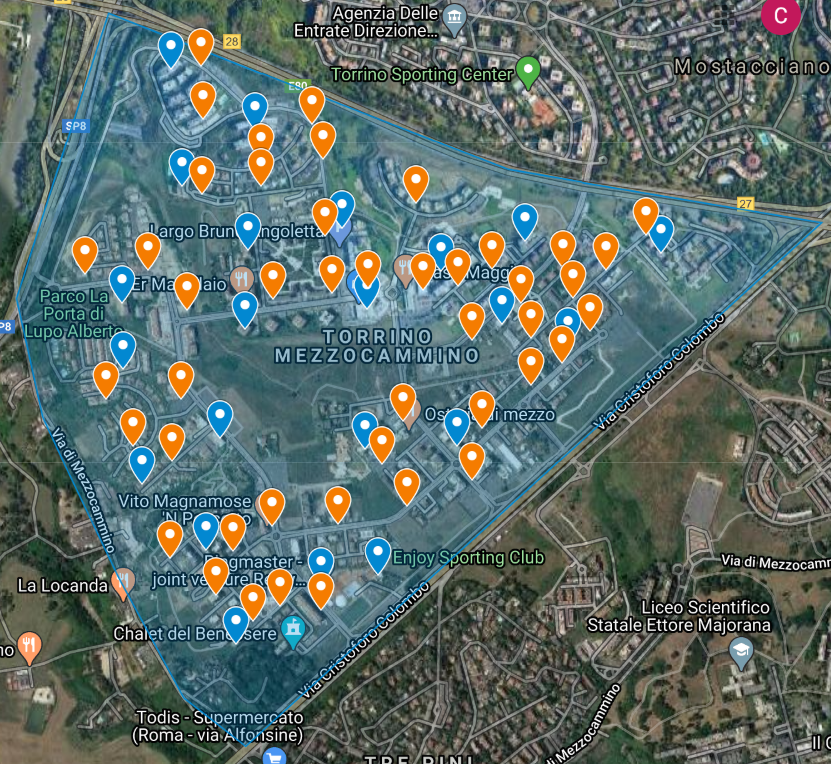}
	\caption{\ac{TMC} map with candidate locations for $f1$ \acp{gNB} (orange pins) and for $f2$ \acp{gNB} (blue pins).}
	\label{fig:torrino_map}
\end{figure}

We consider as reference scenario the \ac{TMC} neighborhood in Rome, Italy. The area under consideration, spanning over $2.47$~[km$^2$], is actually populated by more than 10000 inhabitants. We select the \ac{TMC} neighborhood due to the following reasons: \textit{i}) \ac{TMC} includes residential areas and sensitive places (i.e., public parks, schools, churches, recreation centers); therefore, its territory is subject to very stringent \ac{EMF} regulations (i.e., \textit{R6} regulation of Tab.~\ref{tab:regulation_comparison}), \textit{ii}) the terrain is almost plain, i.e., there are not steep hills and/or large obstacles (apart from the buildings), which would otherwise affect the propagation conditions, \textit{iii}) 5G coverage is not actually provided in the neighborhood, \textit{iv}) pre-5G base stations are installed only outside the neighborhood, \textit{v}) background information about pre-5G coverage and \ac{QoS} levels experienced in \ac{TMC} is already available in \cite{chiaraviglio2018not}.

We then focus on the set of frequencies $\mathcal{F}$ that are employed by 5G \acp{gNB}. More in detail, we consider the exploitation of two distinct frequencies, namely $f1=3700$~[MHz] and $f2=700$~[MHz]. Both $f1$  and $f2$ have been recently auctioned to 5G operators in Italy \cite{mise}. Therefore, we expect that both $f1$ and $f2$ will be used by 5G equipment in the forthcoming years.\footnote{Apart from $f1$ and $f2$, the auction of 5G frequencies in Italy included also a band at 26~[GHz] (i.e., close to mm-Waves) \cite{mise}. However, the 26~[GHz] frequency is intentionally left apart from this paper, due to the following reasons: \textit{i}) at present time, it is unclear at which extent 5G devices operating at 26~[GHz] will be installed over the territory, and \textit{ii}) there are not commercial \acp{gNB} operating at 26~[GHz] currently installed in Italy.} In this work, we assume that $f1$ is exploited by micro \acp{gNB}  {to provide hot-spot capacity}, while the $f2$ is employed by macro \acp{gNB}  {mainly for coverage}.  {As a consequence, we differentiate $f1$ and $f2$ in terms of candidate sites, coverage distance, \ac{EMF} values, 5G performance indicators and costs parameters.}

In the following, we provide more details about the area under consideration, the set of pixels, the set of candidate \acp{gNB} and the set of sensitive places. To this aim, Fig.~\ref{fig:torrino_map} reports: \textit{i}) the \ac{TMC} neighborhood  (transparent blue area), \textit{ii}) the set of locations $\mathcal{L}$ that can host 5G \acp{gNB} (i.e., the union of blue and orange pins), \textit{iii}) the subset of locations that can host $f1$ \acp{gNB} (orange pins), and \textit{iv}) the subset of locations that can host $f2$ \acp{gNB} (blue pins). The selection of locations in \textit{iii}) and \textit{iv}) is driven by the following principles: \textit{a}) installation of $f2$ \acp{gNB} mainly on top of buildings, in order to maximize the coverage over the territory, \textit{b}) installation of $f1$ \acp{gNB} close to the zones where capacity is needed (along the roads and in proximity to the residential buildings). As a result, a total of $|\mathcal{L}|=69$ candidate locations are taken into account in this work {, thus leading to a total number of candidate deployments equal to $2^{69}$}.\footnote{ {Clearly, the introduction of different types of gNBs operating on different frequencies greatly complicates the considered problem, mainly because the number of candidate locations is increased w.r.t. a scenario composed of homogeneous gNBs operating at the same frequency.}} 

\begin{figure}[t]
	\centering
	\includegraphics[width=7.5cm]{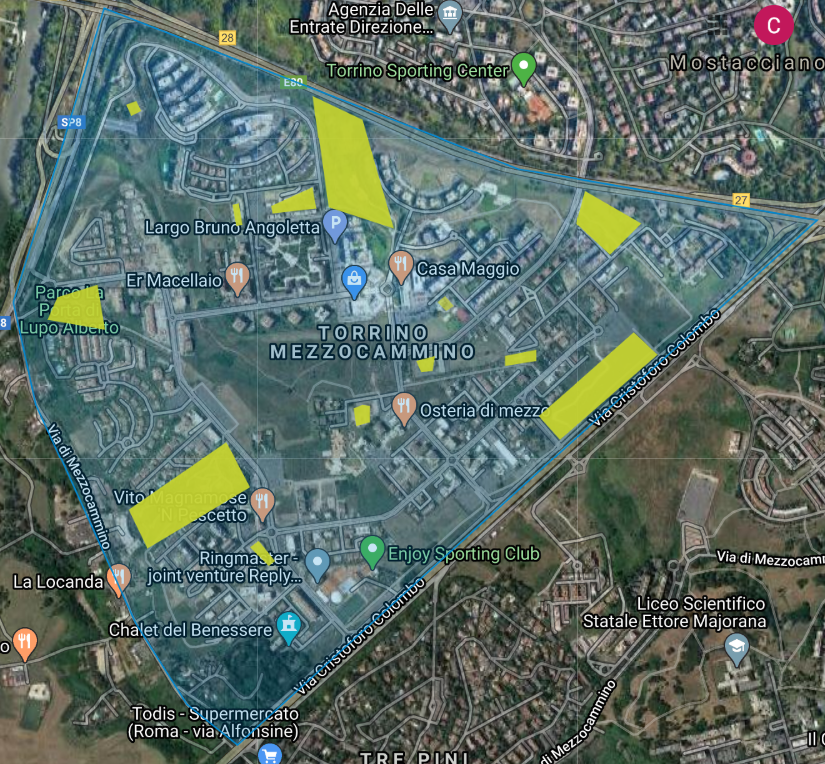}
	\caption{\ac{TMC} map with sensitive places highlighted in yellow.}
	\label{fig:torrino_map_areas}
\end{figure}

In the following, we move our attention to the set of pixels $\mathcal{P}$. More in detail, we assume a pixel tessellation over the \ac{TMC} area, with a pixel granularity equal to $10 \times 10$~[m$^2$]. The total number of pixels $|\mathcal{P}|$ is then equal to 24318. Focusing then on the sensitive places, Fig.~\ref{fig:torrino_map_areas} highlights in yellow the areas hosting public parks, schools and/or churches. By adopting \textit{R6}) regulation from Tab.~\ref{tab:regulation_comparison}, the installation of \acp{gNB} is prohibited within a minimum distance of $D^{\text{MIN}}=100$~[m] from the external perimeter of these sensitive places.  

\begin{table}[t]
    \caption{Setting of the frequency-related parameters.}
    \label{tab:service_power_density_param}
    \scriptsize
    \centering
    \begin{tabular}{|>{\columncolor{DarkLinen}}p{0.10cm}|c|c|c|}
\hline
\rowcolor{Coral} & & & \\[-0.75em]
\rowcolor{Coral} & \textbf{Parameter} & \textbf{$f1=3.7$~[GHz]} & \textbf{$f2=700$~[MHz]}\\[0.25em]
\hline
\rowcolor{Linen} \cellcolor{DarkLinen} & & \multicolumn{2}{c|}{} \\[-0.75em]
\rowcolor{Linen} \cellcolor{DarkLinen} &  $I^{\text{FREQ}}_{(l,f)}$ & \multicolumn{2}{c|}{\begin{minipage}{5.5cm}Based on \ac{TMC} scenario with $N^{\text{MAX}}=1$.\end{minipage}} \\[0.25em]
\rowcolor{Linen} \cellcolor{DarkLinen} & & \multicolumn{2}{c|}{} \\[-0.75em]
\hline
\rowcolor{White} \cellcolor{DarkLinen} & & & \\[-0.75em]
 & \ac{gNB} height & 10~[m] (Pole mounted) & 25~[m] (Roof-top mounted) \\[-0.05em]
& & & \\[-0.75em]
\rowcolor{Linen} \cellcolor{DarkLinen} & &  \multicolumn{2}{c|}{} \\[-0.75em]
\rowcolor{Linen} \cellcolor{DarkLinen} & Pixel height & \multicolumn{2}{c|}{1.5~[m] (std. evaluation height)} \\[-0.05em]
& &  \multicolumn{2}{c|}{} \\[-0.75em]
\multirow{-4}{*}{\begin{sideways}\textbf{Distance}\end{sideways}} & $D_{(p,l,f)}$ & \multicolumn{2}{c|}{\begin{minipage}{5.5cm}Based on \ac{TMC} scenario and \ac{gNB}/pixel heights.\end{minipage}} \\[-0.05em]
& &  \multicolumn{2}{c|}{} \\[-0.75em]
\hline
\rowcolor{Linen} \cellcolor{DarkLinen} &  & & \\[-0.75em]
\rowcolor{Linen} \cellcolor{DarkLinen} & $O^{\text{MAX}}_{(l,f)}$ & 200~[W] $\forall l \in \mathcal{L}$ \cite{air6488} & \begin{minipage}{2.5cm}65~[W] $\forall l \in \mathcal{L}$ (assumed to be in line with pre-5G technologies \cite{itutk70})\end{minipage} \\[0.15em]
\rowcolor{Linen} \cellcolor{DarkLinen} &  & & \\[-0.75em]
\rowcolor{White} \cellcolor{DarkLinen} & &  \multicolumn{2}{c|}{} \\[-0.75em]
& $\eta^{\text{GAIN}}_f$ & \multicolumn{2}{c|}{15~[dB] \cite{itutk70}} \\[0.25em]
\rowcolor{Linen} \cellcolor{DarkLinen} & &  \multicolumn{2}{c|}{} \\[-0.75em]
\rowcolor{Linen} \cellcolor{DarkLinen} & $\eta^{\text{LOSS}}_f$ & \multicolumn{2}{c|}{2.32~[dB] \cite{itutk70}} \\[0.25em]
& & & \\[-0.75em]
& $R^{\text{STAT}}_{(l,f)}$ & 0.25 $\forall l \in \mathcal{L}$ \cite{ThoFurTor:17} & 1 $\forall l \in \mathcal{L}$ (max. value) \\[0.25em]
\rowcolor{Linen} \cellcolor{DarkLinen} & & \multicolumn{2}{c|}{} \\[-0.75em]
\rowcolor{Linen} \cellcolor{DarkLinen}  & $R^{\text{TIME}}_{(l,f)}$ & \multicolumn{2}{c|}{0.3 $\forall l \in \mathcal{L}$ \cite{ecoscienza}} \\[0.25em]
& & & \\[-0.75em]
& Excl. Distance & 11~[m] \cite{5gimpact} & 5~[m] (Roof excl. zone) \\[0.25em]
\rowcolor{Linen} \cellcolor{DarkLinen}& & \multicolumn{2}{c|}{} \\[-0.75em]
\rowcolor{Linen} \cellcolor{DarkLinen}  & $I^{\text{ZONE}}_{(p,l,f)}$ & \multicolumn{2}{c|}{\begin{minipage}{5.5cm}Based on the \ac{TMC} scenario and the exclusion distances.\end{minipage}} \\[0.25em]
\rowcolor{Linen} \cellcolor{DarkLinen} & & \multicolumn{2}{c|}{} \\[-0.75em]
\rowcolor{White} \cellcolor{DarkLinen}  & & \multicolumn{2}{c|}{} \\[-0.75em]
& $P^{\text{BASE}}_{(p,f)}$ & \multicolumn{2}{c|}{\begin{minipage}{5.5cm}Based on real \ac{EMF} measurements in \ac{TMC} scenario.\end{minipage}} \\[0.25em]
& & \multicolumn{2}{c|}{} \\[-0.75em]
\rowcolor{Linen} \cellcolor{DarkLinen}& & \multicolumn{2}{c|}{} \\[-0.75em]
\rowcolor{Linen} \cellcolor{DarkLinen}\multirow{-15}{*}{\begin{sideways}\textbf{5G EMF}\end{sideways}} & $P^{\text{ADD}}_{(p,l,f)}$ & \multicolumn{2}{c|}{\begin{minipage}{5.5cm}Based on point source model \cite{itutk70} over \ac{TMC} scenario and \ac{EMF} parameters.\end{minipage}} \\[0.25em]
\rowcolor{Linen} \cellcolor{DarkLinen}& & \multicolumn{2}{c|}{} \\[-0.75em]

\rowcolor{White} \cellcolor{DarkLinen} & & \multicolumn{2}{c|}{} \\[-0.75em]
 & $L_f$ & \multicolumn{2}{c|}{6~[V/m] (\textit{R6} of Tab.~\ref{tab:regulation_comparison})} \\[0.25em]

\hline
\rowcolor{Linen} \cellcolor{DarkLinen} & & & \\[-0.75em]
\rowcolor{Linen} \cellcolor{DarkLinen}& $D^{\text{MAX}}_f$ & 200~[m]~\cite{oughton2019assessing} & 900~[m]~\cite{chiaraviglio2018not,patwary2020potential} \\[0.25em]
& & & \\[-0.75em]
& $\gamma_f$ & 3.19 \cite{sun2016investigation} & 3 \cite{sun2016investigation} \\[0.25em]
\rowcolor{Linen} \cellcolor{DarkLinen}& & & \\[-0.75em]
\rowcolor{Linen} \cellcolor{DarkLinen}& $\sigma^{\text{SHAD}}_f$ & 8.2~[dB] \cite{sun2016investigation} & 6.8~[dB] \cite{sun2016investigation} \\[0.25em]
& & \multicolumn{2}{c|}{} \\[-0.75em]
& $z_{(l,p,j,f)}$ & \multicolumn{2}{c|}{\begin{minipage}{5.5cm}Log normal random variable \cite{marzetta}.\end{minipage}} \\[0.25em]
\rowcolor{Linen} \cellcolor{DarkLinen}& & \multicolumn{2}{c|}{} \\[-0.75em]
\rowcolor{Linen} \cellcolor{DarkLinen}& $\beta_{(l,p,j,f)}$ & \multicolumn{2}{c|}{\begin{minipage}{5.5cm}Marzetta model \cite{marzetta} based on \ac{TMC} scenario, $\gamma_f$ and $z_{(l,p,j,f)}$.\end{minipage}}\\[0.25em]
& & \multicolumn{2}{c|}{}\\[-0.75em]
& $N^{\text{OFDM}}_f$ & \multicolumn{2}{c|}{14 \cite{etsitr}} \\
& & \multicolumn{2}{c|}{} \\[-0.75em]
\rowcolor{Linen} \cellcolor{DarkLinen} & & \multicolumn{2}{c|}{} \\[-0.75em]
\rowcolor{Linen} \cellcolor{DarkLinen}& $N^{\text{OFDM-PILOT}}_f$ & \multicolumn{2}{c|}{3 \cite{marzetta}} \\[0.25em]
& & \multicolumn{2}{c|}{} \\[-0.75em]
& $\tau^{\text{COHERENCE}}_f$ & \multicolumn{2}{c|}{500~[$\mu$s] \cite{marzetta}}  \\[0.2em]
\rowcolor{Linen} \cellcolor{DarkLinen} & & &  \\[-0.75em]
\rowcolor{Linen} \cellcolor{DarkLinen}& $\delta_f$ & 2.3~[$\mu$s] \cite{etsitr,5gnr} & 4.7~[$\mu$s] \cite{etsitr,5gnr} \\[0.25em]
& & & \\[-0.75em]
& $\Delta_f$ & 30~[kHz] \cite{etsitr} & 15~[kHz] \cite{etsitr} \\[0.25em]
\rowcolor{Linen} \cellcolor{DarkLinen}& & &  \\[-0.75em]
\rowcolor{Linen} \cellcolor{DarkLinen}& $B_f$ & 80~[MHz] \cite{mise}& 20~[MHz] \cite{mise} \\[0.25em]
& & \multicolumn{2}{c|}{}  \\[-0.75em]
\multirow{-15}{*}{\begin{sideways}\textbf{5G Performance}\end{sideways}} & $\epsilon^{\text{F-REUSE}}_f$ & \multicolumn{2}{c|}{1 ( {unity frequency reuse})}  \\[0.25em]
\rowcolor{Linen} \cellcolor{DarkLinen}& & &  \\[-0.75em]
\rowcolor{Linen} \cellcolor{DarkLinen}& $S^{\text{MIN}}_{f}$ & \begin{minipage}{2cm}Set to ensure 30~[Mbps] of min. throughput.\end{minipage} & - \\[0.25em]
\hline
& & &  \\[-0.75em]
& $C^{\text{SITE}}_{(l,f)}$ & 14852~[\euro] $\forall l \in \mathcal{L}$ \cite{oughton2019assessing} & 20101~[\euro] \cite{oughton2019assessing} $\forall l \in \mathcal{L}$  \\[0.25em]
\rowcolor{Linen} \cellcolor{DarkLinen}& & &  \\[-0.75em]
\rowcolor{Linen} \cellcolor{DarkLinen} & $C^{\text{EQUIP}}_{f}$ & 2791~[\euro] \cite{oughton2019assessing} & 45673~[\euro] \cite{oughton2019assessing} \\[0.25em]
& & &  \\[-0.75em]
\multirow{-4}{*}{\begin{sideways}\textbf{Costs}\end{sideways}} & $\alpha_{(l,f)}$ & \begin{minipage}{2cm}$[10^1-10^7]$~[\euro] $\forall l \in \mathcal{L}$\end{minipage} & \begin{minipage}{2cm}$[10^1-10^4]$~[\euro] $\forall l \in \mathcal{L}$\end{minipage} \\[-0.75em]
& & &  \\
\hline
\end{tabular}
\vspace{-4mm}
\end{table}

We then analyze the setting of the remaining frequency-dependent parameters that are required as input. To this aim, Tab.~\ref{tab:service_power_density_param} reports: \textit{i}) $I^{\text{FREQ}}_{(l,f)}$ parameter (obtained from Fig.~\ref{fig:torrino_map} by assuming $N^{\text{MAX}}=1$), \textit{ii}) distance-based parameters, \textit{iii}) 5G \ac{EMF} parameters, \textit{iv}) 5G performance parameters, and \textit{v}) 5G costs parameters. We now focus on the setting of the key parameters, while we refer the reader to the references reported in Tab.~\ref{tab:service_power_density_param} for more information about the setting of each single parameter. More in depth, we assume that the maximum power $O^{\text{MAX}}_{(l,f)}$ that is radiated by a \ac{gNB} operating over $f1$ is actually higher than the one radiated by a \ac{gNB} operating over $f2$. Although this setting may appear quite counter-intuitive at a first glance, since a micro \ac{gNB} is expected to radiate less power than a macro \ac{gNB}, we remind that $O^{\text{MAX}}_{(l,f)}$ refers to the maximum power, which can clearly differ w.r.t. the actual one that is received over the territory. When considering micro \acp{gNB}, $O^{\text{MAX}}_{(l,f)}$ is split across the radiating elements, and thus the actual power that is received over the territory is clearly lower compared to the maximum one. This effect is taken into account when setting the statistical scaling factor $R^{\text{STAT}}_{(l,f)}$. More in detail, we assume a strong statistical scaling factor that is applied to micro \acp{gNB} operating on $f1$, while no statistical scaling factor is applied to macro \acp{gNB} operating on $f2$. This choice is also motivated by the different goals of two \ac{gNB} types, i.e., maximizing throughput for $f1$ (and hence large spatial power variability) vs. ensuring coverage for $f2$ (and hence less spatial power variability). Focusing then on the exclusion zones, we assume again two distinct values for $f1$ and $f2$, which are set in accordance to the minimum distance guaranteeing a \ac{EMF} level below the limit of 6~[V/m] for a single \ac{gNB}.\footnote{The \ac{EMF} is also computed in this case by applying the point source model of \cite{itutk70}.} Eventually, $P^{\text{BASE}}_{(p,f)}$ is retrieved from real measurements over the real scenario, with the methodology described in Appendix~\ref{app:emf_measurements}.  

 {Focusing then on the 5G performance parameters, the maximum coverage distance $D^{\text{MAX}}_f$ is set to 200~[m] and to 900~[m] for $f1$ and $f2$, respectively (in accordance with} \cite{oughton2019assessing,chiaraviglio2018not,patwary2020potential} {). More in depth, the setting of $D^{\text{MAX}}_f$ can be based on e.g., minimum signal strength and/or prediction of radio resource availability and/or average load per gNB, depending on the operator's needs.}  {Eventually, the path loss exponent $\gamma_f$ is tuned in accordance with} \cite{sun2016investigation} {. Clearly, this parameter can be easily varied to match different propagation conditions, e.g., urban, suburban, indoor,
etc. Moreover,} we assume that the minimum \ac{SIR} is set in order to ensure a minimum pixel throughput of 30~[Mbps] for $f1$.  In addition, we do not constrain the \ac{SIR} over $f2$, since the goal of \acp{gNB} operating over this frequency is mainly to provide coverage. Therefore, even low throughput values may be admitted for $f2$.\footnote{We remind that, in any case, pixels beyond the maximum distance coverage $D_f^{\text{MAX}}$ from a given \ac{gNB} can not be served by the \ac{gNB}.} Although these throughput settings may appear relatively loose at a first glance, we will show that the actual throughput levels experienced over the served pixels are not negligible and in line with the 5G service requirements \cite{3gppservice}. 

 {In the following step, we concetrate on the setting of the $\alpha_{(l,f)}$ parameters. We remind that such parameters are provided by the operator as input to our framework. However, since a precise estimation of $\alpha_{(l,f)}$ is beyond our goals, in this work we have performed a sensitivity analysis over huge ranges of $\alpha_{(l,f)}$, in order to consider the following cases: \textit{i}) total costs dominate over coverage revenues, \textit{ii}) revenues and costs are balanced, and \textit{iii}) revenues are much higher than the costs. In this way, we are able to investigate the impact of $\alpha_{(l,f)}$ on the obtained planning on one side, and to provide an indications about the settings that ensure good coverage levels.} For the sake of simplicity, we impose the same weights applied for all the candidate $\acp{gNB}$ working at the same frequency.\footnote{In more complex scenarios, the values of $\alpha_{(l,f)}$ may be tuned for each location $l$, in order to prioritize locations that require huge amount of traffic (e.g., shopping malls, train stations, or airports) w.r.t. other ones. The evaluation of this aspect is left for future work.}  {Moreover, we impose a wider interval of $\alpha_{(l,f1)}$ values w.r.t. $\alpha_{(l,f2)}$ because in this way we test the case in which serving a pixel with a micro gNB has a huge gain w.r.t. to a coverage provided by a macro \ac{gNB}.}

 {Finally}, we provide the setting for the remaining parameters, namely $Z_0$ and $N^{\text{SER}}$. In particular, we set $Z_0=377$~[Ohm] (in accordance to \cite{itutk70}). In addition, we impose $N^{\text{SER}}=2$. In this way, we consider a conservative case in which each pixel is served by at most two \acp{gNB}.\footnote{The evaluation of $N^{\text{SER}}$ values greater than 2 is left for future work.}









\begin{table*}[t]
    \caption{Breakdown of the evaluation metrics.}
    \label{tab:metrics}
    \scriptsize
    \centering
    \begin{tabular}{|c|c|p{4.5cm}|}
\hline
\rowcolor{Coral}  & & \\[-0.75em]
\rowcolor{Coral} \textbf{Metric} & \textbf{Notation/Expression} & \textbf{Reference Equations} \\
\hline
\rowcolor{Linen} & & \\[-0.75em]
\rowcolor{Linen} Total Installation Costs & $C^{\text{TOT}}$ & Eq.~(\ref{tot_cost_bs_installed_computation}) (total costs computation). \\[0.25em]
\hline
& & \\[-0.75em]
Number of Installations & $N_{f1}=\sum_{l \in \mathcal{L}} y_{(l,f1)}$, $\quad N_{f2}=\sum_{l \in \mathcal{L}} y_{(l,f2)}$ & - \\[0.25em]
& & \\[-0.75em]
\hline
\rowcolor{Linen} & & \\[-0.75em]
\rowcolor{Linen} Served Pixels & $X^{\text{SERVED}}_{f1}=\sum_{p \in \mathcal{P}} \sum_{l \in \mathcal{L}} x_{(p,l,f1)}$, $\quad X^{\text{SERVED}}_{f2}=\sum_{p \in \mathcal{P}} \sum_{l \in \mathcal{L}} x_{(p,l,f2)}$ & - \\[0.25em]
\hline
& & \\[-0.75em]
Unserved pixels [\%] & $X^{\text{NOT-SERVED}}=100 \cdot \frac{\sum_{p \in \mathcal{P} : \{\sum_{l \in \mathcal{L}} \sum_{f \in \mathcal{F}} x_{(p,l,f)} == 0\}} \left(1 - \sum_{l \in \mathcal{L}}\sum_{f \in \mathcal{F}}  x_{(p,l,f)}\right)}{|\mathcal{P}|}$ & - \\[0.25em]
& & \\[-0.75em]
\hline
\rowcolor{Linen} & & \\[-0.75em]
\rowcolor{Linen} \multirow{3}{*}{Pixel Throughput} & \multirow{3}{*}{\underline{$T_{(p,f)}=\sum_{l \in \mathcal{L}} \frac{B_f \Gamma_f}{\epsilon^{\text{F-REUSE}}_f} \log_2\left(1 + S_{(p,l,f)}\right)\cdot x_{(p,l,f)}$, $\quad$ $T_p=\sum_{f \in \mathcal{F}} T_{(p,f)}$}} & Eq.~(\ref{eq:sir_not_linear}) (SIR computation on left hand side), Eq.~(\ref{eq:gamma_term})  ($\Gamma_f$ computation), Eq.~(\ref{eq:capacity}) (throughput computation).  \\[0.25em]
\hline
& & \\[-0.75em]
Average Pixel Throughput & \underline{$T^{\text{AVG}}_{f}=\frac{\sum_{p \in \mathcal{P}} T_{(p,f)}}{X^{\text{SERVED}}_f}$}, $\quad T^{\text{AVG}}=\frac{\sum_{p \in \mathcal{P}}T_p}{|\mathcal{P}|\cdot(1 - X^{\text{NOT-SERVED}}/100)}$ & See computation of $T_p$\underline{,  $T_{(p,f)}$} and $X^{\text{NOT-SERVED}}$. \\[0.25em]
\hline
\rowcolor{Linen} & & \\[-0.75em]
\rowcolor{Linen} \multirow{3}{*}{Pixel \ac{EMF}} & \multirow{3}{*}{$E_p=\sqrt{\sum_{f \in \mathcal{F}} \left[ P^{\text{BASE}}_{(p,f)} \cdot (1-w_p) + P^{\text{ADD-TS}}_{(p,f)}\right] \cdot Z_0}$} & Eq.~(\ref{eq:power_density_limit_res}) (total power density computation in the numerator on left-hand side), Eq.~(\ref{eq:tot_emf}) (electric field computation). \\[0.25em]
\hline
& & \\[-0.75em]
Average Pixel \ac{EMF} & $E^{\text{AVG}}=\sqrt{\frac{\sum_{p \in \mathcal{P}^{\text{RES}}}\sum_{f \in \mathcal{F}} \left[ P^{\text{BASE}}_{(p,f)} \cdot (1-w_p) + P^{\text{ADD-TS}}_{(p,f)}\right]}{|\mathcal{P^{\text{RES}}}|}\cdot Z_0}$ & See computation of $E_p$. \\[0.25em]
\hline
\end{tabular}
\end{table*}

\section{Results}
\label{sec:results}

We code \textsc{PLATEA} algorithm in MATLAB R2019b and we run it on a Dell PowerEdge R230 equipped with Intel Xeon E3-1230 v6 3.5~[GHz] processors and 64~[GB] of RAM. We then describe the following steps: \textit{i}) introduction of two reference algorithms as terms of comparison, \textit{ii}) definition of evaluation metrics, \textit{iii}) tuning of \textsc{PLATEA} parameters, \textit{iv}) comparison of \textsc{PLATEA} vs. the reference algorithms, \textit{v}) impact of planning parameters, \textit{vi}) impact of pre-5G exposure levels,  {and \textit{vii}) impact of frequency reuse scheme}. 

\textbf{Reference Algorithms.} 
 {Since the considered scenario is composed of thousands of pixels and dozens of candidate sites, it is not possible to solve the \textsc{OPTPLAN-5G} problem, which we remind belongs to the NP-Hard class. In order to introduce a term of comparison, we have implemented a simple algorithm based on an exhaustive search over all the possible micro and macro \ac{gNB} combinations. We refer the interested reader to Appendix~\ref{app:exhaustive_search} for more details. In brief, the brute-force solution is not practically feasible, mainly due to very large computational times (much higher than the relatively loose time constraints that may be allowed for solving planning problems). To overcome such issue, we have designed two sub-optimal (yet meaningful) algorithms in order to better position \textsc{PLATEA}. The two solutions, named \textsc{Evaluation Algorithm (EA)} and \textsc{Maximum Coverage Macro Algorithm (MCMA)} are detailed in Appendix~\ref{app:reference_algorithms}}. In brief, \textsc{EA} and \textsc{MCMA} evaluate the feasibility constraints of a random set of installed \acp{gNB}, without exploring the possible site  {combinations} (which are instead analyzed by \textsc{PLATEA}).  {In this way, we are able to compare \textsc{PLATEA} against two low-complexity solutions.} More in depth, \textsc{EA} explores a single possible deployment (which is generated from a fixed number of \acp{gNB}, passed as input to the algorithm). On the other hand, \textsc{MCMA} goes one step further, by selecting the set of macro \acp{gNB} operating on $f2$ maximizing the service coverage, given an integer number of $f1$ \acp{gNB} that have to be installed over the territory. 

\textbf{Evaluation Metrics.} We then formally introduce the metrics to evaluate the performance of \textsc{PLATEA}, \textsc{EA} and \textsc{MCMA}. To this aim, Tab.~\ref{tab:metrics} reports: \textit{i}) total installation costs 
$C^{\text{TOT}}$, \textit{ii}) number $N_{f1}$ ($N_{f2}$) of $f1$ ($f2$) \acp{gNB} installations, \textit{iii}) number of pixels $X^{\text{SERVED}}_{f1}$ ($X^{\text{SERVED}}_{f2}$) served by $f1$ ($f2$) \acp{gNB}, \textit{iv}) percentage of unserved pixels $X^{\text{NOT-SERVED}}$, \textit{v})  {per-frequency $T_{(p,f)}$ and total $T_p$} pixel throughput, \textit{vi}) average pixel throughput  {over each frequency $T^{\text{AVG}}_f$ and over all frequencies $T^{\text{AVG}}$, both of them} computed over the pixels that are served by \acp{gNB}, \textit{vii})  pixel \ac{EMF} $E_p$, \textit{viii}) average pixel \ac{EMF} $E^{\text{AVG}}$, computed over the whole set of pixels in the scenario. For each metric, the table reports the metric name, the mathematical notation, and the reference equation(s) used to compute the metric.





\textbf{Tuning of \textsc{PLATEA} Parameters.} We initially concentrate on the impact of the $\alpha_{(l,f)}$ terms that  {control the behavior of} the objective function in \textsc{PLATEA}.  {As already mentioned Sec.~\ref{sec:scenario}}, we explore a wide range of values for both $\alpha_{(l,f1)}$ and $\alpha_{(l,f2)}$ {, in order to test different conditions, e.g., costs minimization, balance between costs and revenues, revenues maximization}. In addition, we initially assume the dismission of legacy pre-5G Base Stations that radiate over \ac{TMC}, in order to evaluate the performance of \textsc{PLATEA} in a clean-slate condition. Therefore, we set $P^{\text{BASE}}_{(p,f)}=0 \quad \forall f \in \mathcal{F}, p \in \mathcal{P}$. We then run \textsc{PLATEA} over the selected ranges of $\alpha_{(l,f1)}$ and $\alpha_{(l,f2)}$, by picking values on logarithmic scales. Fig.~\ref{fig:installation_var} highlights the obtained results in terms of: \textit{i}) total installation costs $C^{\text{TOT}}$ (Fig.~\ref{fig:costs_performance}), \textit{ii}) number $N_{f1}$ of $f1$ \acp{gNB} (Fig.~\ref{fig:n_micro}), \textit{iii}) number $N_{f2}$ of $f2$ \acp{gNB} (Fig.~\ref{fig:n_macro}), \textit{iv}) percentage of not served pixels $X^{\text{NOT-SERVED}}$ (Fig.~\ref{fig:not_served}), \textit{v}) average pixel throughput $T^{\text{AVG}}$ (Fig~\ref{fig:capacity_avg}) and \textit{vi}) average electric field $E^{\text{AVG}}$ (Fig.~\ref{fig:emf_heatmap}). 

\begin{figure}[t]
\centering
\subfigure[{$C^{\text{TOT}}$~[M\euro]}]
{
    \includegraphics[width=42mm]{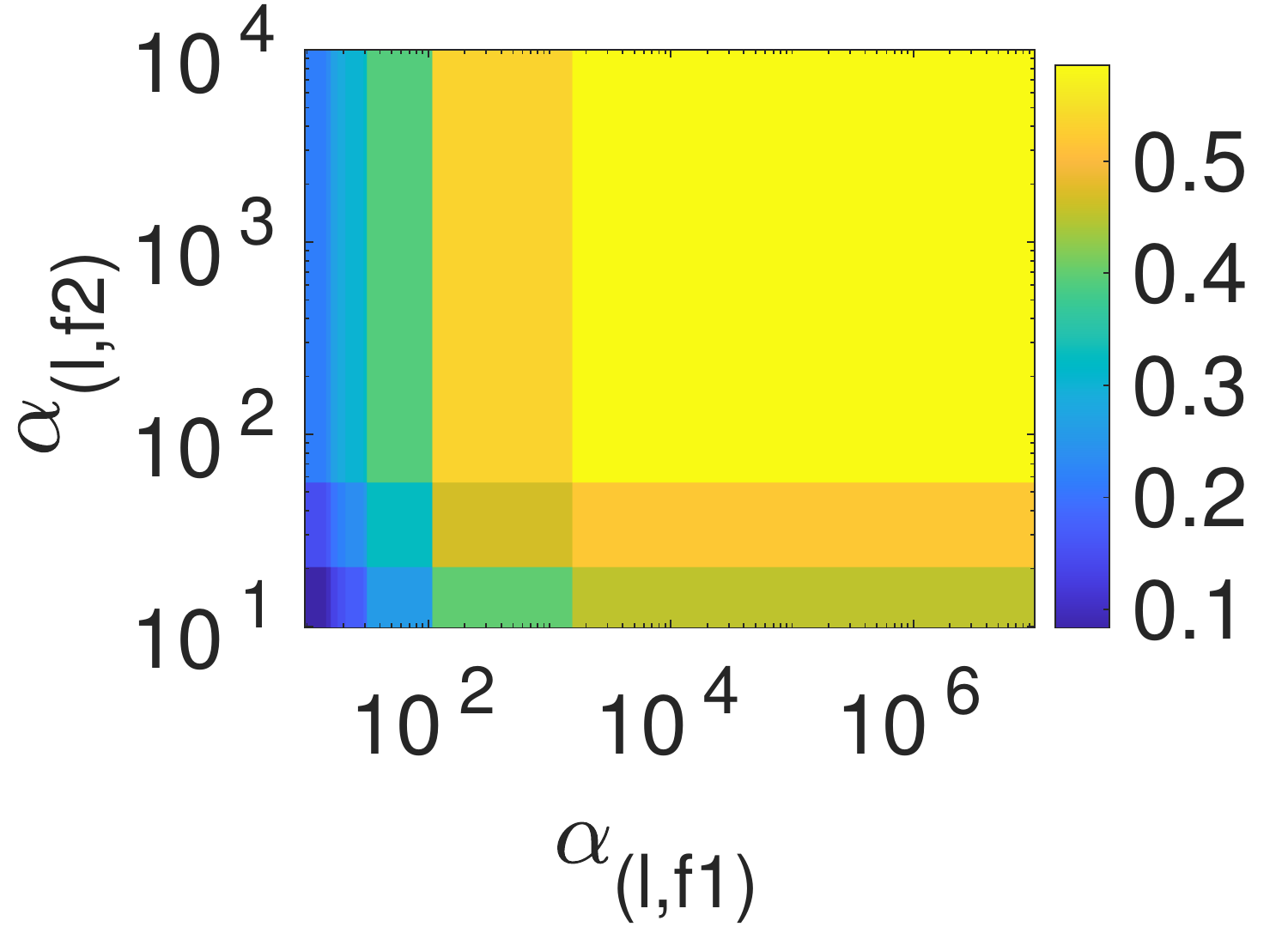}
    \label{fig:costs_performance}
}
\subfigure[$N_{f1}$]
{
	\includegraphics[width=40mm]{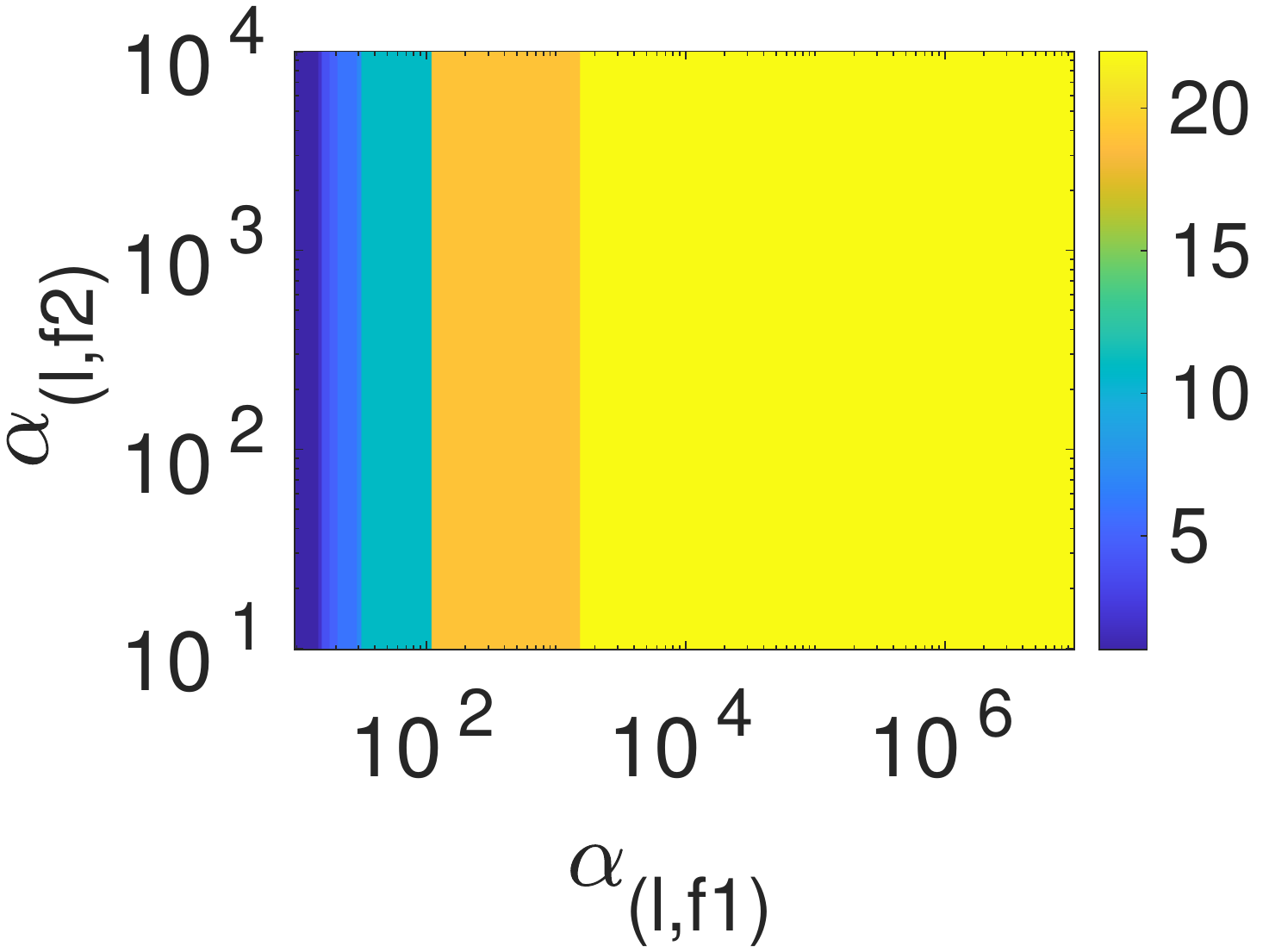}
	\label{fig:n_micro}
}

\subfigure[$N_{f2}$]
{
    \includegraphics[width=40mm]{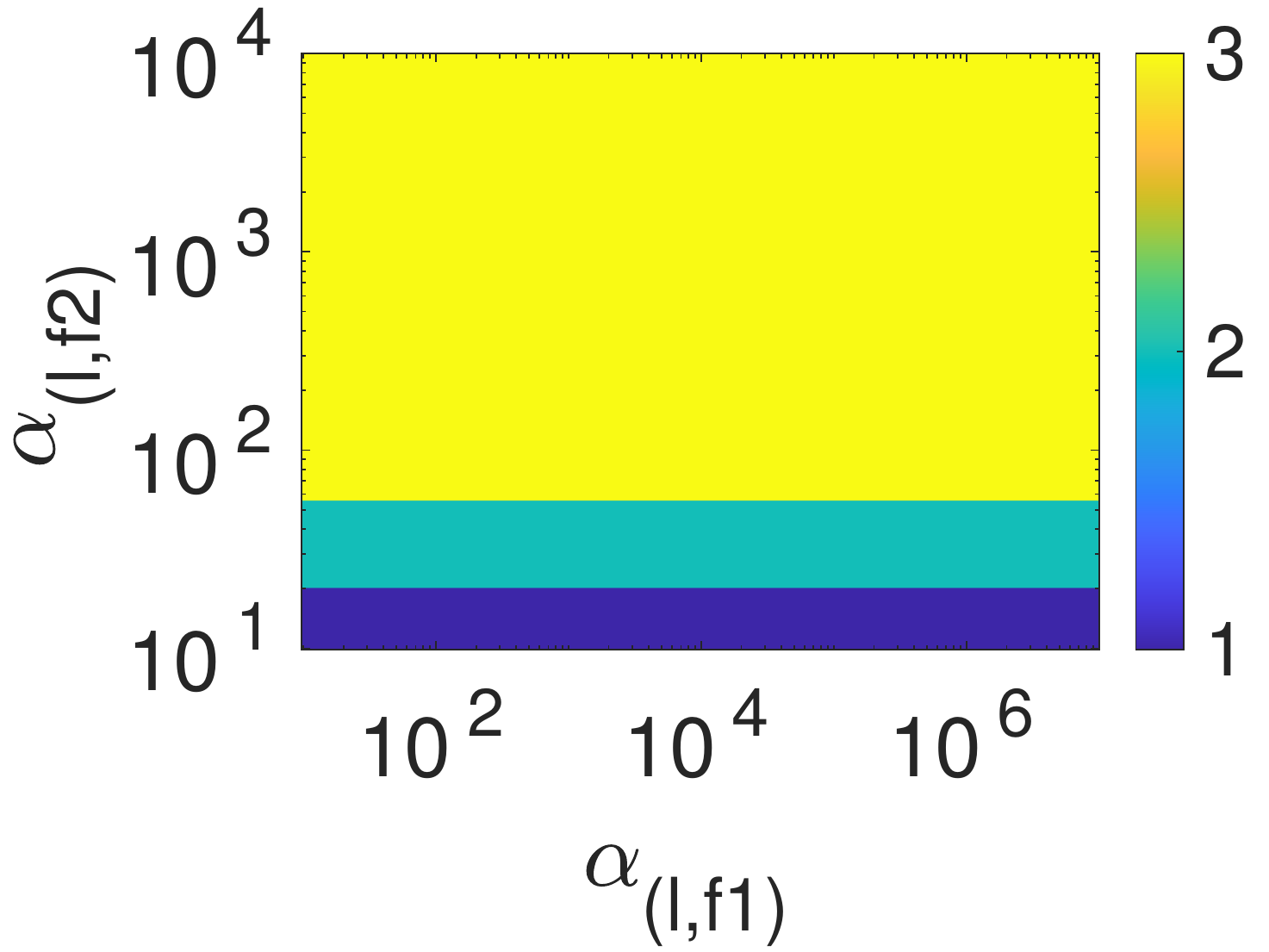}
    \label{fig:n_macro}
}
\subfigure[{$X^{\text{NOT-SERVED}}$~[\%]}]
{
	\includegraphics[width=42mm]{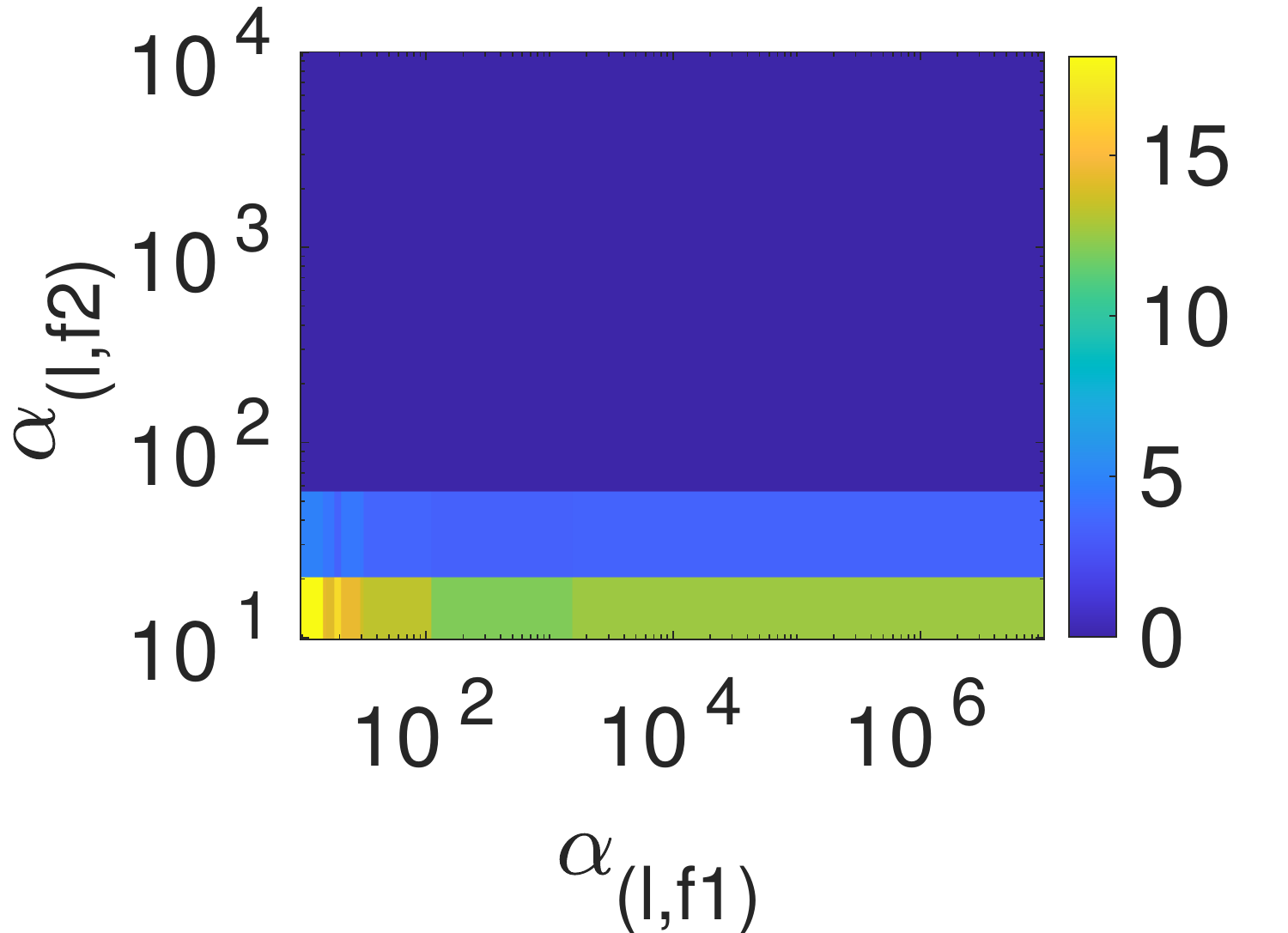}
	\label{fig:not_served}
}

\subfigure[{$T^{\text{AVG}}$~[Gbps]}]
{
	\includegraphics[width=42mm]{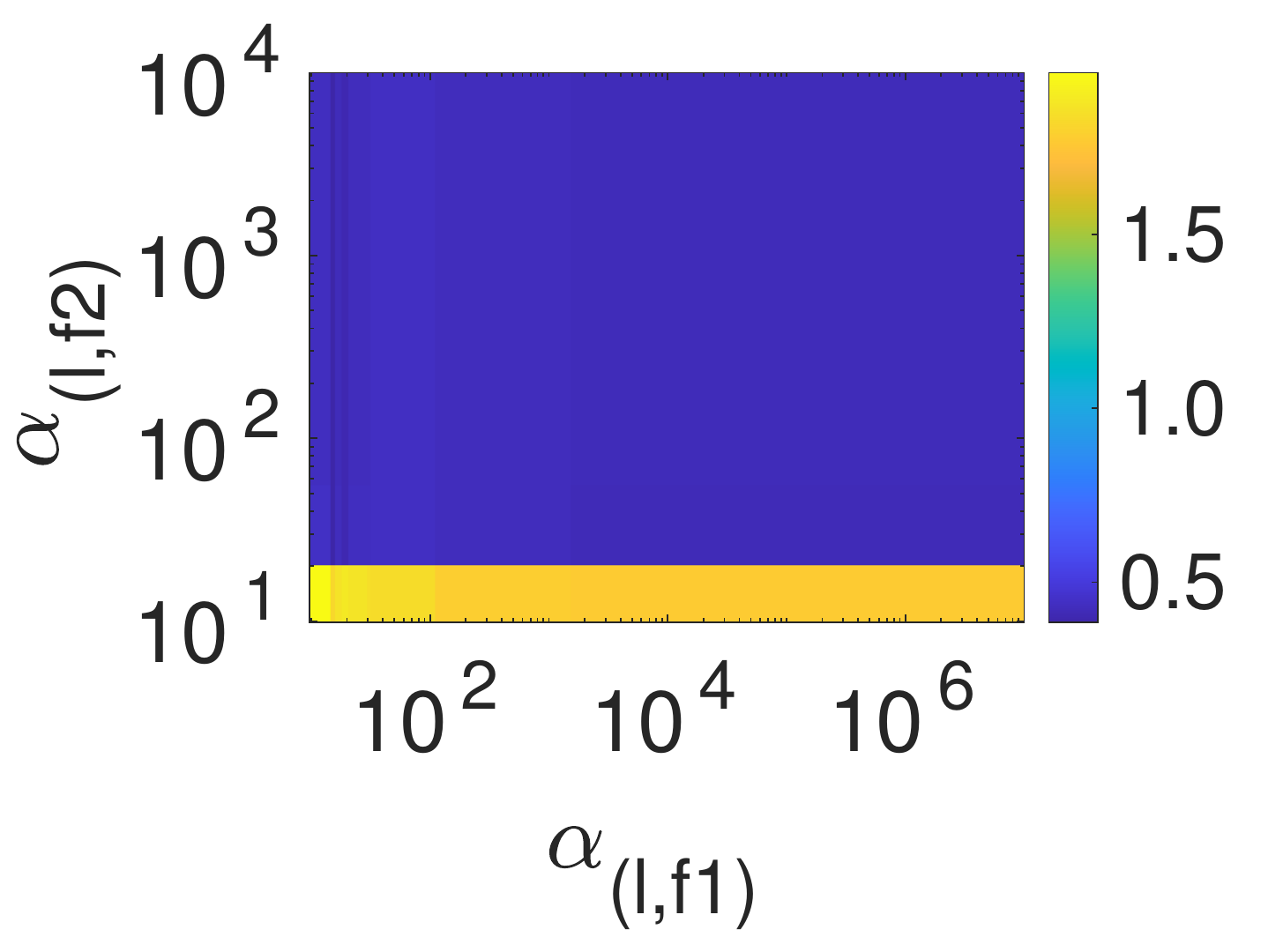}
	\label{fig:capacity_avg}
}
\subfigure[{$E^{\text{AVG}}$~[V/m]}]
{
	\includegraphics[width=42mm]{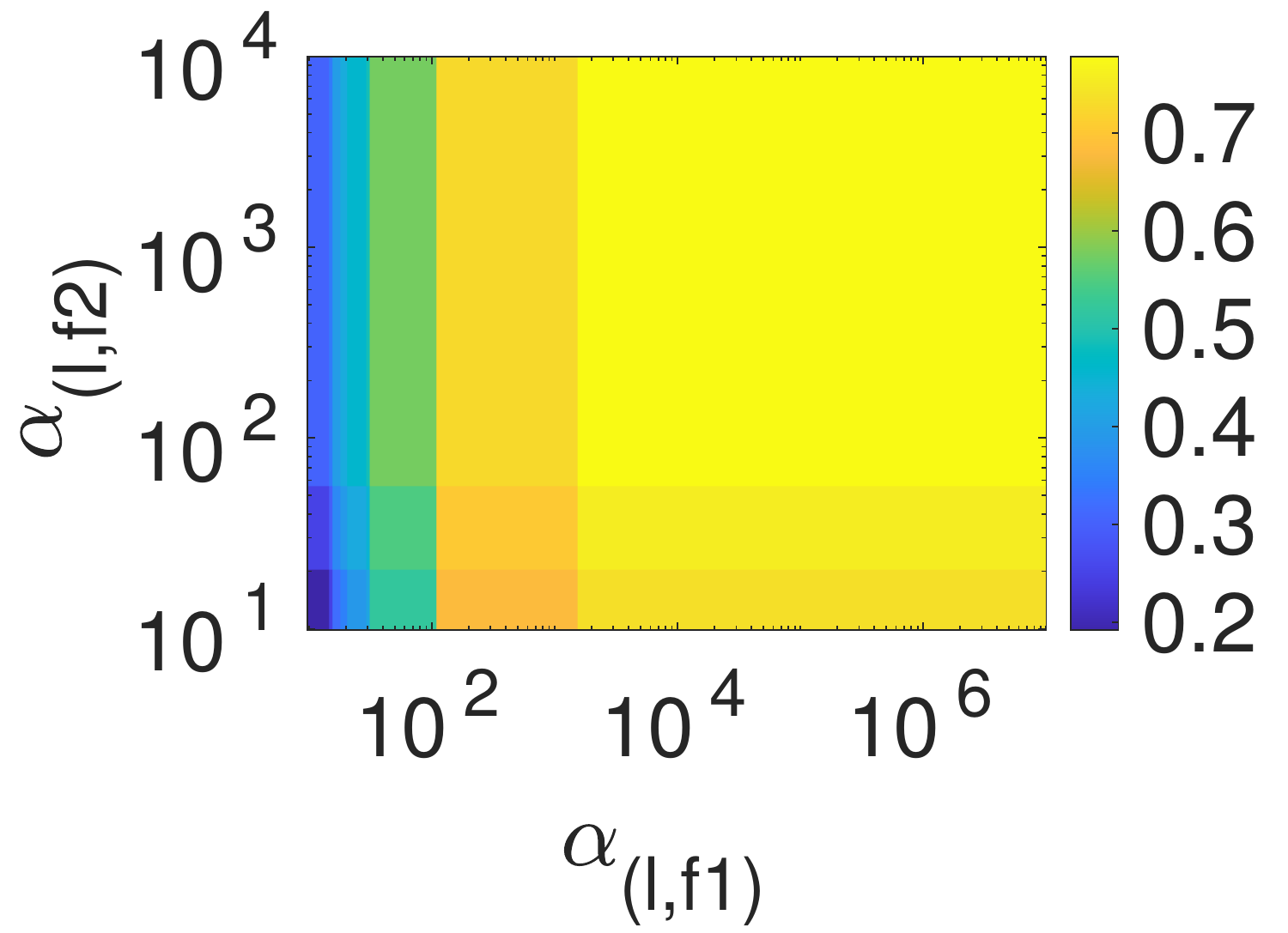}
	\label{fig:emf_heatmap}
}
\caption{Impact of $\alpha_{(l,f)}$ variation on: \textit{i}) total installation costs $C^{\text{TOT}}$, \textit{ii}) number $N_{f1}$ of $f1$ \acp{gNB}, \textit{iii}) number $N_{f2}$ of $f2$ \acp{gNB}, \textit{iv}) percentage of not served pixels $X^{\text{NOT-SERVED}}$, \textit{v}) average pixel throughput $T^{\text{AVG}}$, \textit{vi}) average electric field $E^{\text{AVG}}$.}
\label{fig:installation_var}
\vspace{-3mm}
\end{figure}

Several considerations hold by analyzing in detail Fig.~\ref{fig:installation_var}. First, $C^{\text{TOT}}$ is proportional to $\alpha_{(l,f1)}$ and $\alpha_{(l,f2)}$ (Fig.~\ref{fig:costs_performance}), due to the fact that the weights play a major role in determining the objective function, e.g., cost minimization, service maximization or a mixture between them. Clearly, $\alpha_{(l,f1)}$ ($\alpha_{(l,f2)}$) only affects $N_{f1}$ ($N_{f2}$), as shown in Fig.~\ref{fig:n_micro} (Fig.~\ref{fig:n_macro}). In addition, $X^{\text{NOT-SERVED}}$ is inversely proportional to $\alpha_{(l,f2)}$ (Fig.~\ref{fig:not_served}). For example, when $\alpha_{(l,f2)} \approx 10$, more than 10\% of pixels are not served by any \ac{gNB}. This is due to the fact that the number of $f2$ \acp{gNB} that are installed passes from 3 to 1 (see Fig.~\ref{fig:n_macro}), thus creating coverage holes. On the other hand, the variation of $\alpha_{(l,f1)}$ has a clear impact on $X^{\text{NOT-SERVED}}$ only when $\alpha_{(l,f2)} \approx 10$, i.e., when $f2$ \acp{gNB} are not able to cover the whole territory. Moreover, Fig~\ref{fig:capacity_avg} reveals that $T^{\text{AVG}}$ has a complex trend, which results from the combination of: \textit{i}) the coverage provided by $f1$ and $f2$ \acp{gNB} installed over the territory, \textit{ii}) the amount of interference, which tends to be impacted by the number of neighboring \acp{gNB} operating at the same frequency, and \textit{iii}) the percentage of served pixels, since $T^{\text{AVG}}$ is computed over the pixels that receive service coverage from at least one \ac{gNB}. As a consequence, $T^{\text{AVG}}$ is not always proportional or inversely proportional with $\alpha_{(l,f)}$. For example, the maximum value of $T^{\text{AVG}}$ is achieved when $\alpha_{(l,f2)}=10$, which however leads to a huge number of unserved pixels (i.e., more than 10\%). Finally, $E^{\text{AVG}}$ is proportional to $\alpha_{(l,f)}$, due to the variation in the number of radiating sources that contribute to the \ac{EMF} exposure. However, we point out that the average \ac{EMF} level is almost one order of magnitude lower than the 6~[V/m] restrictive limit.

\begin{figure}[t]
\centering
\subfigure[Map of installed \acp{gNB} (orange pins: $f1$ \acp{gNB}, blue pins: $f2$ \acp{gNB})]
{
    \includegraphics[width=6cm]{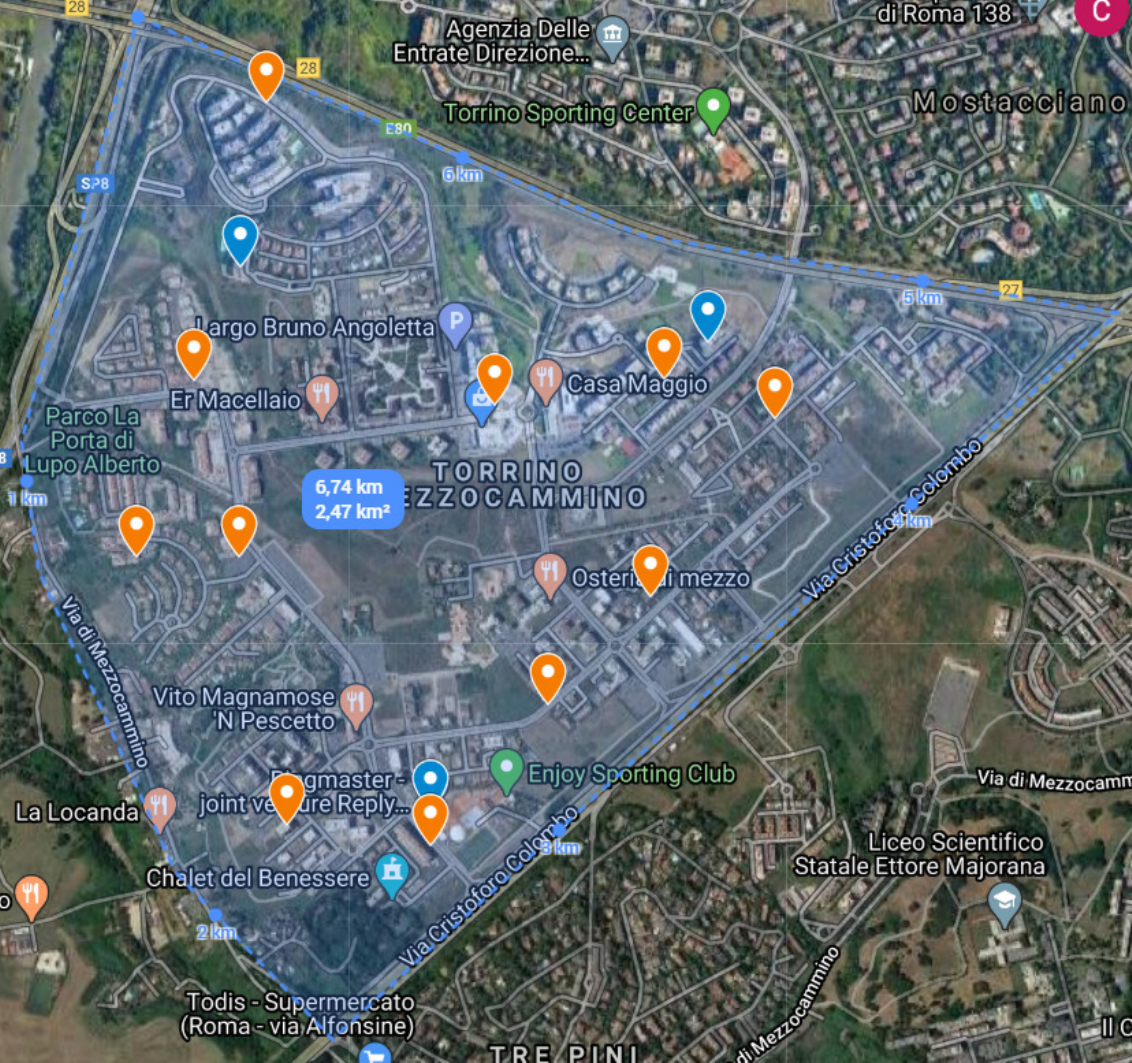}
    \label{fig:installed_sites}
}
\subfigure[{Electric field levels $E_{p}$~[V/m]}]
{
	\includegraphics[width=8cm]{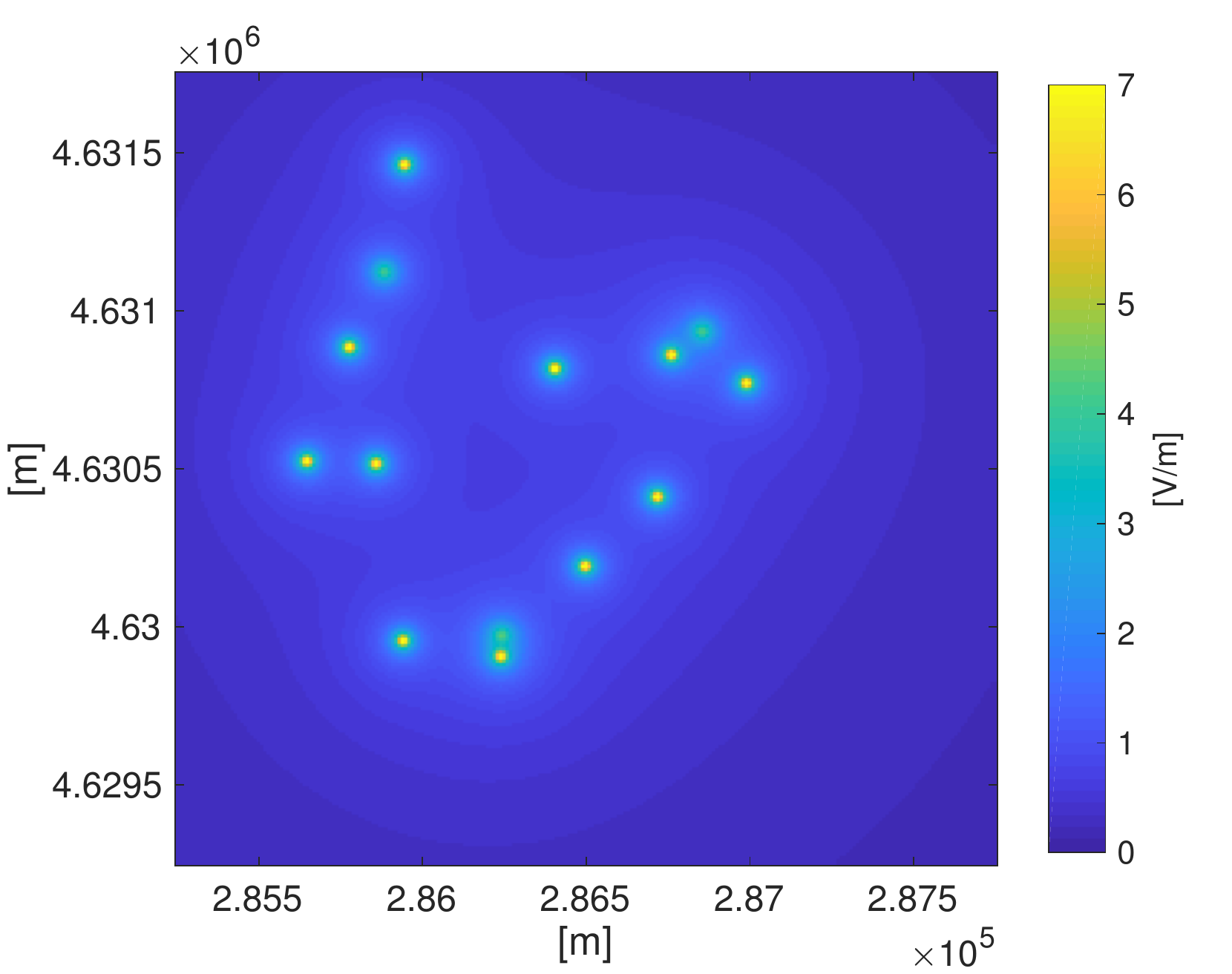}
	\label{fig:emission_assessment}
}
\caption{Visualization of the best scenario retrieved by \textsc{PLATEA} with $\alpha_{(l,f1)}=50$~[\euro] and $\alpha_{(l,f2)}=500$~[\euro].}
\label{fig:best_scenario}
\vspace{-3mm}
\end{figure}

Based on the above considerations, we select $\alpha_{(l,f1)}=50$~[\euro] and $\alpha_{(l,f2)}=500$~[\euro] henceforth. In this way, we balance between: i) increasing $T^{\text{AVG}}$, ii) reducing $C^{\text{TOT}}$, iii) minimizing $X^{\text{NOT-SERVED}}$, iv) reducing $E^{\text{AVG}}$. To give more insights, Fig.~\ref{fig:best_scenario} shows a run of the planning selected by \textsc{PLATEA} with the aforementioned setting. Interestingly, only a subset of the candidate \acp{gNB}, i.e., 11 $f1$ \acp{gNB} and 3 $f2$ \acp{gNB},  {are} deployed over the \ac{TMC} scenario (Fig.~\ref{fig:installed_sites}). On the other hand, the resulting \ac{EMF} levels are always pretty low (see Fig.~\ref{fig:emission_assessment}), with an electric field close to the 6~[V/m] limit only in proximity to the $f1$ \acp{gNB}.

\textbf{Algorithms Comparison.} In the following, we compare the performance of \textsc{PLATEA} against \textsc{EA} and \textsc{MCMA}. Unless otherwise specified, we compute each metric by averaging the results over 10 independent runs. Focusing on the number of $f1$ and $f2$ \acp{gNB} selected by \textsc{PLATEA}, we have found that our solution  requires on average $N_{f1}=10.7$ and $N_{f2}=3$, respectively. Consequently, we have passed to \textsc{EA} and \textsc{MCMA} a number of $f1$ \acp{gNB} equal to 11. In addition, \textsc{EA} requires the number of $f2$ \ac{gNB}, which is set to 3. Tab.~\ref{tab:compare_results_costs} reports the performance of the algorithms over the different metrics. More in detail, the total installation costs $C^{\text{TOT}}$ of \textsc{PLATEA} and \textsc{EA} are clearly lower than the ones of \textsc{MCMA}. Clearly, since \textsc{EA} requires as input the same (integer) number of $f1$ and $f2$ \acp{gNB} of \textsc{PLATEA}, it is natural that the two solutions achieve almost the same $C^{\text{TOT}}$. On the other hand, \textsc{MCMA} requires a larger number of $f2$ \acp{gNB}, in order to ensure full coverage. Focusing then on the number of pixels served by $f1$ and $f2$ \acp{gNB}, \textsc{PLATEA} operates a wiser choice compared to \textsc{EA} and \textsc{MCMA}, with several pixels that are served by $f1$ \acp{gNB}. Clearly, \textsc{MCMA} guarantees full coverage of the territory, $X^{\text{NOT-SERVED}}=0$\%. On the other hand,  {6\%} of pixels are not served with \textsc{EA}. Eventually, \textsc{PLATEA} ensures service coverage for  {99.97\%} of pixels. Moreover, \textsc{PLATEA} achieves a clearly higher throughput $T^{\text{AVG}}$ compared to \textsc{EA} and \textsc{MCMA}. In particular, the throughput difference of \textsc{PLATEA} w.r.t. \textsc{EA} and \textsc{MCMA} is huge, i.e., more than  {90}~[Mbps] on average.  {In addition, we can note that $T^{\text{AVG}}_{f}$ obtained by \textsc{PLATEA} is consistently higher than \textsc{EA} and \textsc{MCMA} over both frequencies.} Finally, the \ac{EMF} levels introduced by \textsc{PLATEA} and \textsc{EA} are clearly lower than \textsc{MCMA}.  In conclusion, \textsc{PLATEA} outperforms both \textsc{MCMA} and \textsc{EA} when the different metrics are jointly considered. We refer the interested reader to Appendix~\ref{app:add_results} for further comparisons between \textsc{PLATEA} and the reference algorithms. In the following, we will analyze in more detail the impact of the planning parameters on the \textsc{PLATEA} performance.

\begin{table}[t]
\centering
\caption{Comparison of \textsc{PLATEA} vs. reference algorithms \textsc{EA} and \textsc{MCMA}.}
\label{tab:compare_results_costs}
\begin{tabular}{|c|c|c|c|}
\hline
\rowcolor{Coral} \textbf{Metric} & \textbf{EA} & \textbf{MCMA} & \textbf{PLATEA}  \\
\hline
& & &\\[-0.95em]
         $C^{\text{TOT}}$~[k\euro]  & 391.4 & \underline{555.8} & 386.1 \\
\rowcolor{Linen}         $N_{f1}$ & 11 & 11 & 10.7 \\
         $N_{f2}$ & 3 & \underline{5.5} & 3 \\[0.05em]
\rowcolor{Linen} & & &\\[-0.95em]
\rowcolor{Linen}         $X^{\text{SERVED}}_{f1}$ & \underline{7960}  & \underline{7960}  & \underline{9103} \\[0.05em]
& & &\\[-0.95em]
         $X^{\text{SERVED}}_{f2}$  & \underline{14897} & \underline{16357}  & \underline{15307} \\[0.05em]
\rowcolor{Linen} & & &\\[-0.95em]
\rowcolor{Linen}         $X^{\text{NOT-SERVED}}$~[\%]  & \underline{6} & 0 & \underline{0.03} \\
 & & &\\[-0.95em]
         $T^{\text{AVG}}$~[Mbps] & \underline{337.1} & \underline{319.3}  & \underline{428.3} \\
 & & &\\[-0.95em]
\rowcolor{Linen}         $T^{\text{AVG}}_{f1}$~[Mbps] & \underline{356.6} & \underline{356.6}  & \underline{391.5}  \\
 & & &\\[-0.95em]
         $T^{\text{AVG}}_{f2}$~[Mbps] & \underline{199.3} & \underline{166.4}  & \underline{245.2}  \\
\rowcolor{Linen}         $E^{\text{AVG}}$~[V/m]  & 0.57 & \underline{0.63} & 0.57\\
\hline
\end{tabular}
\end{table}

\textbf{Impact of planning parameters.} We then focus on the impact of the planning parameters, namely: \textit{i}) the scaling parameters $R^{\text{TIME}}_{(l,f)}$ and $R^{\text{STAT}}_{(l,f)}$, which affect the \ac{EIRP} and consequently the \ac{EMF} exposure generated by \acp{gNB} and \textit{ii}) the minimum distance from sensitive places $D^{\text{MIN}}$, which influences the subset of sites that can host \acp{gNB}. Focusing on \textit{i}) we perform a sensitivity analysis by running \textsc{PLATEA} over a wide range of $R^{\text{TIME}}_{(l,f)}$ and $R^{\text{STAT}}_{(l,f)}$ values. For the sake of simplicity, we impose $R^{\text{TIME}}_{(l,f)} \in [0.1-0.6] \quad \forall f \in \mathcal{F}, l \in \mathcal{L}$. On the other hand, we set $R^{\text{TIME}}_{(l,f1)} \in [0.1-0.6] \quad \forall l \in \mathcal{L}$ and $R^{\text{STAT}}_{(l,f2)}=1 \quad \forall l \in \mathcal{L}$. In this way, $f1$ \acp{gNB} are subject to temporal and statistical scaling factors, while $f2$ \ac{gNB} are affected only by temporal scaling factors.

Fig.~\ref{fig:scaling} reports the obtained results in terms of: \textit{i}) average \ac{EMF} $E^{\text{AVG}}$, \textit{ii}) number $N_{{f1}}$ of $f1$ \acp{gNB}, \textit{iii}) number $N_{{f2}}$ of $f2$ \acp{gNB}, \textit{iv}) average throughput $T^{\text{AVG}}$, \textit{v}) percentage of not served pixels $X^{\text{NOT-SERVED}}$, \textit{vi}) number of pixels served by $f2$ \acp{gNB} $X^{\text{SERVED}}_{f2}$. Interestingly, the choice of $R^{\text{TIME}}_{(l,f)}$ and $R^{\text{STAT}}_{(l,f)}$ has a huge impact on the obtained planning. In particular, when $R^{\text{TIME}}_{(l,f)}$ and $R^{\text{STAT}}_{(l,f)}$ are close to 0.1, the average \ac{EMF} exposure is very low (i.e., lower than 0.4~[V/m]), as shown at the bottom left corner of Fig.~\ref{fig:scaling_emf}. In this region, \textsc{PLATEA} installs more than 10 $f1$ \ac{gNB} (Fig.~\ref{fig:scaling_n_micro}) and 3 $f2$ \acp{gNB} (Fig.~\ref{fig:scaling_n_macro}). In addition, a large throughput is achieved (Fig.~\ref{fig:scaling_throughput}) and (almost) all the pixels are served (Fig.~\ref{fig:scaling_unserved}).  On the other hand, when $R^{\text{TIME}}_{(l,f)}$ and $R^{\text{STAT}}_{(l,f)}$ are increased, the average \ac{EMF} tends to increase and therefore it is challenging to ensure the strict \ac{EMF} constraints in the proximity of the installed \acp{gNB}. Therefore, the number of installed \acp{gNB} is reduced, the throughput is decreased and the percentage of not served pixels abruptly increases. Eventually, for large values of $R^{\text{TIME}}_{(l,f)}$ and $R^{\text{STAT}}_{(l,f)}$ (top right corner of subfigures), it is not possible to install any \ac{gNB} and therefore all the pixels are unserved. In addition, we can note that a frontier region emerges for intermediate values of the scaling parameters. Interestingly, for most of $R^{\text{TIME}}_{(l,f)}$ and $R^{\text{STAT}}_{(l,f)}$ combinations laying on the frontier, a huge amount of pixels is served by $f2$ \acp{gNB} (Fig.~\ref{fig:scaling_served}). 

\begin{figure}[t]
\centering
\subfigure[{$E^{\text{AVG}}$~[V/m]}]
{
    \includegraphics[width=42mm]{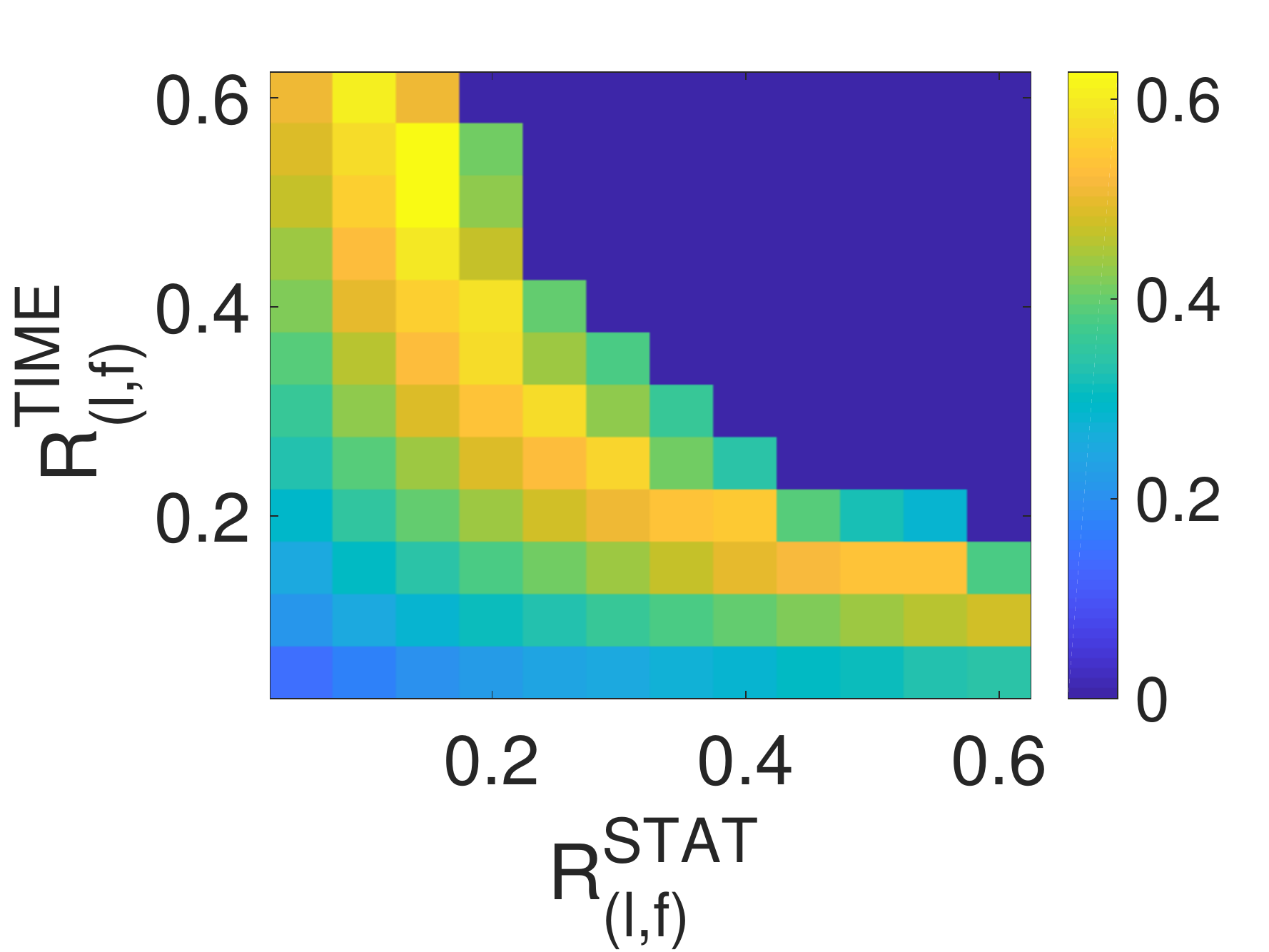}
    \label{fig:scaling_emf}
}
\subfigure[$N_{{f1}}$]
{
    \includegraphics[width=42mm]{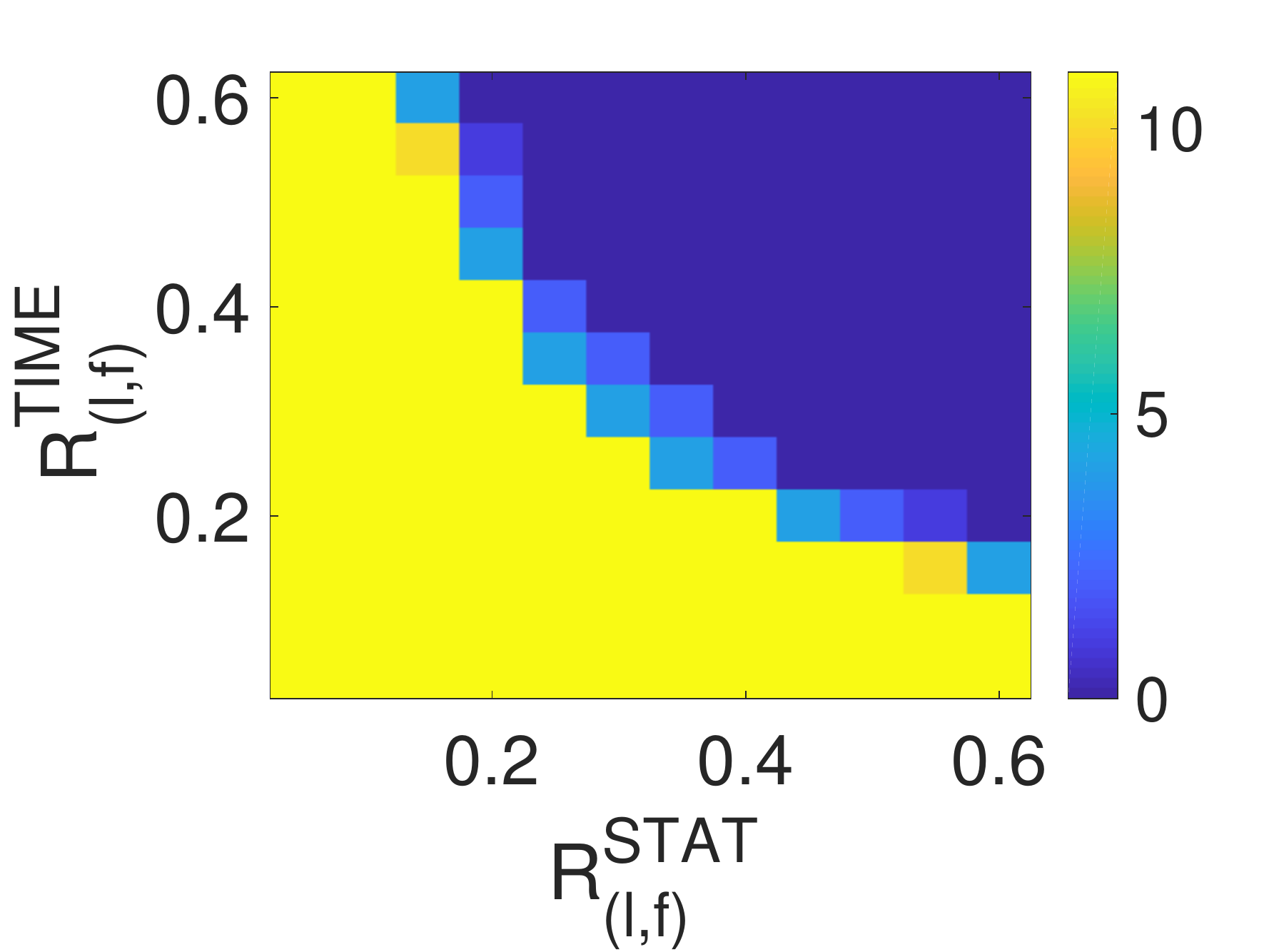}
    \label{fig:scaling_n_micro}
}

\subfigure[$N_{{f2}}$]
{
    \includegraphics[width=42mm]{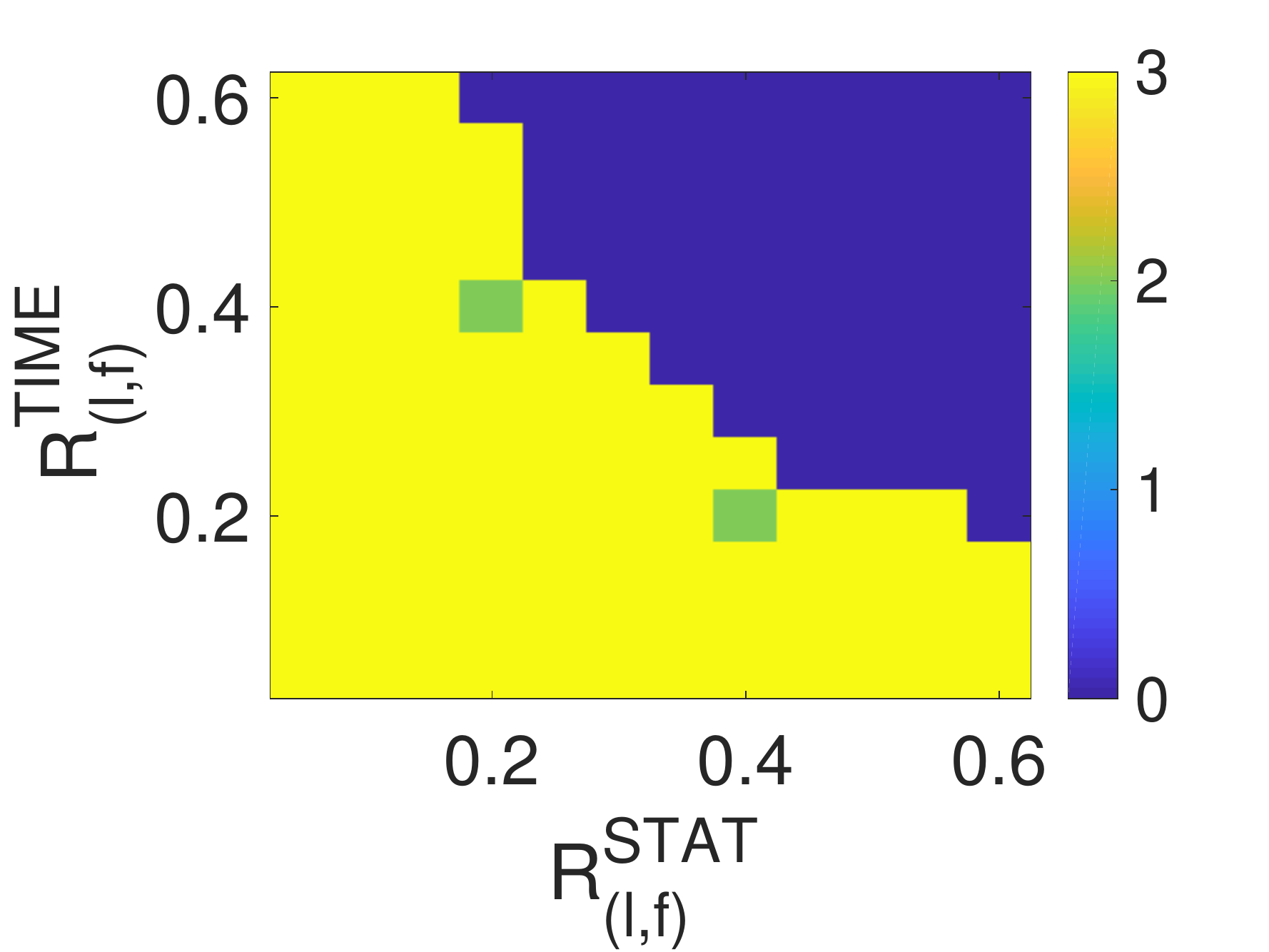}
    \label{fig:scaling_n_macro}
}
\subfigure[{$T^{\text{AVG}}$~[\underline{Mbps}]}]
{
    \includegraphics[width=42mm]{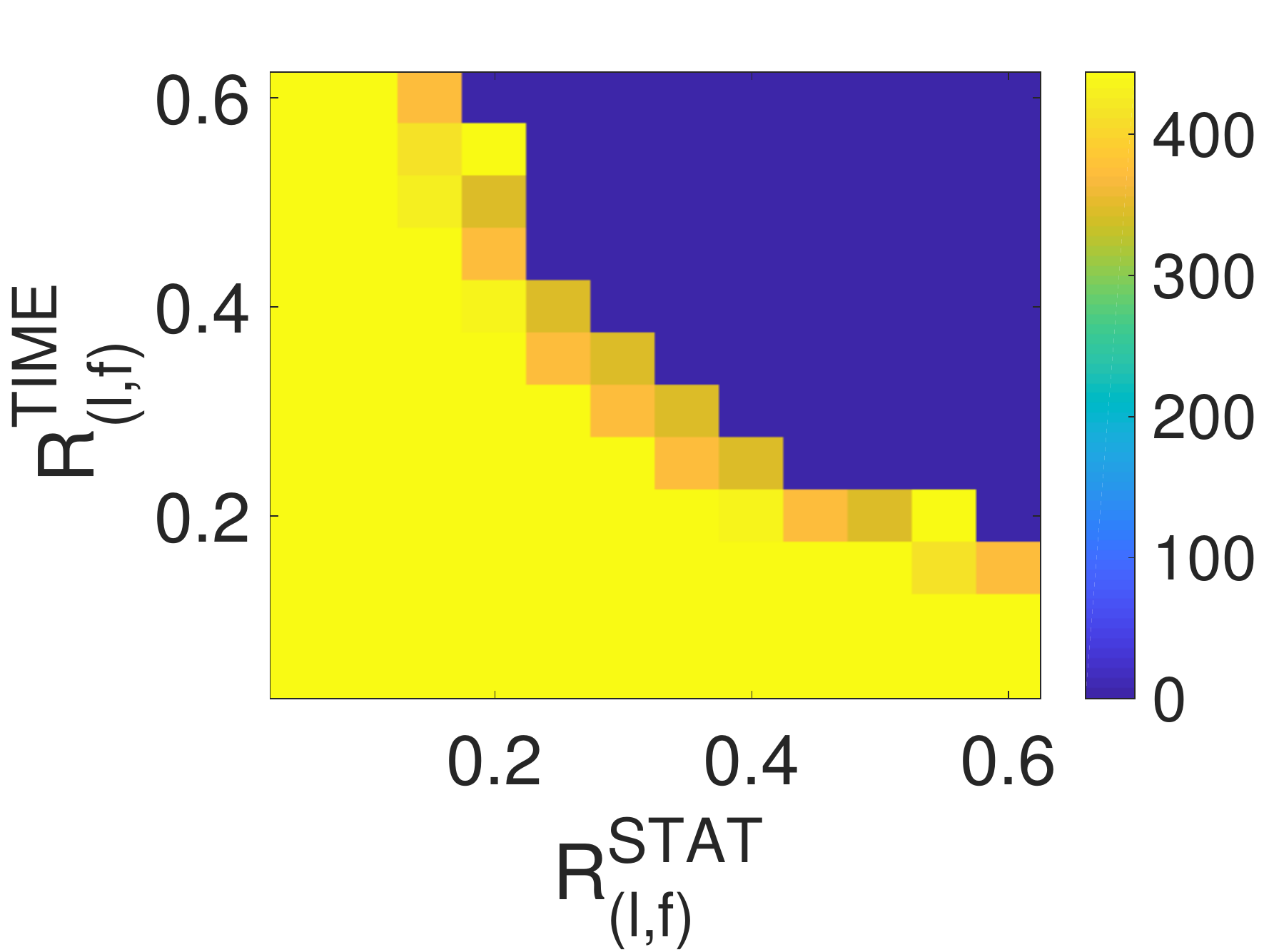}
    \label{fig:scaling_throughput}
}

\subfigure[{$X^{\text{NOT-SERVED}}$~[\%]}]
{
    \includegraphics[width=42mm]{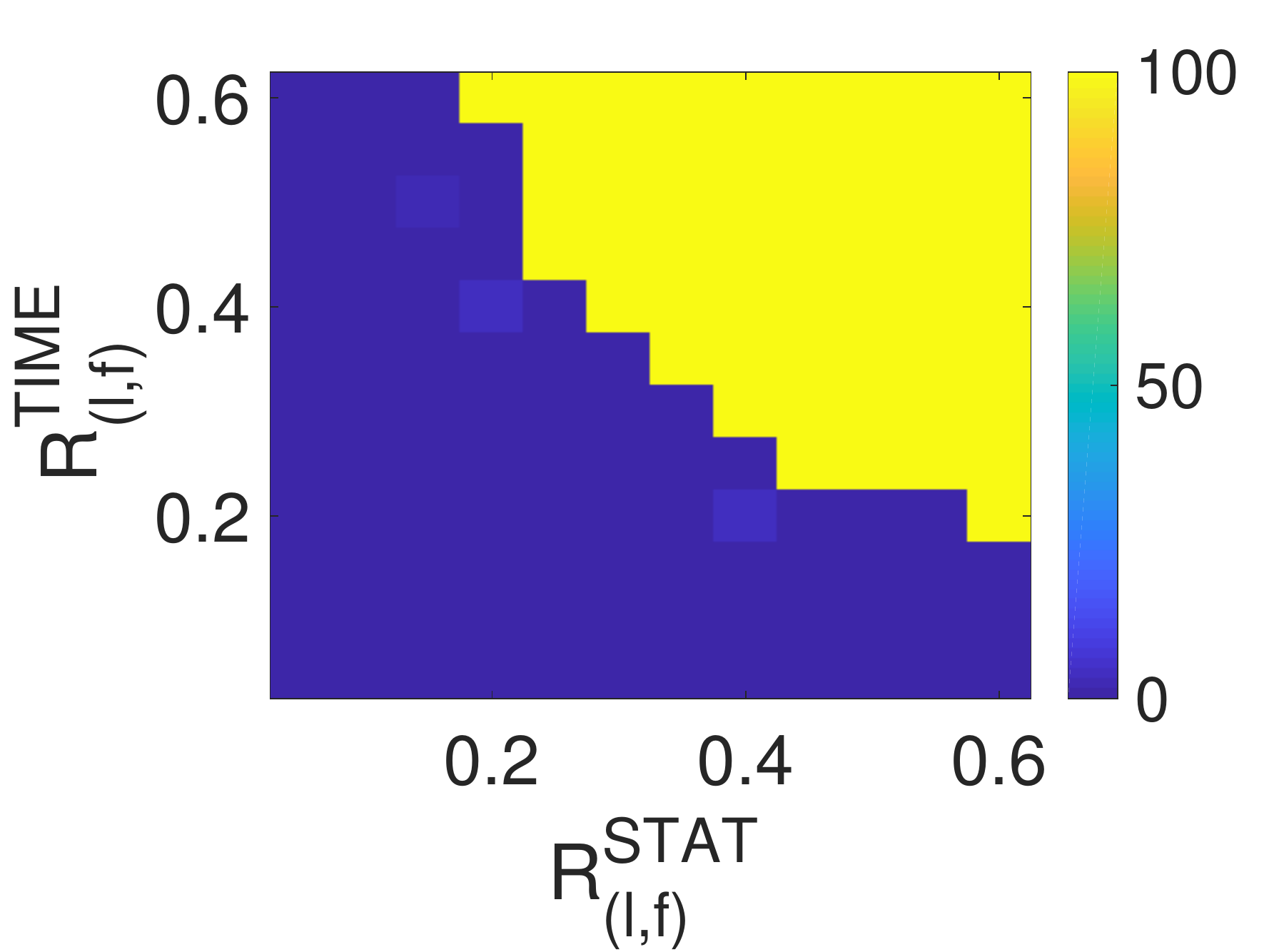}
    \label{fig:scaling_unserved}
}
\subfigure[$X^{\text{SERVED}}_{f2}$]
{
    \includegraphics[width=42mm]{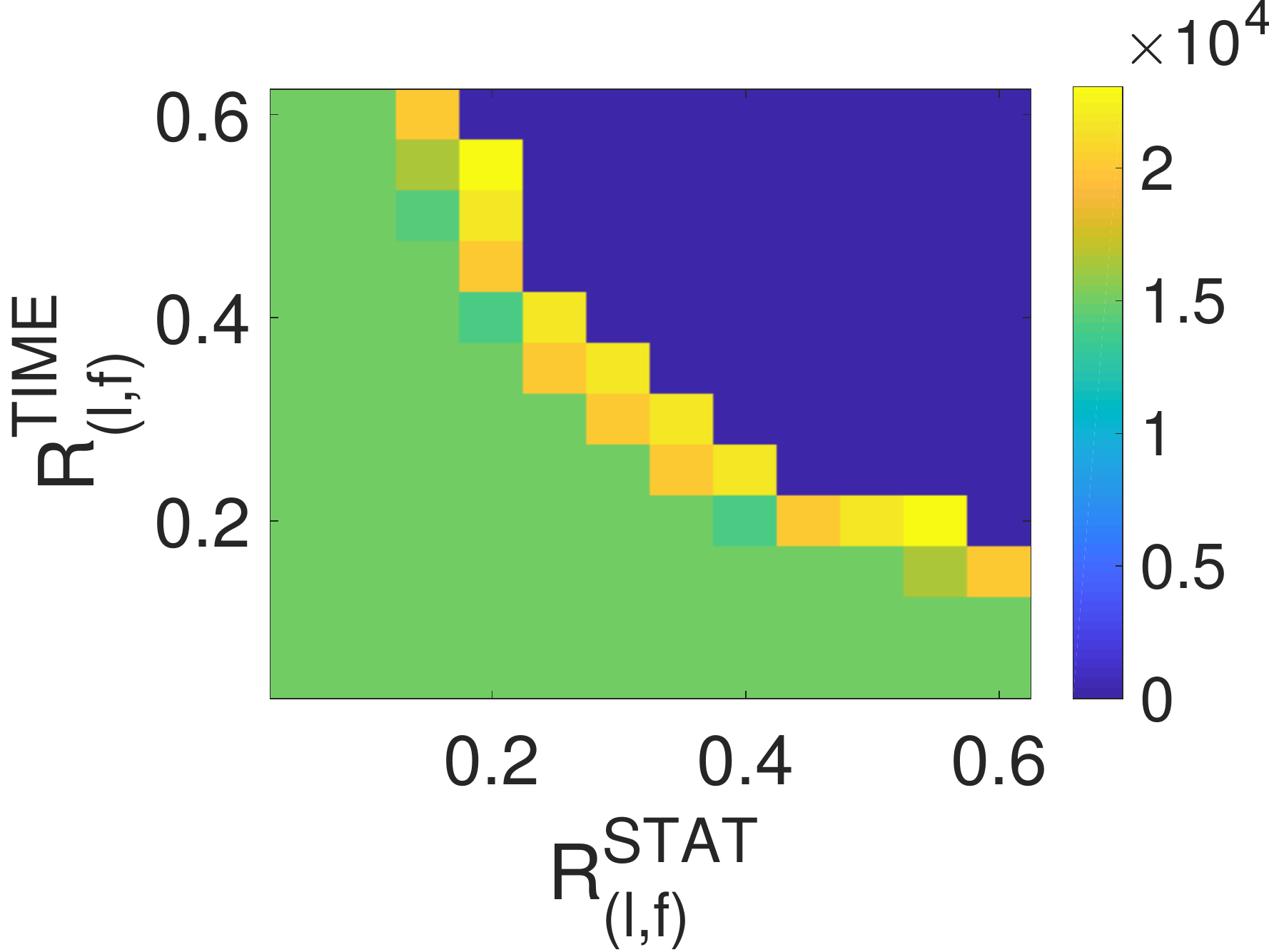}
    \label{fig:scaling_served}
}
\caption{Impact of the variation of the scaling parameters on: \textit{i}) average electric field $E^{\text{AVG}}$, \textit{ii}) number $N_{f1}$ of $f1$ \acp{gNB}, \textit{iii}) number $N_{f2}$ of $f2$ \acp{gNB}, \textit{iv}) average throughput $T^{\text{AVG}}$, \textit{v)} percentage of not served pixels $X^{\text{NOT-SERVED}}$, \textit{vi}) number of pixels served by $f2$ \acp{gNB} $X^{\text{SERVED}}_{f2}$.}
\label{fig:scaling}
\end{figure}

Summarizing, our results demonstrate that the setting of the scaling parameters plays a major role in the planning of 5G networks, especially for countries adopting strict \ac{EMF} limits. As a side comment, we believe that the methodology to estimate the values of $R^{\text{TIME}}_{(l,f)}$ and $R^{\text{STAT}}_{(l,f)}$ should be integrated in the national \ac{EMF} regulations. Clearly, the exact settings of the scaling parameters depend on the considered scenario.

We then evaluate the impact of varying $D^{\text{MIN}}$, which we remind is an additional restriction imposed by the municipality of Rome. We report here the main outcomes from this test, while we refer the interested reader to Appendix~\ref{app:min_distance} for more details. In brief, values of $D^{\text{MIN}}<100$~[m] do not significantly alter the  results presented so far. On the contrary, a value of $D^{\text{MIN}}=150$~[m] introduces huge limitations on the \acp{gNB} installations, and consequently on the 5G service in terms of throughput and number of served pixels.


%


\textbf{Impact of pre-5G exposure levels.}  {We then focus on} the impact of  adding the pre-5G exposure term on the 5G planning.  {Due to the lack of space, we refer the reader to Appendix~\ref{app:pre-exposure} for a detailed report about these outcomes, while here we briefly summarize the salient features. In brief, when} the pre-5G exposure is increased, we can note: \textit{i}) a reduction in the number of $f1$ \acp{gNB}, and consequently of total costs, \textit{ii}) an increase in the number of pixels served by $f2$ \acp{gNB},  \textit{iii}) a slight throughput decrease, and \textit{iv}) an \ac{EMF} increase, mainly due to the pre-5G exposure term. Overall, these results prove that the performance metrics are impacted by the level of background exposure. However, \textsc{PLATEA} is always able to retrieve a feasible planning, with a percentage of unserved pixels at most equal to 0.16\%. 

\textbf{ {Impact of frequency reuse scheme.}}  {In the last part of our work, we have investigated the impact of changing the frequency reuse factor $\epsilon^{\text{F-REUSE}}_f$. We refer the reader to Appendix~\ref{app:freq_reuse} for the details. In brief, we have found that, as $\epsilon^{\text{F-REUSE}}_f$ is increased, the throughput, the total costs and the average \ac{EMF} levels are decreased, mainly because less \acp{gNB} operating on frequency $f1$ are installed.}


\section{Conclusions and Future Works}
\label{sec:conclusions}

We have focused on the problem of planning a 5G network under service and \ac{EMF} constraints. To this aim, we have targeted an objective function that balances between \ac{gNB} installation costs and 5G service coverage level. After providing the \textsc{OPTPLAN-5G} \ac{MILP} formulation, we have demonstrated that the considered problem is NP-Hard, and therefore very challenging to be solved even for small problem instances. To face this issue, we have designed the \textsc{PLATEA} algorithm, which is able to select a 5G planning by iterating over the set of candidate \acp{gNB}. In addition, \textsc{PLATEA} exploits the linear constraints that have been defined for \textsc{OPTPLAN-5G}. We have then considered the \ac{TMC} scenario, which is subject to very strict \ac{EMF} regulations that include minimum distances from sensitive places and very stringent \ac{EMF} limits.

Results, obtained by running \textsc{PLATEA} over the considered scenario, prove that our solution outperforms \textsc{EA} and \textsc{MCMA}. In addition, we have demonstrated that the 5G planning is overall feasible, i.e., it is possible to serve a huge amount of pixels while limiting the installation costs and while ensuring the \ac{EMF} constraints outside the exclusion zones of the installed \acp{gNB}. However, our work points out an important aspect: the scaling parameters that are used to estimate the exposure level from 5G \acp{gNB} play a fundamental role in determining the problem feasibility and consequently the set of installed \acp{gNB}. Eventually, when pre-5G exposure is considered, \textsc{PLATEA} is still able to retrieve an admissible planning, with a moderate impact on  pixel throughput.  {Finally, the adopted frequency reuse scheme strongly influences the obtained planning.}

We believe that this work could be the first step towards a more comprehensive approach. First of all, the integration of \acp{gNB} operating on mm-Waves may be an interesting future work. In addition, the optimization of the scaling parameters may be another research avenue. Eventually, we plan to integrate detailed propagation models (e.g., including indoor evaluation) {, more complex} \ac{EMF} models,  {more detailed service models (even different than a Massive-}\ac{MIMO}  {system)},  {fronthaul/backhaul constraints governing the installation of \acp{gNB}, multiple spectrum sharing options, multiple power budget levels across the \acp{gNB} operating at the same frequency, as well as the integration of edge/caching considerations in the adopted models}.  {Finally, we will investigate possible algorithmic solutions to parallelize the planning computation when large territory areas are taken under consideration.}

\bibliographystyle{ieeetr}

\begin{IEEEbiography}[{\includegraphics[width=1in,height=1.25in,clip,keepaspectratio]{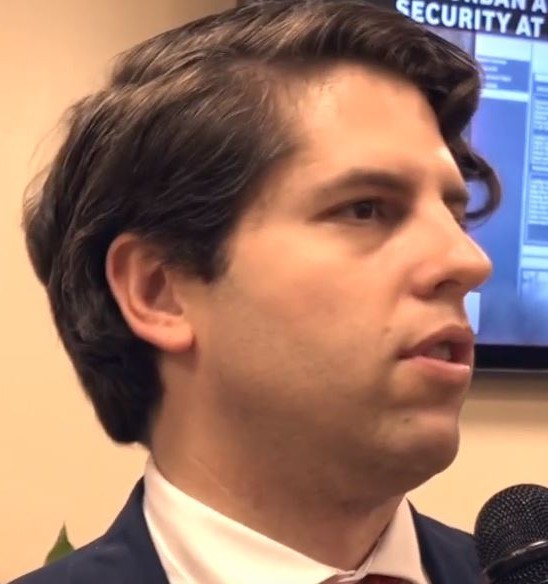}}]{Luca Chiaraviglio} (M'09-SM'16) is Associate Professor at the
University of Rome Tor Vergata (Italy). He holds a Ph.D. in Telecommunication and Electronics Engineering, obtained from Politecnico di Torino (Italy). Luca has co-authored 140+ papers published in international journals, books and conferences. He is TPC member of IEEE INFOCOM, associate editor for \textsc{IEEE Communications Magazine} and \textsc{IEEE Transactions on Green Communications and Networking}, and Specialty Chief Editor of \textsc{Frontiers in Communications and Networks}. Luca has received the Best Paper Award at IEEE VTC-Spring 2020, IEEE VTC-Spring 2016 and ICIN 2018, all of them appearing as first author. Some of his papers are listed as Best Readings on Green Communications by IEEE. Moreover, he has been recognized as an author in the top 1\% most highly cited papers in the ICT field worldwide. His current research topics cover 5G networks, optimization applied to telecommunication networks, and health risks assessment of 5G communications.
\end{IEEEbiography}

\begin{IEEEbiography}[{\includegraphics[width=1in,height=1.25in,clip,keepaspectratio]{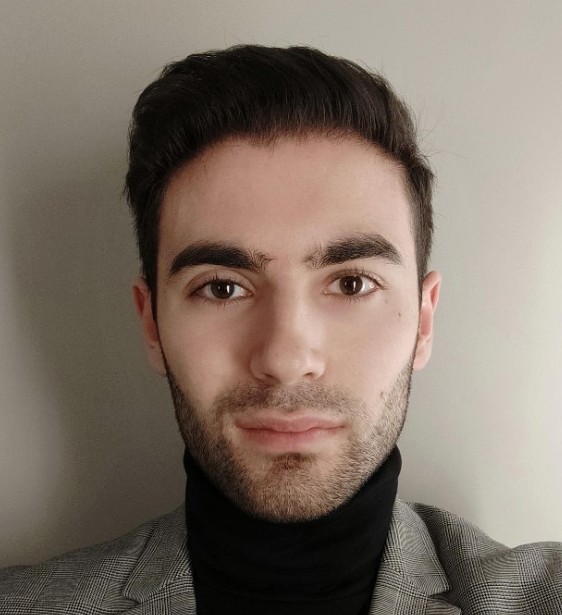}}]{Cristian Di Paolo} graduated in ICT and Internet Engineering from the University of Rome Tor Vergata in February 2020. Between January and March 2020 he has been a CNIT Researcher. He currently works as Network Engineer at Huawei.
\end{IEEEbiography}

\begin{IEEEbiography}[{\includegraphics[width=1in,height=1.25in,clip,keepaspectratio]{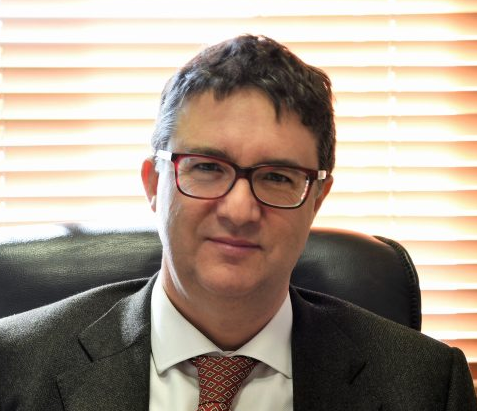}}]{Nicola Blefari-Melazzi} is currently a Full
Professor of telecommunications with the University
of Rome ``Tor Vergata'', Italy. He is currently
the Director of CNIT, a consortium of 37 Italian
Universities. He has participated in over 30 international
projects, and has been the principal investigator
of several EU funded projects. He has been
an Evaluator for many research proposals and a
Reviewer for numerous EU projects. He is the
author/coauthor of about 200 articles, in international
journals and conference proceedings. His research interests include
the performance evaluation, design and control of broadband integrated
networks, wireless LANs, satellite networks, and of the Internet.
\end{IEEEbiography}

\clearpage
\newpage

\appendices

\section{Considerations on the Point Source Model}
\label{app:point-source}

 {Hereafter, we report some considerations about the point source model and its positioning against synthetic and full-wave models. In brief, the point source model is based on a conservative approach, which in general brings an over-estimation EMF w.r.t to the exposure level that is then experienced once the network is put under operation. As a result, the exposure levels that are observed after the \acs{gNB} installation are generally lower than the ones predicted by the point source model during the planning phase. Clearly, if the adherence to the exposure limit is verified from the point source model, this condition will also hold after the installation of the deployment.}


 {On the other hand, both synthetic and full-wave models are able to provide a better estimation of the EMF, which can be an appealing property from the operator's side. However, such approaches introduce different challenges, which include:}
\begin{itemize}
\item  {an increased complexity in integrating these models in the considered framework;}
\item  {a strong dependence with extremely detailed antenna parameters, which are typically not disclosed by operators.}
\end{itemize}

 {In conclusion, we have selected the point source model because in this way we adopt a conservative approach to verify the compliance of the EMF levels against the regulations. In addition, we exploit the model to generate a set of linear constraints that are easily integrated in our framework.}

\section{Pixel Tessellation}
\label{app:pixel}

 By applying a tessellation based on pixels, three important goals can be met,  {namely}: 
\begin{enumerate}
\item the power density terms are computed over the whole area under consideration. More in depth, we evaluate the power density that is received over the pixel center from all the installed \acp{gNB}. In this way, we are able to extend the \ac{EMF} compliance assessment over the whole territory. In addition, we model the presence of exclusion zones in proximity to the installed \acp{gNB}, i.e., zones of the territory that are not accessed by users and therefore in such zones the \ac{EMF} compliance assessment is not required for the general public; 

\item  {we evaluate the throughput/coverage in each pixel over the territory, and not only in the locations where the users are supposed to be.} This assumption appears to be meaningful in the context of 5G, especially for the \ac{eMBB} scenario \cite{3gppservice}. In this way, it is possible to control the amount of throughput provided to each pixel, and consequently to the \ac{UE} that are located in the pixel; 

\item  {a minimum throughput requirement may be introduced for each pixel rather than for single users.} In this way, it is possible to (indirectly) take into account the effects of \ac{UE} densification and/or \ac{UE} mobility. For example, by assigning different values of required throughput, it is possible to model high density zones, where the throughput requirements are high, compared to other zones, which instead are not visited by users. In a similar way, it is possible to vary the throughput requirements based on the \ac{UE} mobility, e.g., by increasing the throughput for the zones that are subject to high \ac{UE} mobility, in order to take into account the effect of handovers and/or possible traffic spikes.
\end{enumerate}

 {Clearly, our framework guarantees throughput requirements for each pixel, which is not exactly equivalent to guarantee the same throughput for any user within the pixel. Intuitively, in fact, the radio resources consumed by a given pixel have to be split among the users located in the pixel. To further complicate this picture, the allocation of radio resources over the same pixel may be not uniform across the users, due to different service/QoS requirements. Since the goal of the planing phase is not to consider this level of detail, in this work we are assuming throughput requirements
on a pixel base, and not on single users. Clearly, the pixel throughput corresponds to the maximum value that can be experienced by a UE within the pixel.}

\section{Notation}
\label{app:notation}

 {Tab.~\ref{tab:notation} reports the main notation defined in the problem formulation.}

\begin{table}[t]
    \caption{Main Notation.}
    \label{tab:notation}
    \scriptsize
    \centering
    \begin{tabular}{|>{\columncolor{DarkLinen}}p{0.15cm}|c|p{6.1cm}|}
\hline
 \rowcolor{Coral}  & \multicolumn{1}{|c|}{\textbf{Symbol}} & \textbf{Description}  \\
\hline
\rowcolor{Linen}  \cellcolor{DarkLinen} &  \multicolumn{1}{|c|}{$\mathcal{P}$} & Set of pixels \\[-0.05em]
& \multicolumn{1}{|c|}{$\mathcal{P}^{\text{RES}} \in \mathcal{P}$} & Set of pixels in residential areas \\[-0.05em]
\rowcolor{Linen} \cellcolor{DarkLinen}  & \multicolumn{1}{|c|}{$\mathcal{P}^{\text{GEN}} \in \mathcal{P}$} & Set of pixels in general public areas\\[-0.05em]
& \multicolumn{1}{|c|}{$\mathcal{P}^{\text{SENS}} \in \mathcal{P}$} & Set of pixels in sensitive areas \\[-0.05em]
\rowcolor{Linen} \cellcolor{DarkLinen} & \multicolumn{1}{|c|}{$\mathcal{L}$} & set of candidate locations for installing 5G gNBs\\[-0.05em]
 \cellcolor{DarkLinen} \multirow{-6}{0.5cm}{\begin{sideways}\textbf{Set Notation}\end{sideways}} & \multicolumn{1}{|c|}{$\mathcal{F}$} & Set of operating frequencies for 5G gNB\\[-0.05em]
\hline
\rowcolor{Linen} \cellcolor{DarkLinen}& \multicolumn{1}{|c|}{\multirow{2}{*}{$S_f^{\text{MIN}}$}} & Minimum SIR to achieve in order to guarantee the required 5G service on frequency $f \in \mathcal{F}$\\[-0.05em]
& \multicolumn{1}{|c|}{\multirow{2}{*}{$\beta_{(l,p,l_2,f)}$}} & Signal/interference contribution from 5G gNB $l2 \in \mathcal{L}$ on frequency $f \in \mathcal{F}$ over pixel $p \in \mathcal{P}$ served by \ac{gNB} $l \in \mathcal{L}$\\[-0.05em]
\rowcolor{Linen} \cellcolor{DarkLinen}& \multicolumn{1}{|c|}{\multirow{2}{*}{$C^{\text{SITE}}_{(l,f)}$}} & Site installation cost of a 5G gNB site operating on frequency $f \in \mathcal{F}$ at location $l \in \mathcal{L}$\\[-0.05em]
 \cellcolor{DarkLinen}  & \multicolumn{1}{|c|}{} & \\[-0.9em]
& \multicolumn{1}{|c|}{$C^{\text{EQUIP}}_{f}$} & Equipment cost of a 5G gNB operating on frequency $f \in \mathcal{F}$;  \\[-0.05em]
\rowcolor{Linen} \cellcolor{DarkLinen} & \multicolumn{1}{|c|}{\multirow{2}{*}{$P^{\text{BASE}}_{(p,f)}$}} & Baseline power density on frequency $f \in \mathcal{F}$ received by pixel $p \in \mathcal{P}$\\[-0.05em]
& \multicolumn{1}{|c|}{\multirow{3}{*}{$P^{\text{ADD}}_{(p,l,f)}$}} & Additional power density received by pixel $p \in \mathcal{P}$ from a 5G gNB installed at location $l \in \mathcal{L}$ operating on frequency $f \in \mathcal{F}$\\[-0.05em]

\rowcolor{Linen} \cellcolor{DarkLinen} & \multicolumn{1}{|c|}{\multirow{2}{*}{$L^{\text{RES}}_f$}} & Power density limit over frequency $f \in \mathcal{F}$ for a pixel belonging to a residential area\\[-0.05em]
& \multicolumn{1}{|c|}{\multirow{2}{*}{$L^{\text{GEN}}_f$}} & Power density limit over frequency $f \in \mathcal{F}$ for a pixel belonging to a general public area\\[-0.05em]
\rowcolor{Linen} \cellcolor{DarkLinen} & \multicolumn{1}{|c|}{\multirow{2}{*}{$D^{\text{MIN}}$}} & Minimum distance between an installed 5G gNB site and a sensitive place\\[-0.05em]
& \multicolumn{1}{|c|}{\multirow{2}{*}{$D^{\text{MAX}}_f$}} & Max. 5G coverage distance between a 5G gNB operating on frequency $f$ and a covered pixel\\[-0.05em]
\rowcolor{Linen} \cellcolor{DarkLinen} & \multicolumn{1}{|c|}{\multirow{2}{*}{$D_{(p,l,f)}$}} & Distance between pixel $p \in \mathcal{P}$ and a \ac{gNB} operating on frequency $f \in \mathcal{F}$ and installed at location $l \in \mathcal{L}$\\[-0.05em]
& \multicolumn{1}{|c|}{\multirow{3}{*}{$I^{\text{ZONE}}_{(p,l,f)}$}} & Exclusion zone indicator: 1 if pixel $p \in \mathcal{P}$ falls inside the exclusion zone of a 5G \ac{gNB} installed at location $l \in \mathcal{L}$ and operating on frequency $f \in \mathcal{F}$, 0 otherwise\\[-0.05em]
\rowcolor{Linen} \cellcolor{DarkLinen} & \multicolumn{1}{|c|}{\multirow{2}{*}{$N^{\text{SER}}$}} & Maximum number of 5G \acp{gNB} that can serve a single pixel\\[-0.05em]
& \multicolumn{1}{|c|}{\multirow{2}{*}{$N^{\text{MAX}}$}} & Maximum number of 5G \acp{gNB} that can be installed in a location\\[-0.05em]
\rowcolor{Linen} \cellcolor{DarkLinen} & \multicolumn{1}{|c|}{\multirow{2}{*}{$I^{\text{FREQ}}_{(l,f)}$}} & Indicator parameter: 1 if \ac{gNB} operating on frequency $f \in \mathcal{F}$ can be installed at location $l \in \mathcal{L}$, 0 otherwise\\[-0.05em]
  & \multicolumn{1}{|c|}{\multirow{2}{*}{$R^{\text{TIME}}_{(l,f)}$}} & Temporal scaling factor for a 5G \ac{gNB} operating on frequency $f \in \mathcal{F}$ and installed at location $l \in \mathcal{L}$\\[-0.05em]
\rowcolor{Linen} \cellcolor{DarkLinen} & \multicolumn{1}{|c|}{\multirow{2}{*}{$R^{\text{STAT}}_{(l,f)}$}} & Statistical scaling factor for a 5G \ac{gNB} operating on frequency $f \in \mathcal{F}$ and installed at location $l \in \mathcal{L}$;\\[-0.05em]
 \multirow{-36}{0.5cm}{\begin{sideways}\textbf{Parameters}\end{sideways}} & \multicolumn{1}{|c|}{\multirow{3}{*}{$\alpha_{(l,f)}$}} & Objective weight factor for the service coverage variables, depending on frequency $f \in \mathcal{F}$ and \ac{gNB} location $l \in \mathcal{L}$.\\[-0.05em]
\hline
\rowcolor{Linen} \cellcolor{DarkLinen} & \multicolumn{1}{|c|}{\multirow{3}{*}{$y_{(l,f)}$}} & 5G gNB  equipment binary variable: 1 if a 5G gNB equipment operating on frequency $f \in \mathcal{F}$ is installed at location $l \in \mathcal{L}$, 0 otherwise\\[-0.05em]
 & \multicolumn{1}{|c|}{\multirow{3}{*}{$x_{(p,l,f)}$}} & Binary 5G service variable : 1 if pixel $p \in \mathcal{P}$ is served by 5G gNB at location $l \in \mathcal{L}$ with frequency $f \in \mathcal{F}$, 0 otherwise\\[-0.05em]
\rowcolor{Linen} \cellcolor{DarkLinen} & \multicolumn{1}{|c|}{\multirow{3}{*}{$P^{\text{ADD-TS}}_{(p,f)}$}} & Additional power density received by pixel $p \in \mathcal{P}$ from all the 5G \acp{gNB} operating on frequency $f \in \mathcal{F}$, computed over temporal and statistical scaling factors\\[-0.05em]
 & \multicolumn{1}{|c|}{\multirow{3}{*}{$P^{\text{ADD-NOTS}}_{(p,f)}$}} & Additional power density received by pixel $p \in \mathcal{P}$ from all the 5G gNB operating on frequency $f \in \mathcal{F}$, computed without temporal and statistical scaling factors\\[-0.05em]
 \rowcolor{Linen} \cellcolor{DarkLinen} & \multicolumn{1}{|c|}{\multirow{3}{*}{$w_p$}} & Pixel in exclusion zone binary variable : 1 if pixel $p \in \mathcal{P}$ falls inside an exclusion zone of an installed 5G gNB, 0 otherwise\\[-0.05em] \multirow{-12}{0.5cm}{\begin{sideways}\textbf{Variables}\end{sideways}}  & \multicolumn{1}{|c|}{$C^{\text{TOT}}$} &  {Variable storing the total installation costs for the installed 5G \acp{gNB}.}\\
\hline
\end{tabular}
\vspace{-5mm}
\end{table}

\section{SIR Linearization}
\label{app:linearization}

In order to linearize Eq.~(\ref{eq:sir_not_linear}), we initially exploit the following equivalence:
\begin{equation}
\label{eq:equivalent_form}
\sum_{l_2 \neq l \in \mathcal{L}}\beta^2_{(l,p,l_2,f)}\cdot y_{(l_2,f)} = \sum_{l_2 \in \mathcal{L}}\beta^2_{(l,p,l_2,f)}\cdot y_{(l_2,f)} - \beta^2_{(l,p,l,f)}\cdot y_{(l,f)}
\end{equation}
By assuming that the right-hand side of Eq.~(\ref{eq:equivalent_form}) is greater than or equal to 0, we replace the denominator of Eq.~(\ref{eq:sir_not_linear}) with Eq.~(\ref{eq:equivalent_form}), thus obtaining:
\begin{eqnarray}
\label{eq:sir_not_linear_2}
\beta^2_{(l,p,l,f)}\cdot y_{(l,f)} \geq S^{\text{MIN}}_{f} \cdot x_{(p,l,f)} \cdot & & \nonumber \\ \cdot \left[\sum_{l_2 \in \mathcal{L}}\beta^2_{(l,p,l_2,f)}\cdot y_{(l_2,f)}  -  \beta^2_{(l,p,l,f)}\cdot y_{(l,f)}\right], && \nonumber \\  \forall p \in \mathcal{P}, l \in \mathcal{L}, f \in \mathcal{F} &&
\end{eqnarray}

We then divide both sides of the constraint by the left hand side term, thus obtaining:
\begin{eqnarray}
\label{eq:sir_not_linear_3}
S^{\text{MIN}}_{f} \cdot x_{(p,l,f)} \cdot \left[\frac{\sum_{l_2 \in \mathcal{L}}\beta^2_{(l,p,l_2,f)}\cdot y_{(l_2,f)}}{\beta^2_{(l,p,l,f)}\cdot y_{(l,f)}} - 1 \right] \leq 1, \nonumber \\ \quad \forall p \in \mathcal{P}, l \in \mathcal{L}, f \in \mathcal{F} 
\end{eqnarray}

The previous constraint is equivalent to the following one:
\begin{eqnarray}
\label{eq:sir_not_linear_3bis}
 S^{\text{MIN}}_{f} \left[\frac{\sum_{l_2 \in \mathcal{L}}\beta^2_{(l,p,l_2,f)}\cdot y_{(l_2,f)} \cdot x_{(p,l,f)}}{\beta^2_{(l,p,l,f)}\cdot y_{(l,f)}} - x_{(p,l,f)} \right] \leq 1,  \nonumber \\
\forall p \in \mathcal{P}, l \in \mathcal{L}, f \in \mathcal{F} 
\end{eqnarray}

By recalling constraint (\ref{eq:cov_pixel}), we know that $x_{(p,l,f)}=1$ only if $y_{(l,f)}=1$ (and both Eq.~(\ref{eq:cov_pixel}) and Eq.~(\ref{eq:maximum_number_of_cells}) are ensured). In other words, $x_{(p,l,f)}$ can not be set to 1 if the \ac{gNB} operating on $f$ is not installed in $l$, i.e., $y_{(l,f)}=0$. As a result, the ratio $x_{(p,l,f)}/y_{(l,f)}$ can be simply expressed as $x_{(p,l,f)}$. Consequently, constraint (\ref{eq:sir_not_linear_3bis}) can be rewritten in the following equivalent form:
\begin{eqnarray}
\label{eq:sir_not_linear_4}
S^{\text{MIN}}_{f} \cdot \left[\sum_{l_2 \in \mathcal{L}}\frac{\beta^2_{(l,p,l_2,f)}}{\beta^2_{(l,p,l,f)}}\cdot y_{(l_2,f)}\cdot x_{(p,l,f)} - x_{(p,l,f)} \right] \leq 1, \nonumber \\ \quad \forall p \in \mathcal{P}, l \in \mathcal{L}, f \in \mathcal{F} 
\end{eqnarray}

The previous constraint can be easily linearized by: \textit{i}) introducing the binary auxiliary variable $v_{(l,p,l_2,f)} \in \{0,1\}$, \textit{ii}) replacing (\ref{eq:sir_not_linear_4})  with the following set of constraints:
\begin{equation}
\label{eq:sir_aux_1app}
v_{(l,p,l_2,f)} \leq x_{(p,l,f)}, \quad \forall p \in \mathcal{P}, l \in \mathcal{L}, l_2 \in \mathcal{L},  f \in \mathcal{F} 
\end{equation}

\begin{equation}
\label{eq:sir_aux_2app}
v_{(l,p,l_2,f)} \leq y_{(l_2,f)}, \quad \forall p \in \mathcal{P}, l \in \mathcal{L}, l_2 \in \mathcal{L},  f \in \mathcal{F} 
\end{equation}

\begin{eqnarray}
\label{eq:sir_aux_3app}
v_{(l,p,l_2,f)} \geq x_{(p,l,f)} + y_{(l_2,f)} - 1, \quad \quad \quad \quad \quad \quad \nonumber \\ \quad \forall p \in \mathcal{P}, l \in \mathcal{L}, l_2 \in \mathcal{L},  f \in \mathcal{F} 
\end{eqnarray}

\begin{eqnarray}
\label{eq:sir_linearapp}
S^{\text{MIN}}_{f} \cdot \left[\sum_{l_2 \in \mathcal{L}}\frac{\beta^2_{(l,p,l_2,f)}}{\beta^2_{(l,p,l,f)}}\cdot v_{(l,p,l_2,f)} - x_{(p,l,f)} \right] \leq 1 \nonumber \\ \quad \forall p \in \mathcal{P}, l \in \mathcal{L}, f \in \mathcal{F} 
\end{eqnarray}

\section{\textsc{PLATEA} Computational Complexity}
\label{app:complexity}

\begin{table}[t]
\caption{Computational complexity of the routines and \textsc{PLATEA} algorithm}
\label{tab:comp_compl}
\centering
\begin{tabular}{|c|c|}
\hline
\rowcolor{Coral}	& \\[-0.8em]
\rowcolor{Coral} \textbf{Procedure} & \textbf{Complexity}\\[0.2em]
\hline
& \\[-0.8em]
\textsc{initial\_sol} & $\mathcal{O}(|\mathcal{P}|\times |\mathcal{L}| \times |\mathcal{F}|)$ \\[0.2em]
\hline
\rowcolor{Linen} & \\[-0.8em]
\rowcolor{Linen} \textsc{extract\_sites} & $\mathcal{O}(N^{\text{COMB}})$\\[0.2em]
\hline
	& \\[-0.8em]
 \textsc{install\_check} & $\mathcal{O}(|\mathcal{P}|\times |\mathcal{L}| \times |\mathcal{F}|)$\\[0.2em]
\hline
\rowcolor{Linen} & \\[-0.8em]
\rowcolor{Linen} \textsc{associate\_pixels} & $\mathcal{O}(|\mathcal{P}|\times |\mathcal{L}|^2 \times |\mathcal{F}|)$\\[0.2em]
\hline
	& \\[-0.8em]
 \textsc{compute\_obj} & $\mathcal{O}(|\mathcal{P}|\times |\mathcal{L}| \times |\mathcal{F}|)$\\[0.2em]
\hline
\rowcolor{Linen}	& \\[-0.8em]
\rowcolor{Linen} \textsc{save\_sol} & $\mathcal{O}(|\mathcal{P}|\times |\mathcal{L}| \times |\mathcal{F}|)$\\[0.2em]
\hline
	& \\[-0.8em]
 \textsc{initialize} & $\mathcal{O}(|\mathcal{P}|\times |\mathcal{L}| \times |\mathcal{F}|)$\\[0.2em]
\hline
\rowcolor{Linen}	& \\[-0.8em]
\rowcolor{Linen} \textsc{all\_served} & $\mathcal{O}(|\mathcal{P}|\times |\mathcal{L}| \times |\mathcal{F}|)$\\[0.2em]
\hline
 & \\[-0.8em]
 \textsc{select\_best\_set\_f1} & $\mathcal{O}(N^{\text{COMB}}\times|\mathcal{P}|\times |\mathcal{L}|^2 \times |\mathcal{F}|)$\\[0.2em]
\hline
\rowcolor{Linen}	& \\[-0.8em]
\rowcolor{Linen} \textsc{PLATEA} & $\mathcal{O}(N^{\text{COMB}}\times|\mathcal{P}|\times |\mathcal{L}|^4 \times |\mathcal{F}|)$\\[0.2em]
\hline
\end{tabular}
\end{table}

Tab.~\ref{tab:comp_compl} reports the computational complexity of the routines, functions and the whole \textsc{PLATEA} algorithm. Several considerations hold by analyzing the table. First, we denote with $N^{\text{COMB}}$ the number of  {combinations} that are generated by the \textsc{extract\_sites} routine. Second, the computational complexity of  \textsc{install\_check}, \textsc{associate\_pixels} and \textsc{compute\_obj} are derived from the implementation of constraints Eq.~(\ref{eq:lower_bound_exclusion}), (\ref{eq:upper_bound_exclusion}), (\ref{eq:lin1_res})-(\ref{eq:indicator_parameter}), (\ref{eq:cov_pixel}),(\ref{eq:maximum_number_of_cells}), (\ref{eq:sir_aux_1})-(\ref{eq:sir_linear}), (\ref{tot_cost_bs_installed_computation}), (\ref{eq:tot_obj}). Third, the whole computational complexity of \textsc{PLATEA} grows linearly with the number of pixels and with the number of frequencies. Fourth, although the complexity of \textsc{PLATEA} may appear substantially large at a first glance, due to the term $N^{\text{COMB}} \times |\mathcal{L}|^4$, we remind that the number of candidate sites $|\mathcal{L}|$ is rather limited in practice, due to the intrinsic difficulty in finding suitable locations that can host \ac{gNB} equipment. In addition, in our work we constrain $N^{\text{COMB}}$ to be in the same order of magnitude of $|\mathcal{L}|$. Therefore, overall complexity of \textsc{PLATEA} is in the order of $\mathcal{O}(|\mathcal{P}|\times |\mathcal{L}|^5 \times |\mathcal{F}|)$.

\section{EMF measurements in the TMC scenario}
\label{app:emf_measurements}

\begin{figure}[t]
	\centering
 	\subfigure[Pole mounted]
	{
		\includegraphics[width=4cm]{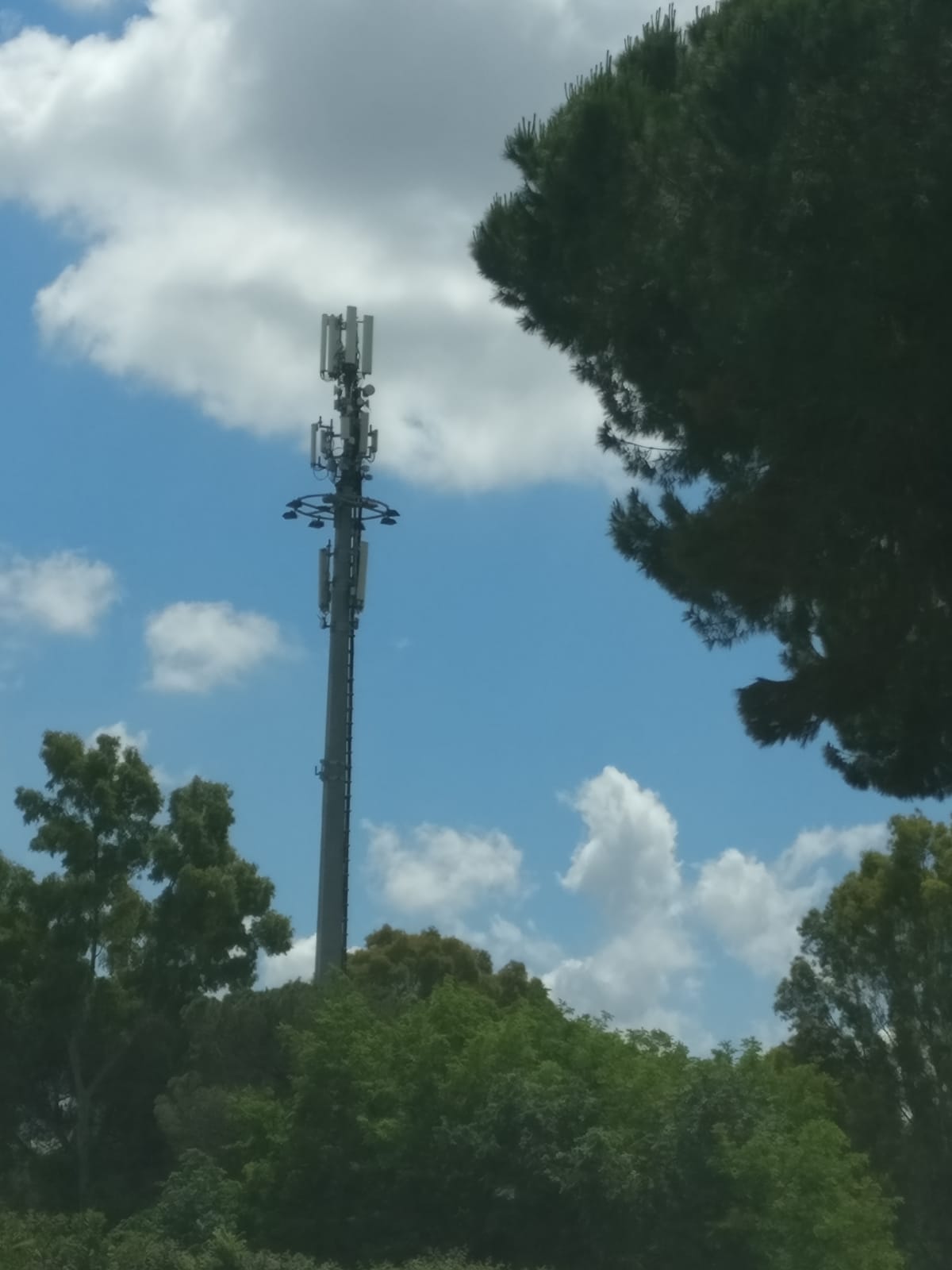}
		\label{fig:polemounted}

	}
 	\subfigure[Roof mounted (free)]
	{
		\includegraphics[width=4cm]{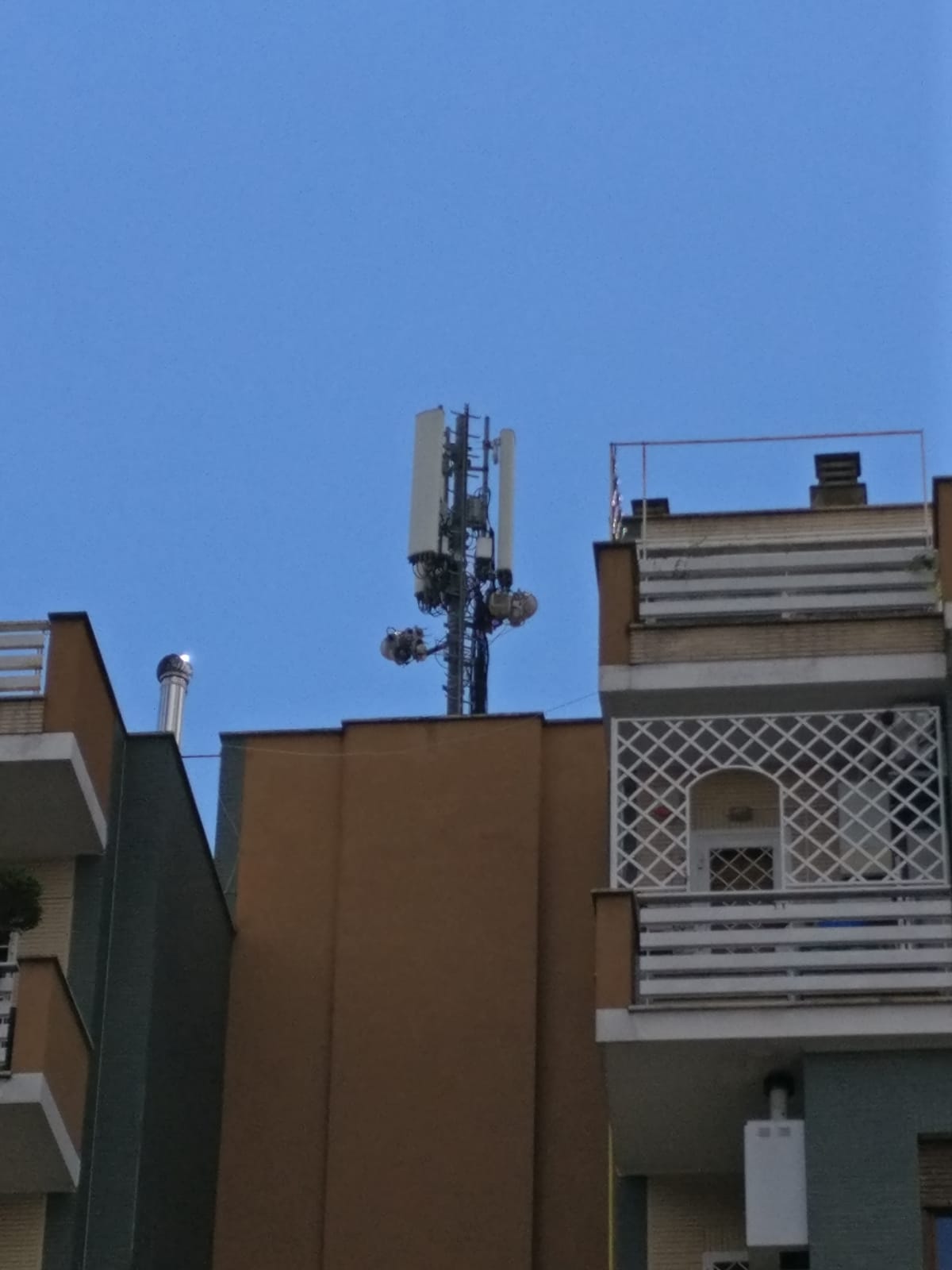}
		\label{fig:roofmountedfree}
	}

 	\subfigure[Roof mounted (hidden)]
	{
		\includegraphics[width=5cm]{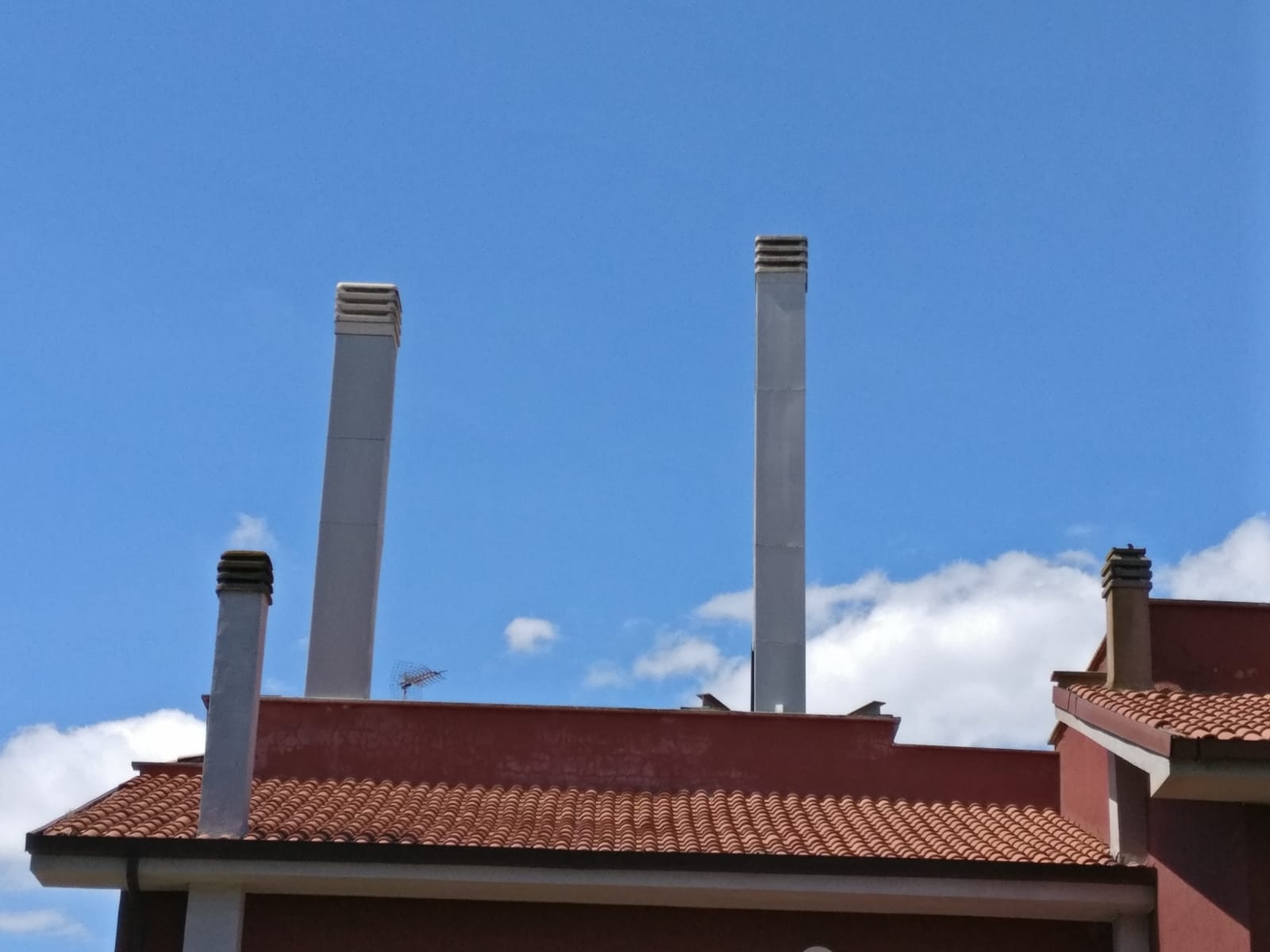}
		\label{fig:roofmountedmasked}
	}
	\caption{Three examples of pre-5G Base Stations serving the \ac{TMC} neighborhood.}
	\label{fig:gnb_servingtmc}
\end{figure}

\begin{figure}[t]
\centering
\includegraphics[width=5.8cm]{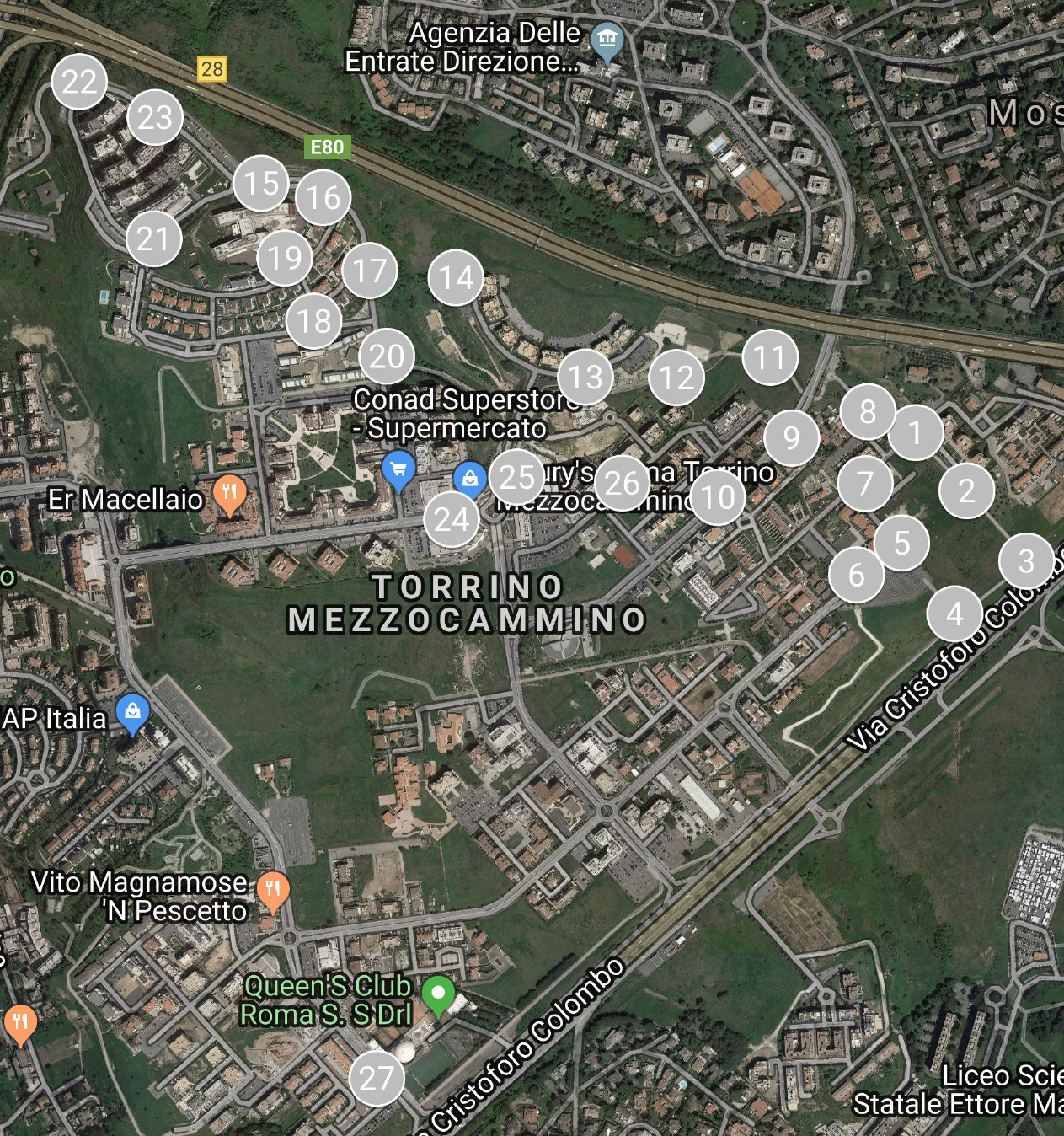}
\caption{Measurement locations in the \ac{TMC} scenario.}
\label{fig:measurement_ID}
\end{figure}

We describe here the methodology adopted to perform the \ac{EMF} measurements in the \ac{TMC} scenario, as well as the main outcomes from the measurements.
Actually, the \ac{TMC} neighborhood does not host any installation of legacy pre-5G Base Stations, mainly due to the following reasons: \textit{i}) \ac{TMC} is a relatively new neighborhood, which was build during the last decade and \textit{ii}) all the requests done by operators to install Base Stations in the neighborhood have been denied by the municipality, since the selected locations did not ensure the minimum distance from the sensitive places. As a result, the cellular service over \ac{TMC} is provided by a set of Base Stations installed in other neighborhoods. We refer the interested reader to Fig.~8 of \cite{chiaraviglio2018planning} for the maps reporting the localization of pre-5G Base Stations serving \ac{TMC}. In brief, these Base Stations include pole mounted and roof mounted installations (with some examples reported in Fig.~\ref{fig:gnb_servingtmc}), mainly close to the north-east and east borders of the neighborhood. Therefore, rather than measuring the \ac{EMF} levels in each pixel of  \ac{TMC} neighborhood (which would require a huge amount of time), we concentrate on the \ac{TMC} zones in close proximity to the Base Stations installed in the other neighborhoods, since these zones are expected to receive the highest \ac{EMF} exposure levels. To this aim, Fig.~\ref{fig:measurement_ID} reports the considered measurement locations, each of them labelled with a unique ID. In addition, we place the \ac{EMF} meter in outdoor locations and in general in zones not covered by obstacles. To this aim, Fig.~\ref{fig:measurementlocations} reports three examples of measurement locations, by differentiating between: roadside positioning (Fig~\ref{fig:roadside}), countryside positioning (Fig.~\ref{fig:countryside}), and positioning in the main square (Fig.~\ref{fig:mainsquare}).

\begin{figure}[t]
	\centering
 	\subfigure[Roadside]
	{
		\includegraphics[width=4.4cm,angle=270]{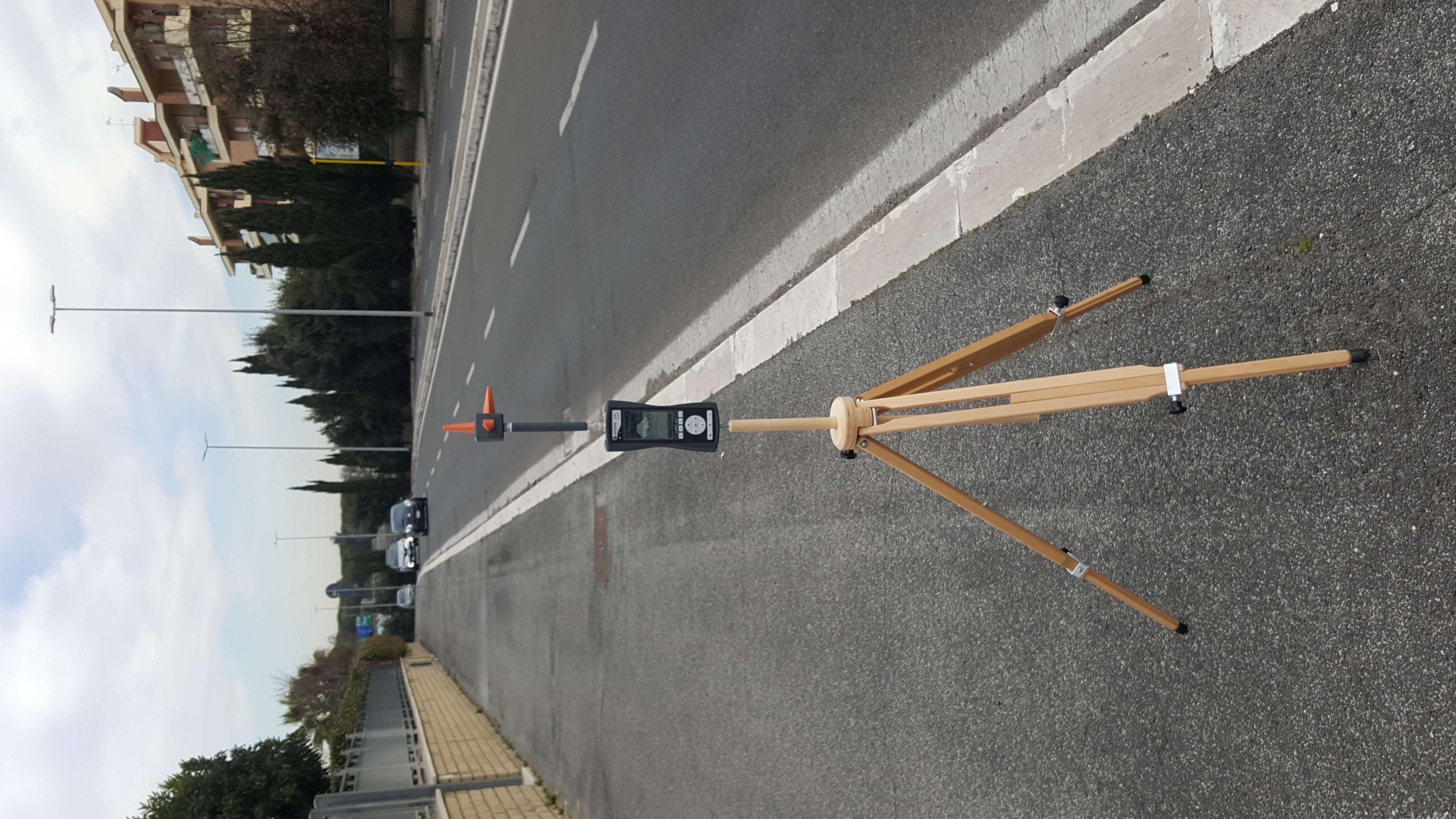}
		\label{fig:roadside}

	}
 	\subfigure[Countryside]
	{
		\includegraphics[width=4.4cm,angle=270]{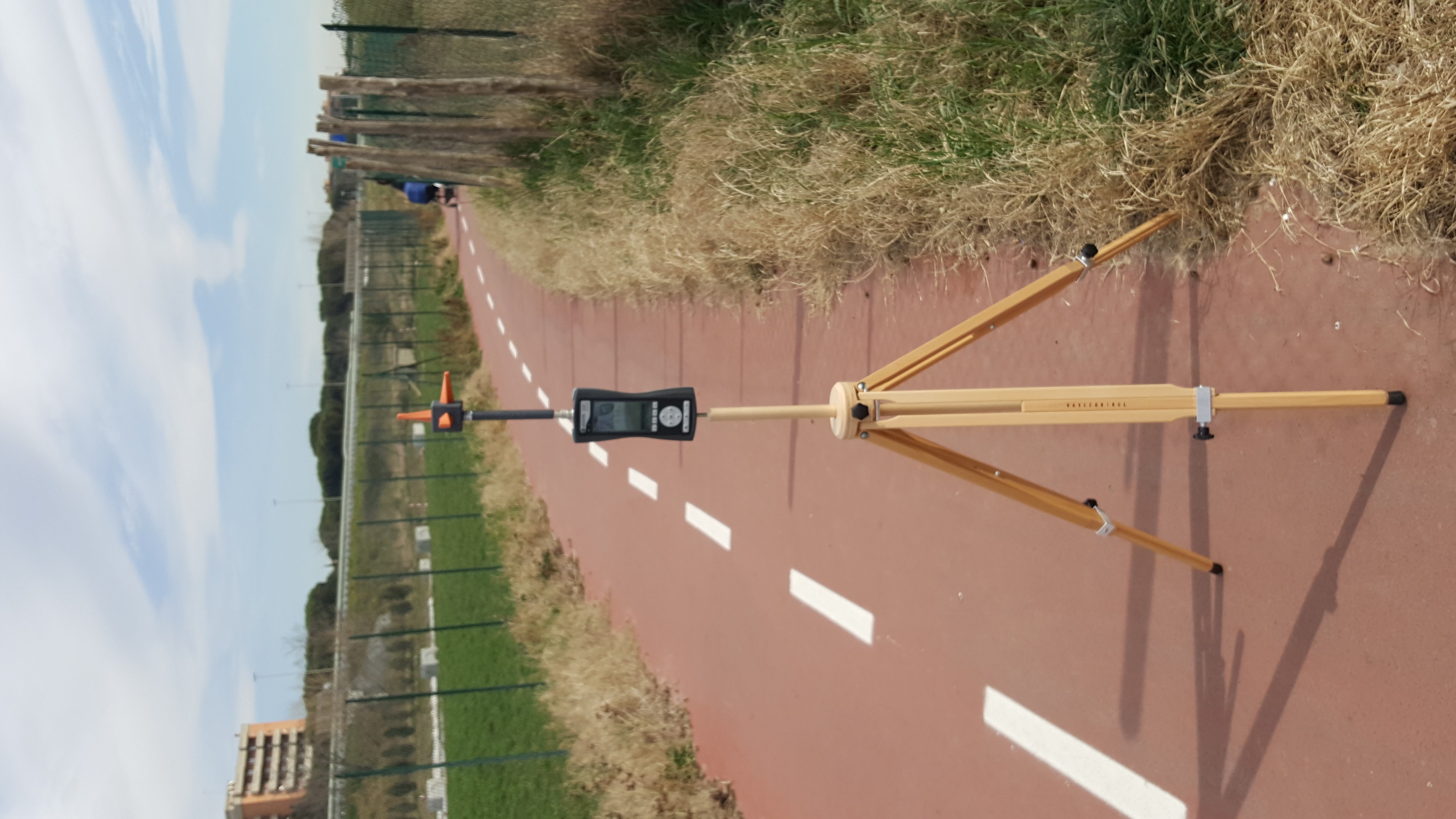}
		\label{fig:countryside}
	}
 	\subfigure[Main Square]
	{
		\includegraphics[width=4.4cm,angle=270]{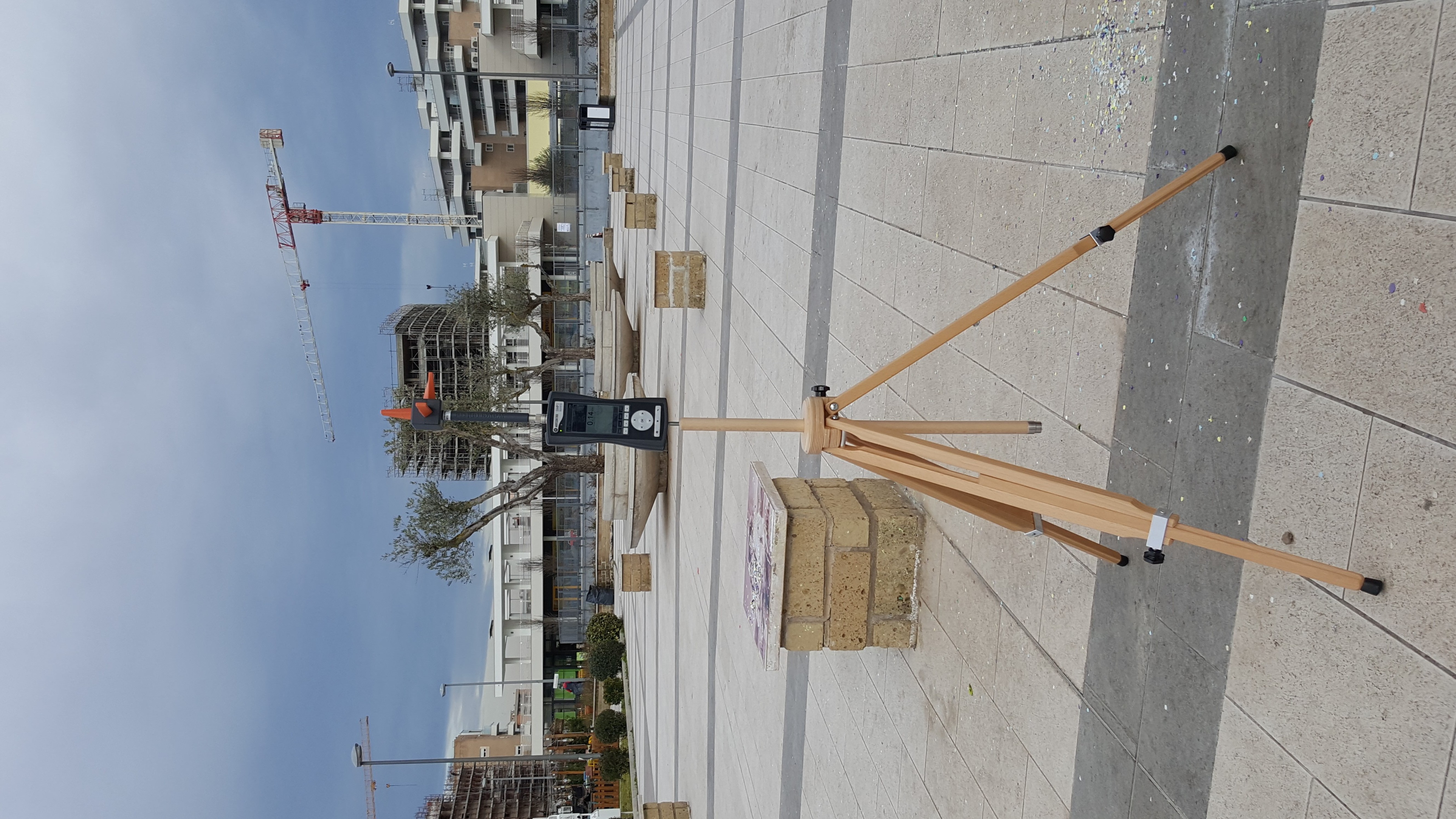}
		\label{fig:mainsquare}
	}

	\caption{Three examples of \ac{EMF} measurement location in the \ac{TMC} neighborhood.}
	\label{fig:measurementlocations}
\end{figure}

In the following, we measure the electric field in each measurement point, by adopting a professional \ac{EMF} meter, whose settings are detailed in Tab.~\ref{tab:meter_settings}. More in depth, the meter provides the total electric field over the set of frequencies used by cellular operators. In addition, all the measurements have been performed during morning/afternoon hours of business days. In this way, the measured electric field is retrieved under moderate/high utilization of the cellular network. In addition, we consider the value of 6~[min] as the reference time interval to compute the average electric field (in accordance to regulation \textit{R1} of Tab.~\ref{tab:regulation_comparison}). Although Italian regulations integrate longer time-intervals to compute the average electric field (see e.g., \textit{R3}-\textit{R4} of Tab.~\ref{tab:regulation_comparison}), we believe that the considered scenario is rather conservative, since the measurements are performed during moderate/high utilization of the cellular network, and hence during time intervals during which the electric field is higher compared to low traffic conditions (e.g., at night, during holidays, during week-ends).

\begin{table}[t]
\caption{Main settings/features of the \ac{EMF} meter}
\label{tab:meter_settings}
\begin{tabular}{|c|c|}
\hline
\rowcolor{Coral} \textbf{Setting/Feature} & \textbf{Value}\\
\hline
Measured Frequencies & \begin{minipage}{5cm}700-900~[MHz], 1800-1900~[MHz], 2100~[MHz], 2600~[MHz]\end{minipage}\\ 
\hline
\rowcolor{Linen} Measurable \ac{EMF} Range & 0.04-65~[V/m]\\
\hline
Average Interval & 6~[min] \\
\hline
\rowcolor{Linen} Height from ground & 1.5~[m]\\
\hline
\end{tabular}
\end{table}

Fig.~\ref{fig:average_emf} reports the average \ac{EMF} levels over the measurement locations. Several consideration hold by analyzing the figure. First, the average electric field is always pretty low, i.e., always lower than 0.9~[V/m]. Second, the electric field generally varies across the locations. Third, very low levels of electric field are measured over locations not in proximity to the \ac{TMC} border (e.g., locations with ID 24, 25, 26, 10 of Fig.~\ref{fig:measurement_ID}). This fact further corroborates our intuition that the electric field from pre-5G Base Stations is almost negligible for most of the pixels in the \ac{TMC} neighborhood.

\begin{figure}[t]
\centering
\includegraphics[width=7cm]{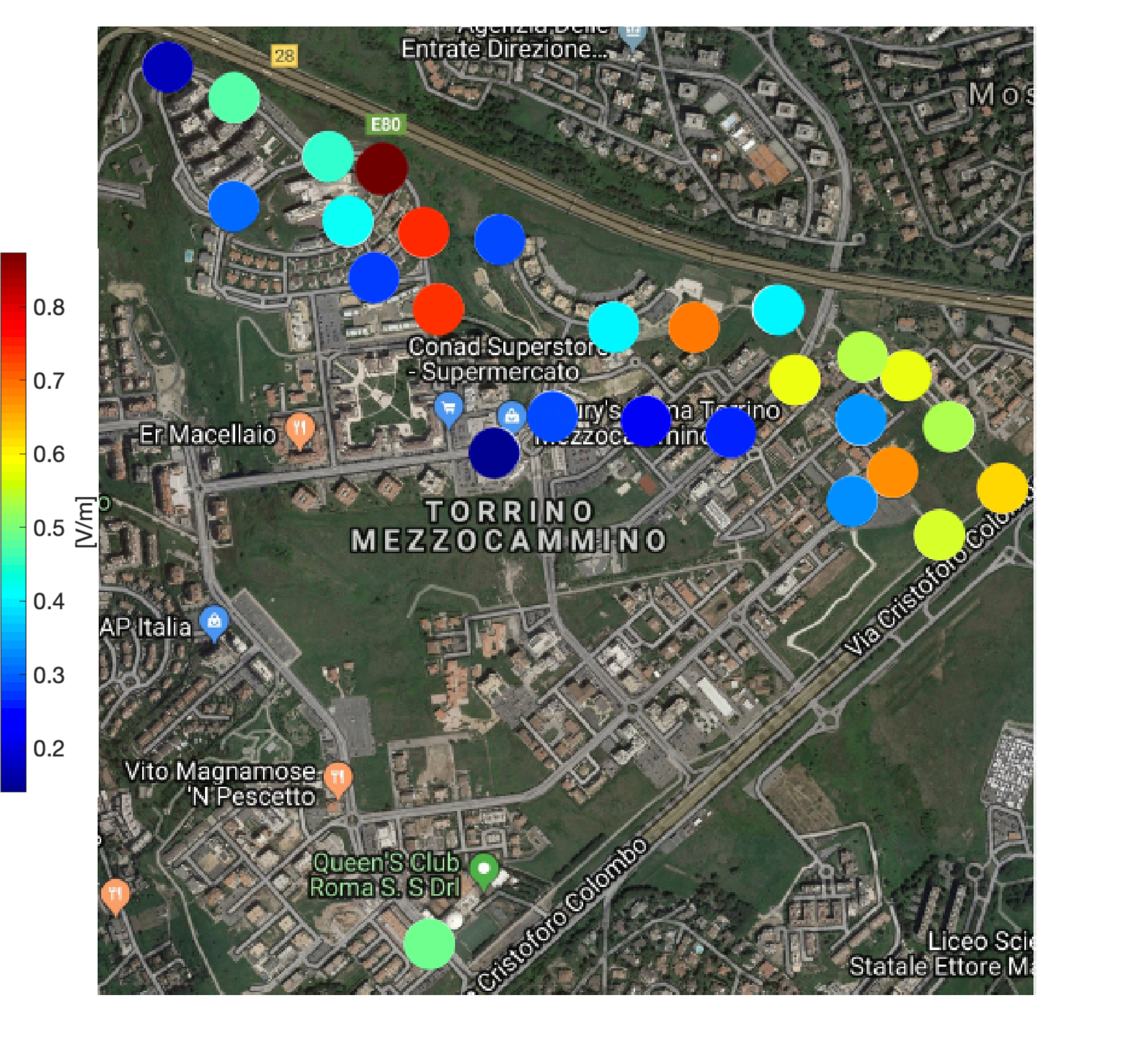}
\caption{Average \ac{EMF} of each measurement in the \ac{TMC} scenario.}
\label{fig:average_emf}
\end{figure}

In the following part of this step, we analyze in more detail the \ac{EMF} measurements. To this aim, Fig.~\ref{fig:average_emf_conf_intervals} reports the electric field vs. the measurement ID. Bar report average electric field values, while error ranges denote the confidence intervals (which are computed by assuming a 95\% of confidence levels). Interestingly, we can note that the confidence interval tends to be reduced when the measured electric field level decreases. 
\begin{figure}[t]
	\centering
	\includegraphics[width=7cm]{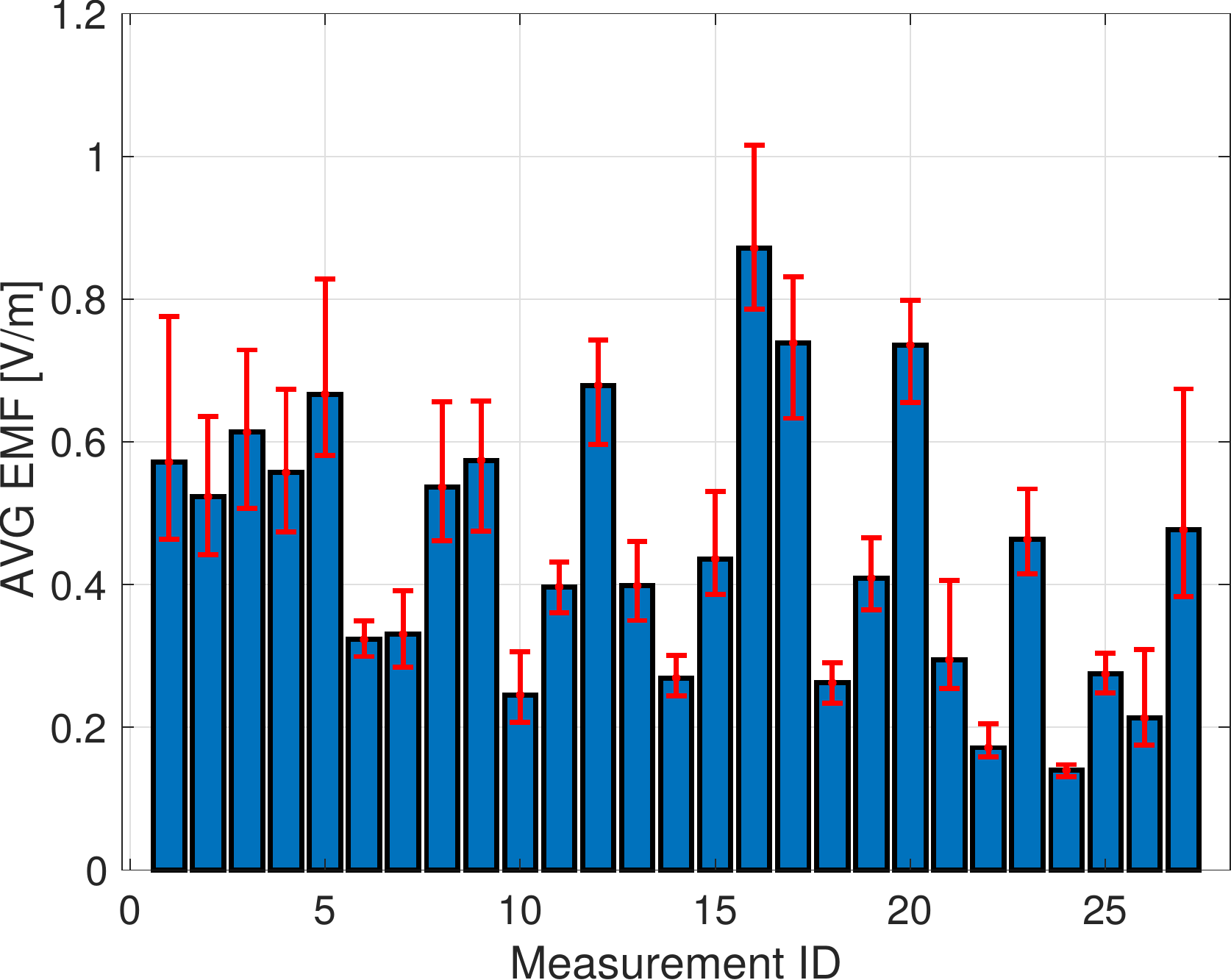}
	\caption{Measured electric field vs. measurement ID.}
         \label{fig:average_emf_conf_intervals}
\end{figure}

Finally, we remind that all the \ac{EMF} measurements based on electric field are converted to power density  {values} by applying Eq.~(\ref{eq:tot_emf}) with $Z_0$=377~[$\Omega$], in accordance with \cite{itutk70}.

\begin{algorithm*}[t]
\small
\caption{Pseudo-Code of the \textsc{Evaluation Algorithm (EA)}}
\label{alg:evalalgo}
	\textbf{Input:} num\_f1 of deployed \acp{gNB}  with frequency $f1$, num\_f2 of deployed \acp{gNB}  with frequency $f2$\\  
	\textbf{Output:} flag\_check flag with installation status (true = installation successful, false = installation unsuccessful),variables $y$, $x$, $P^{\text{ADD-TS}}$, $P^{\text{ADD-NOTS}}$, $w$, $C^{\text{TOT}}$
       \begin{algorithmic}[1]	
	\State{cand\_sites=\textsc{generate\_sites}(num\_f1,num\_f2); \texttt{// Random generation of a candidate deployment}}
	\State{[x\_temp y\_temp pd\_temp]=\textsc{initialize}(cand\_sites);}
	\State{[flag\_check, pd\_temp]=\textsc{install\_check}(cand\_sites, y\_temp, pd\_temp); \texttt{// Based on Eq.~(\ref{eq:lower_bound_exclusion}), (\ref{eq:upper_bound_exclusion}), (\ref{eq:lin1_res})-(\ref{eq:indicator_parameter})}}
	\If{flag\_check==true}
		\State{[x\_temp]=\textsc{associate\_pixels}(y\_temp, x\_temp); \texttt{// Based on Eq.~(\ref{eq:cov_pixel}),(\ref{eq:maximum_number_of_cells}),(\ref{eq:sir_aux_1})-(\ref{eq:sir_linear})}}
		\State{[$y$, $x$, $P^{\text{ADD-TS}}$, $P^{\text{ADD-NOTS}}$, $w$, $C^{\text{TOT}}$]=\textsc{save\_sol}(y\_temp, x\_temp, pd\_temp); \texttt{// Solution saving}}
	\EndIf
	\end{algorithmic}
\end{algorithm*}

\begin{algorithm*}[t]
\small
\caption{Pseudo-Code of the \textsc{Maximum Coverage Macro Algorithm (MCMA)}}
\label{alg:mcaalgo}
	\textbf{Input:} num\_f1 of deployed \acp{gNB}  with frequency $f1$\\  
	\textbf{Output:} flag\_sol flag with installation status (true = feasible solution found, false = no feasible solution found),variables $y$, $x$, $P^{\text{ADD-TS}}$, $P^{\text{ADD-NOTS}}$, $w$, $C^{\text{TOT}}$
       \begin{algorithmic}[1]	
	\State{num\_f2\_max=$\sum_{l \in \mathcal{L}} I^{\text{FREQ}}_{(f2,l)}$;}
	\State{flag\_end=false;}
	\State{flag\_sol=false;}
	\State{num\_f2=1;}
	\While{flag\_end==false}
		\State{cand\_sites=\textsc{generate\_sites}(num\_f1,num\_f2); \texttt{// Random generation of a candidate deployment}}		
		\State{[x\_curr y\_curr pd\_curr]=\textsc{initialize}(cand\_sites);}		
		\State{[flag\_check, pd\_temp]=\textsc{install\_check}(cand\_sites, y\_curr, pd\_curr); \texttt{// Based on Eq.~(\ref{eq:lower_bound_exclusion}), (\ref{eq:upper_bound_exclusion}), (\ref{eq:lin1_res})-(\ref{eq:indicator_parameter})}}
		\If{flag\_check==true}
			\State{[x\_curr]=\textsc{associate\_pixels}(y\_curr, x\_curr); \texttt{// Based on Eq.~(\ref{eq:cov_pixel}),(\ref{eq:maximum_number_of_cells}),(\ref{eq:sir_aux_1})-(\ref{eq:sir_linear})}}
			\State{flag\_sol=true;}
		\EndIf
		\If{(\textsc{all\_served}(x\_curr)==true) or (num\_f2==num\_f2\_max))}
			\If{flag\_check==true}
				\State{[$y$, $x$, $P^{\text{ADD-TS}}$, $P^{\text{ADD-NOTS}}$, $w$, $C^{\text{TOT}}$]=\textsc{save\_sol}(y\_curr, x\_curr, pd\_curr); \texttt{// best Solution Saving}}
			\EndIf
			\State{flag\_end=true; \texttt{// Stop condition}}
										
		\Else
			\State{num\_f2++};
		\EndIf
	\EndWhile
	\end{algorithmic}
\end{algorithm*}

\section{Exhaustive Search: A Feasible Approach?}
\label{app:exhaustive_search}

 {We report here the main insights about the \textsc{Exhaustive Search-based} (\textsc{ES}) algorithm. Ideally, this approach allow to compute the optimal (or near-optimal) planning by evaluating all the possible combinations of micro and macro \acp{gNB} in the considered scenario. More formally, \textsc{ES} is based on the following steps:}
\begin{enumerate}
\item  {extract all the possible combinations of micro and macro \acp{gNB};} 
\item  {for each combination, run the \textsc{install\_check} routine. If the solution is feasible, then run \textsc{associate\_pixels} and \textsc{compute\_obj} routines;}
\item  {select the best feasible solution that minimizes the objective function;}
\end{enumerate}

 {Focusing on the computational complexity of \textsc{ES}, let us assume that the routines to verify the constraints have the same complexity of \textsc{PLATEA}, as shown in Tab.~\ref{tab:comp_compl}. Consequently, the overall complexity of the constraint check is in the order of $O(|\mathcal{P}|\times|\mathcal{L}|^4\times|\mathcal{F}|)$, where we remind that $|\mathcal{P}|$, $|\mathcal{L}|$ and $|\mathcal{F}|$ represent the cardinality for the set of pixels, locations and frequencies, respectively. Clearly, the constraint evaluation has to be done for all the possible site combinations, starting from the subsets with cardinality equal to one and ending with the set $\mathcal{L}$. Since there are a total of $2^{|\mathcal{L}|}$ distinct site combinations (e.g., subsets of candidate locations), the overall complexity of \textsc{ES} is in the order of:}
\begin{equation}
\label{eq:complexity_es}
\mathcal{O}(2^{|\mathcal{L}|}\times|\mathcal{P}|\times|\mathcal{L}|^4\times|\mathcal{F}|)
\end{equation}
 {By analyzing the terms in Eq.~(\ref{eq:complexity_es}), we can conclude that \textsc{ES} is a very time consuming approach in the TMC scenario, which we remind is composed of $|\mathcal{L}|=69$ candidate locations, $|\mathcal{P}|=24318$ pixels and $|\mathcal{F}|=2$ frequencies. To practically demonstrate the complexity of this solution, we have run the \textsc{ES} implementation on the same high performance server used to execute \textsc{PLATEA}, \textsc{EA} and \textsc{MCMA}. After a period of 7 days, only $2^6$ out of the total $2^{69}$ subset combinations were analyzed, corresponding to less than $1.1 \cdot 10^{-17}$~[\%] of the solutions space. Although we recognize that the planning problem is solved offline (i.e., not during the management of the network), the searching of the best solution with \textsc{ES} appears to be a very challenging step. In addition, we remind that the variation of input parameters (including e.g., the weights $\alpha_{(l,f)}$ in the objective function) would require to further increase the required number of runs. Therefore, this evidence further motivates the need of adopting sub-optimal (yet meaningful) reference algorithms, like the \textsc{EA} and \textsc{MCMA} approaches adopted in our work to better position \textsc{PLATEA}.} 

\section{Reference Algorithms}
\label{app:reference_algorithms}

We describe here in more detail the \textsc{EA} and \textsc{MCMA} algorithms.

\subsection{\textsc{EA} Description}
\label{sec:ea}

Alg.~\ref{alg:evalalgo} reports the \textsc{EA} pseudo-code. This algorithm takes as input the number of deployed \acp{gNB} over the two frequencies, which are stored in the \texttt{num\_f1} and \texttt{num\_f2} parameters. \textsc{EA} then returns as output a warning flag, which is set to false if the installation have been unsuccessful (due to the fact the constraints are not ensured), true otherwise. In addition, the information about the obtained planning is returned in the $y$, $x$, $P^{\text{ADD-TS}}$, $P^{\text{ADD-NOTS}}$, $w$, $C^{\text{TOT}}$ variables. Initially (line 1), \textsc{EA} generates a random deployment with \texttt{num\_f1} $f1$ \acp{gNB} and \texttt{num\_f2} $f2$ \acp{gNB}. In the following, \textsc{EA} initializes a set of internal variables (line 2). Then, \textsc{EA} performs in line 3 the feasibility checks of: \textit{i}) power density limits outside exclusion zones, \textit{ii}) minimum distance from sensitive places, \textit{iii}) maximum number of \acp{gNB} installed in each site and \textit{iv}) site installation constraints. If all these constraints are ensured (line 4), \textsc{EA} performs the pixel-\ac{gNB} association (by first iterating over the $f1$ \acp{gNB} and then over the $f2$ \acp{gNB}) and then the solution is saved (lines 5-6).

\textbf{Computational Complexity.} The \textsc{generate\_sites} function has a complexity of $\mathcal{O}(|\mathcal{L}|)$. The complexity  of \textsc{initialize}, \textsc{install\_check}, \textsc{associate\_pixels} and \textsc{save\_sol} is reported in Tab.~\ref{tab:comp_compl}. Therefore, the total computational complexity of \textsc{EA} is in the order of $\mathcal{O}(|\mathcal{P}|\times|\mathcal{L}|^2\times|\mathcal{F}|)$.

\subsection{\textsc{MCMA} Description}

Alg.~\ref{alg:mcaalgo} reports the pseudo-code of \textsc{MCMA} algorithm. The algorithm requires as input the number of $f1$ \acp{gNB} to be installed. \textsc{MCMA} then produces as output a flag, indicating if a feasible solution has been found, and the problem variables $y$, $x$, $P^{\text{ADD-TS}}$, $P^{\text{ADD-NOTS}}$, $w$, $C^{\text{TOT}}$. The main intuition behind \textsc{MCMA} is to sequentially iterate over the number of $f2$ \acp{gNB} to be installed, i.e. from 1 up to \texttt{num\_f2\_max} (lines 5-21). In particular, \textsc{MCMA} generates the set of candidate \acp{gNB} from \texttt{num\_f1} and \texttt{num\_f2} (line 6), initializes the variables (line 7), runs the \textsc{install\_check} function (line 8) and eventually associates the pixels to the installed \acp{gNB} (line 10). The iteration stops if all the pixels have been served or the maximum number of $f2$ \acp{gNB} have been evaluated (line 13). Under this condition, the current solution is eventually saved and the algorithm is ended (lines 14-17). Otherwise, the current number of $f2$ \ac{gNB} is increased (line 18) and lines 5-21 are evaluated again.

\textbf{Computational complexity.} The functions employed by \textsc{MCMA} are the ones also used by \textsc{EA} (detailed in Sec.~\ref{sec:ea}) and \textsc{PLATEA} (detailed in Tab.~\ref{tab:comp_compl}). In addition, the while cycle in line 5 has a complexity of $\mathcal{O}(|\mathcal{L}|)$. Therefore, the total complexity of \textsc{MCMA} is in the order of $\mathcal{O}(|\mathcal{P}|\times|\mathcal{L}|^3\times|\mathcal{F}|)$.


\section{Additional Results}
\label{app:add_results}

In this section, we provide a set of additional results, in order to better position the \textsc{PLATEA} algorithm and the results presented in Sec.~\ref{sec:results}.

\subsection{Throughput Comparison}

Tab.~\ref{tab:compare_results_costs} in Sec.~\ref{sec:results} highlights that \textsc{PLATEA} achieves a better average throughput $T^{\text{AVG}}$ compared to \textsc{EA} and \textsc{MCMA}. However, no indication about the throughput of the single pixels (which we remind is denoted with $T_p$) is provided. Therefore, a natural question is: which is the performance of \textsc{PLATEA} when considering $T_p$ and not $T^{\text{AVG}}$? To answer this question, Fig.~\ref{fig:compare_throughput} reports the \ac{CDF} of the throughput $T_p$ obtained with \textsc{PLATEA}. In addition, the figure shows the \ac{CDF} obtained by running \textsc{EA}. Three considerations hold by analyzing Fig.~\ref{fig:compare_throughput}. First, \textsc{PLATEA} is able to guarantee more than 500~[Mbps] of throughput for more than 25\% of pixels. Second, the percentage of pixels receiving very low throughput is overall pretty limited (less than 10\%). Third, \textsc{EA} performs consistently worse, being its \ac{CDF} clearly moved on the left w.r.t. \textsc{PLATEA}.

\begin{figure}[t]
\centering
\includegraphics[width=7cm]{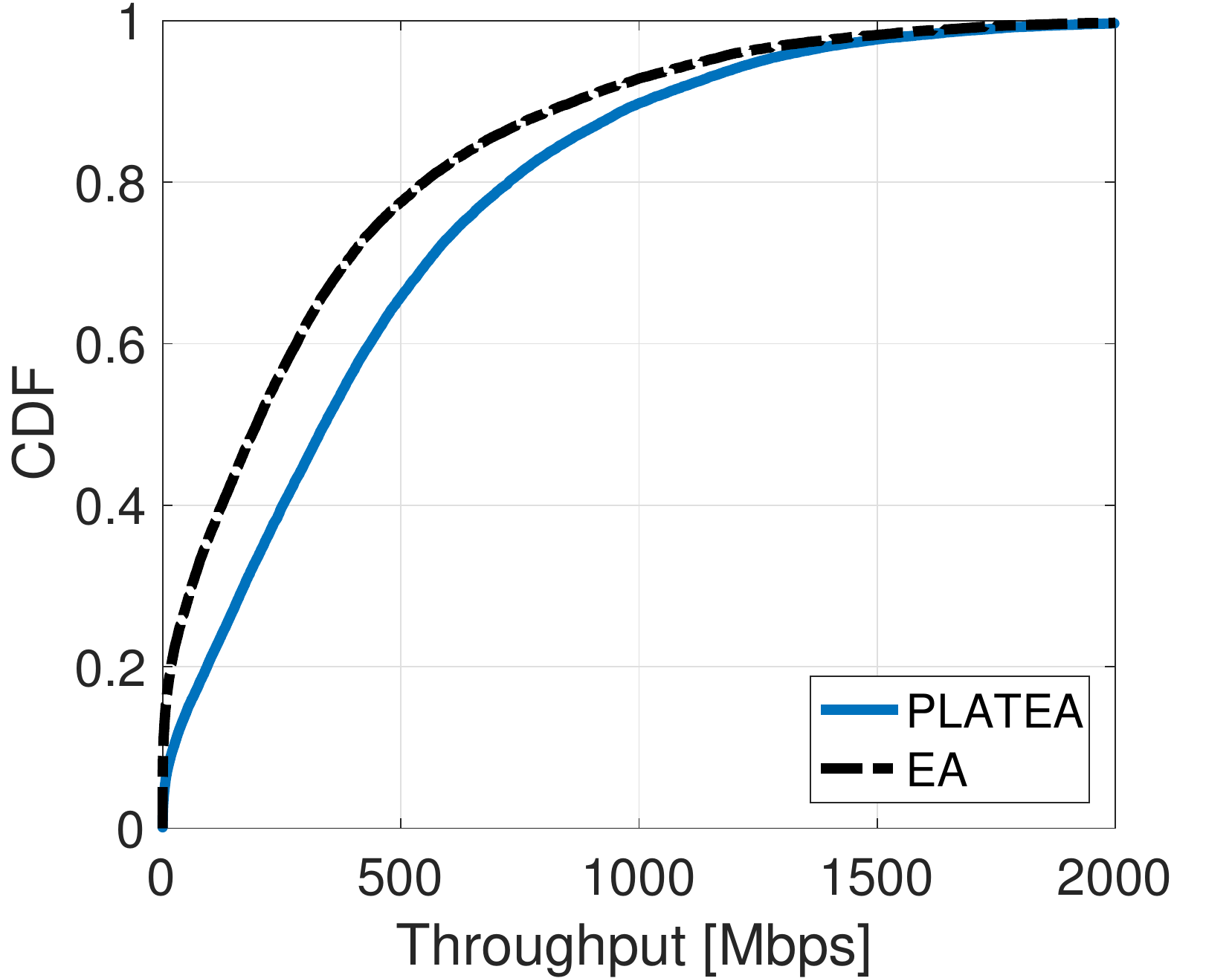}
\caption{CDF of the pixel throughput $T_p$ obtained by \textsc{PLATEA} and \textsc{EA}.}
\label{fig:compare_throughput}
\end{figure}


\subsection{Number of $f2$ \acp{gNB} selected by \textsc{MCMA}}

\begin{figure}[t]
\centering
\includegraphics[width=7cm]{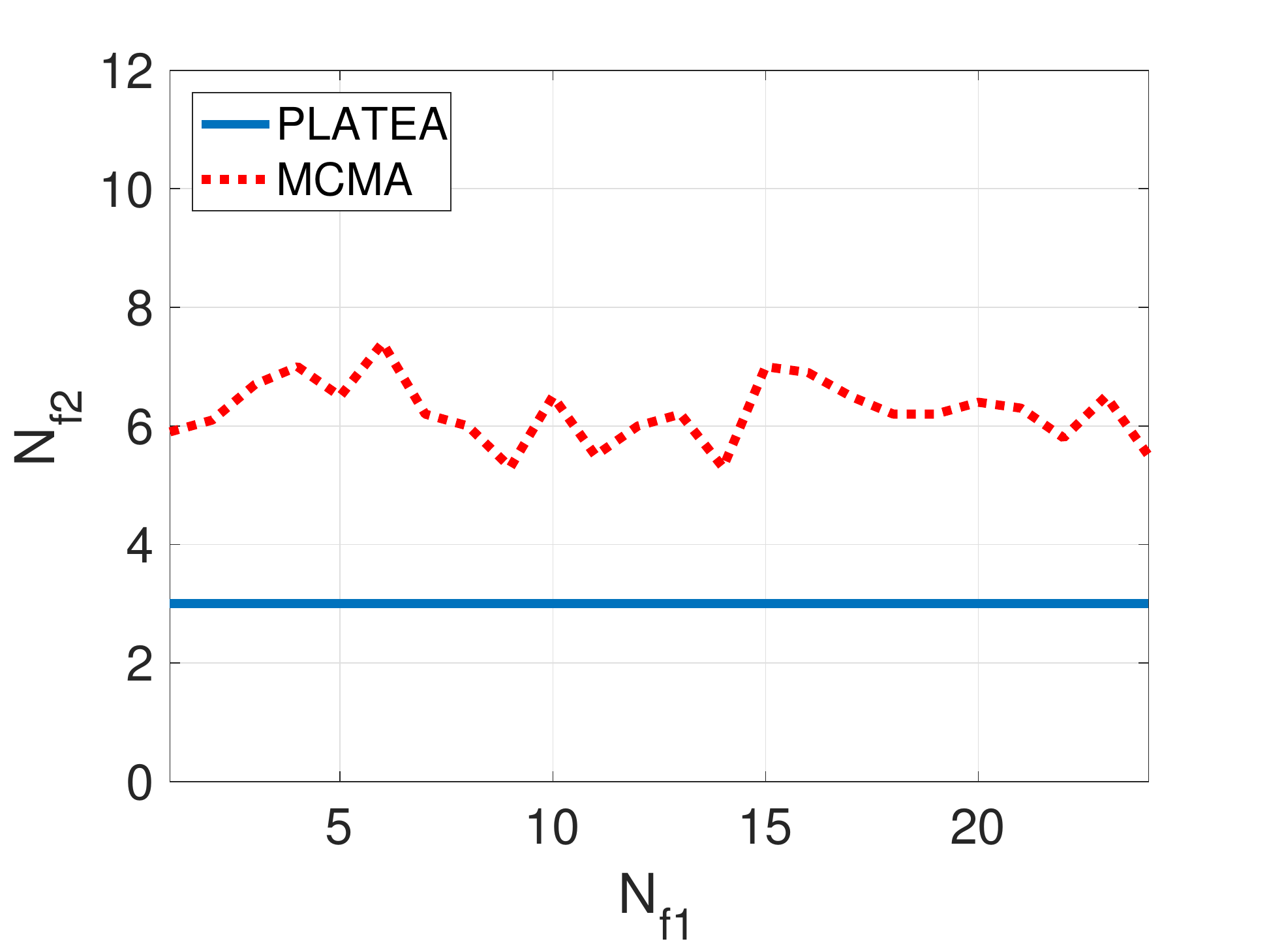}
\caption{Number of $f2$ \acp{gNB} installed by \textsc{PLATEA} and \textsc{MCMA} for different numbers of $f1$ \acp{gNB}.}
\label{fig:comparison_macro_micro}
\end{figure}

According to Tab.~\ref{tab:compare_results_costs} in Sec.~\ref{sec:results}, \textsc{MCMA} requires a consistently higher amount of $f2$ \acp{gNB} compared to \textsc{PLATEA}. However, a natural question is: Does \textsc{MCMA} always install more $f2$ \acp{gNB} w.r.t. \textsc{PLATEA}? To answer this question, we run \textsc{MCMA} by varying \texttt{num\_f1} between 1 and 24 (which corresponds to the maximum number of \ac{gNB} installed by \textsc{PLATEA} for very large values of $\alpha_{(l,f1)}$). For each value of \texttt{num\_f1}, we then perform 10 executions of \textsc{MCMA}, and then we collect the average number of installed $f2$ \acp{gNB} over the 10 runs. Fig.~\ref{fig:comparison_macro_micro} reports the variation of the average number of installed $f2$ \acp{gNB} vs. \texttt{num\_f1}. For completeness, the figure reports also the number of $f2$ \acp{gNB} that are installed by \textsc{PLATEA}. Interestingly, we can note that the average number of installed $f2$ \acp{gNB} is always higher than the one selected by \textsc{PLATEA}. This difference may be explained in the different planning policies adopted by the two solutions. \textsc{MCMA}, in fact, sequentially evaluates an increasing number of $f2$ \ac{gNB}, given the number of $f1$ \acp{gNB} that is passed as input parameter. On the other hand, \textsc{PLATEA} operates a wiser choice, by: i) evaluating a number of  {combinations} that increases with the number of \acp{gNB} that need to be installed, and ii) jointly varying the number of $f1$ \ac{gNB} and $f2$ \ac{gNB} when evaluating the problem constraints and the objective function.

\subsection{Impact of minimum distance constraint}
\label{app:min_distance}

\begin{figure}[t]
\centering
\subfigure[{$E^{\text{AVG}}$~[V/m]}]
{
    \includegraphics[width=42mm]{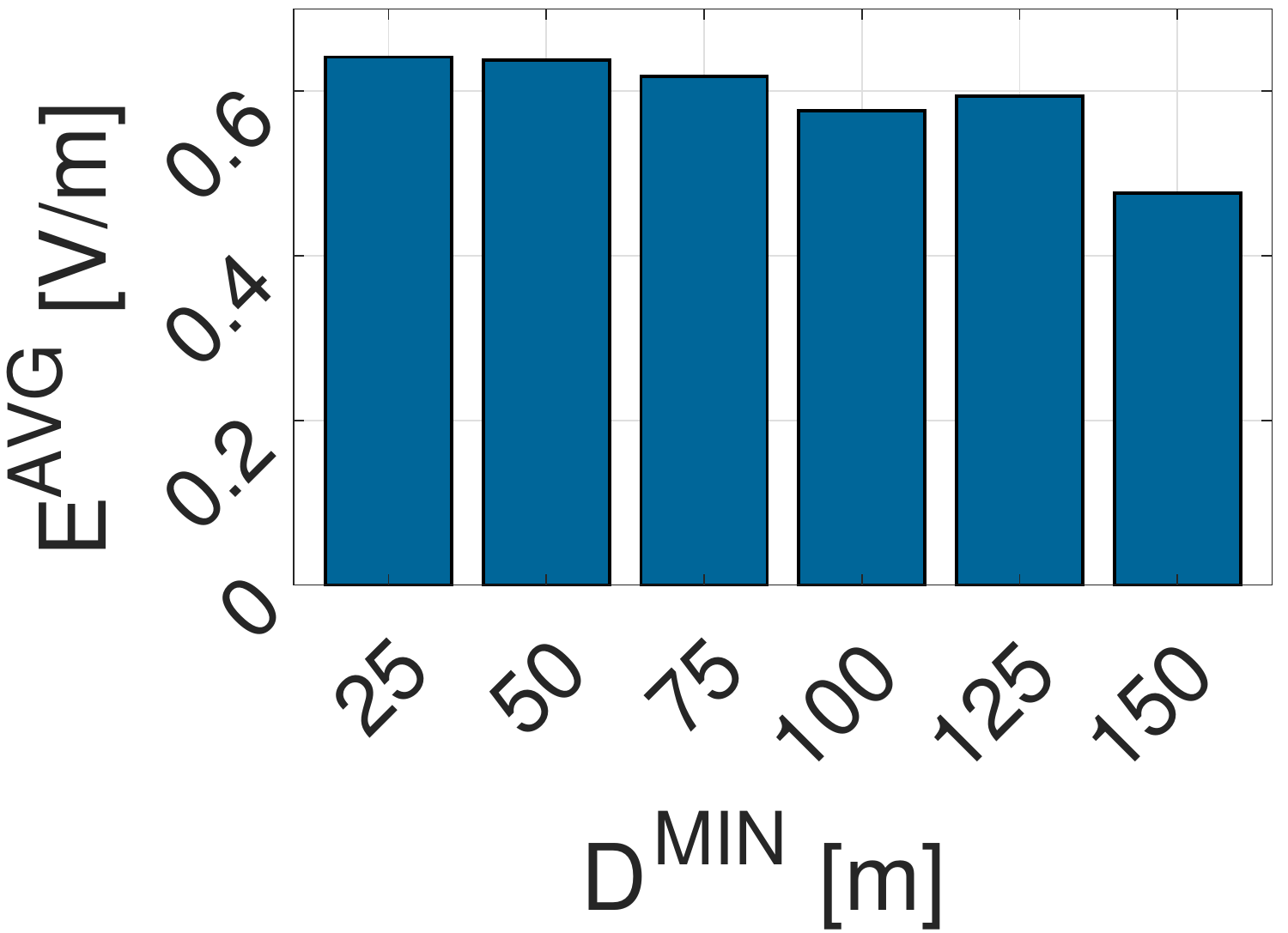}
    \label{fig:dmin_emf}
}
\subfigure[{$T^{\text{AVG}}$~[Mbps]}]
{
    \includegraphics[width=42mm]{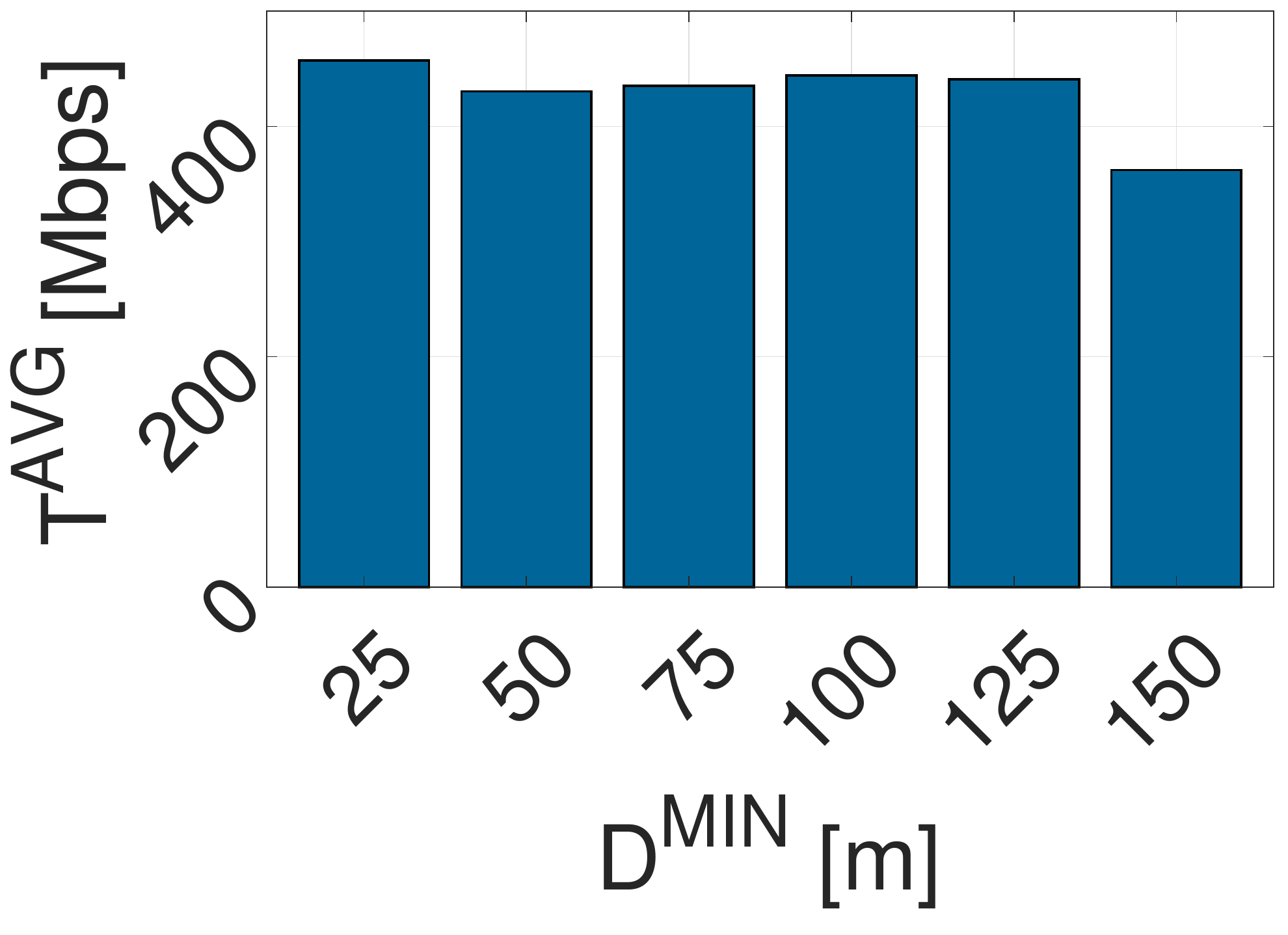}
    \label{fig:dmin_tavg}
}

\subfigure[$N_{\text{f1}}$]
{
    \includegraphics[width=42mm]{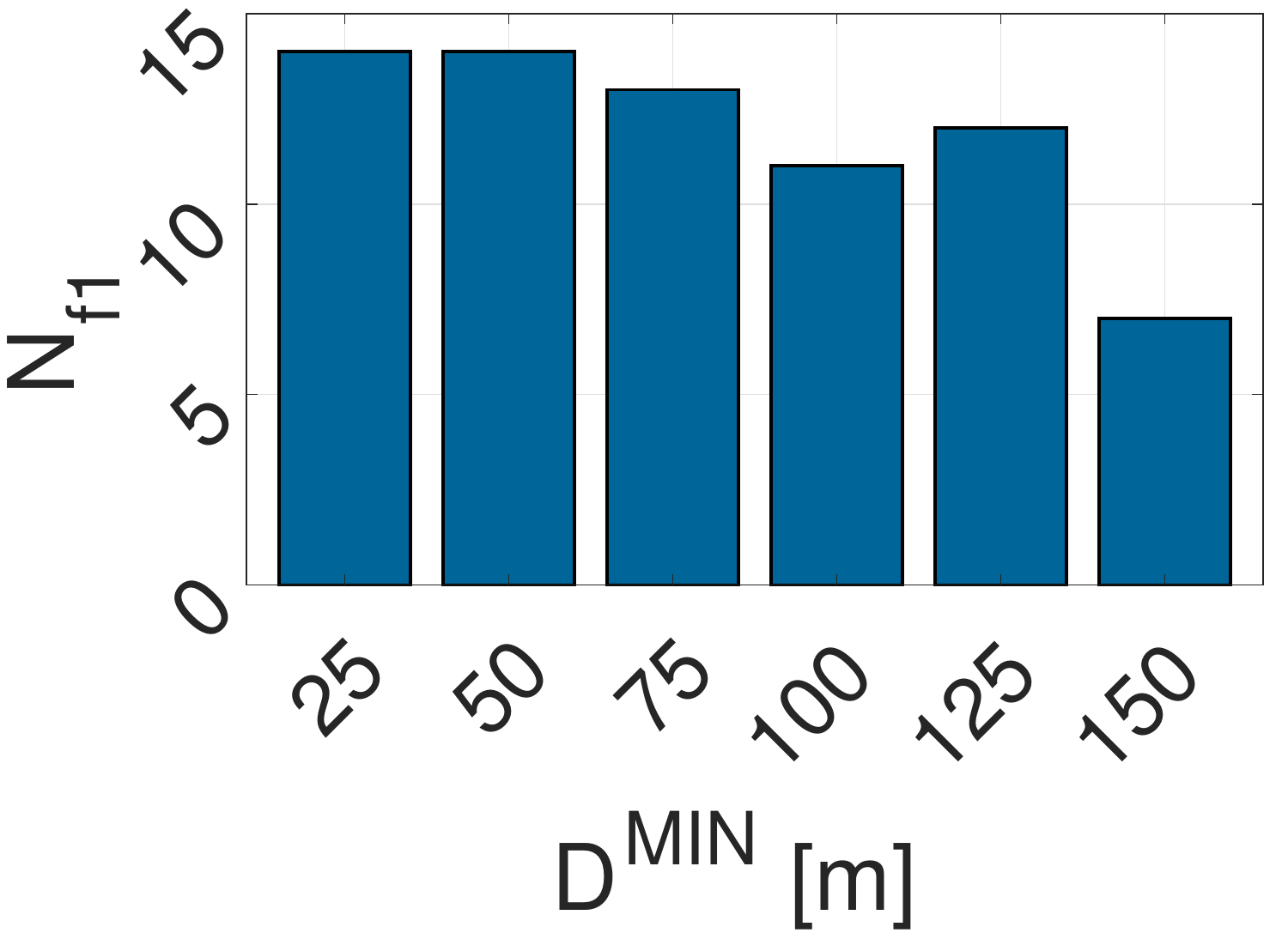}
    \label{fig:dmin_nf2}
}
\subfigure[{$X^{\text{NOT-SERVED}}$~[\%]}]
{
    \includegraphics[width=42mm]{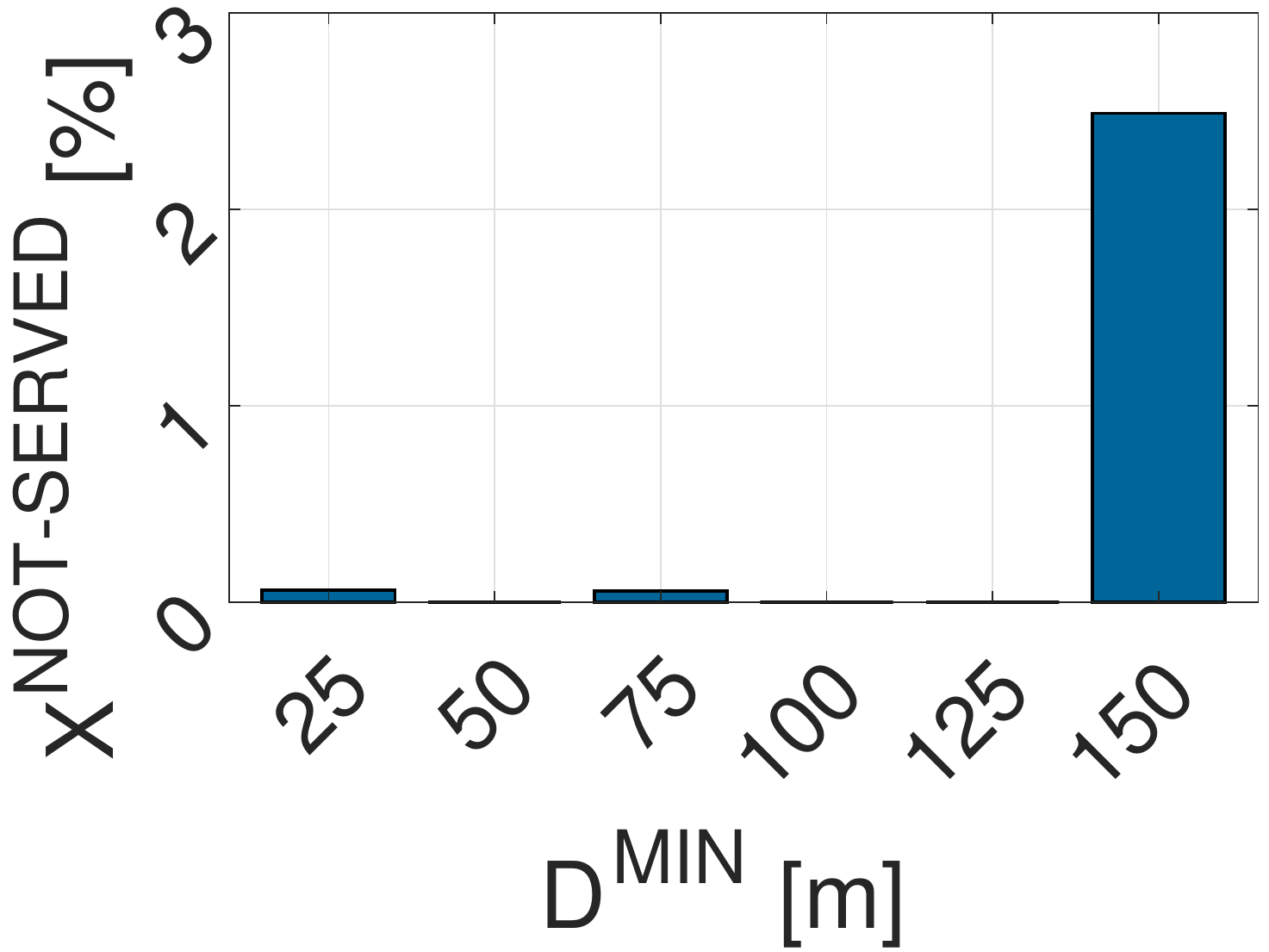}
    \label{fig:dmin_notserved}
}
\caption{Impact of $D^{MIN}$ variation on: \textit{i}) average electric field $E^{\text{AVG}}$, \textit{ii}) average throughput $T^{\text{AVG}}$, \textit{iii}) number $N_{f1}$ of $f1$ \acp{gNB} and \textit{iv}) percentage of not served pixels $X^{\text{NOT-SERVED}}$.}
\label{fig:var_dmin}
\end{figure}

The set of regulations taken under consideration in this work (namely \textit{R6} of Tab.~\ref{tab:regulation_comparison}) includes a minimum distance $D^{\text{MIN}}$ between each installed \acp{gNB} and each sensitive place. A natural question is then: What is the impact of $D^{\text{MIN}}$ variation on the planning? To answer this question, we have assumed that $D^{\text{MIN}}$ can take different values w.r.t. the ones reported in the regulations. In particular, we have considered the following range of values for $D^{\text{MIN}}=\{25,50,75,100,125,150\}$. We have then run \textsc{PLATEA} algorithm for each $D^{\text{MIN}}$ value, and we have collected the performance metrics. Fig.~\ref{fig:var_dmin} reports the obtained results in terms of: \textit{i}) average electric field $E^{\text{AVG}}$, \textit{ii}) average throughput $T^{\text{AVG}}$, \textit{iii}) number $N_{f1}$ of $f1$ \acp{gNB}  and \textit{iv}) percentage of not served pixels $X^{\text{NOT-SERVED}}$. We remind that $D^{\text{MIN}}=100$~[m] is the value currently enforced in the Rome regulations. Interestingly, when $D^{\text{MIN}}<100$~[m], $E^{\text{AVG}}$, $T^{\text{AVG}}$ and $N_{f1}$ tend to be increased, due to the fact that it is possible to install more \acp{gNB} over the territory while ensuring the minimum distance constraint. On the other hand, the opposite holds $D^{\text{MIN}}>100$~[m]. In particular, when $D^{\text{MIN}}=150$~[m], an abrupt decrease of $E^{\text{AVG}}$, $T^{\text{AVG}}$ and $N_{f1}$ is observed. These results are a direct consequence of the $D^{\text{MIN}}$, which prevents the installation of \acp{gNB} in many \ac{TMC} locations. In addition, the setting of $D^{\text{MIN}}=150$~[m] is also detrimental for \ac{UE}, since the percentage of not served pixels is abruptly increased to more than 2\%. Therefore, we can conclude that the variation of $D^{\text{MIN}}$ affects the selected planning, especially for values larger than 100~[m].

\subsection{Impact of pre-5G exposure levels} 
\label{app:pre-exposure}

\begin{table}[t]
\centering
\caption{\textsc{PLATEA} performance vs. different pre-5G exposure terms.}
\label{tab:compare_results_costs_background}
\begin{tabular}{|c|c|c|c|}
\hline
\rowcolor{Coral} & \multicolumn{3}{c|}{\textbf{Pre-5G Exposure}}\\[-0.05em]
\rowcolor{Coral} \multirow{-2}{*}{\textbf{Metric}} & \textbf{1.0~[V/m]} & \textbf{1.5~[V/m]} & \textbf{2.0~[V/m]} \\
\hline
& &  &\\[-0.95em]
         $C^{\text{TOT}}$~[k\euro] & 371.9 & 343.3 & 267.9 \\
\rowcolor{Linen}         $N_{f1}$ & 9.9 & 7.9 & 4 \\
         $N_{f2}$ & 3 & 3.1 & 3  \\
\rowcolor{Linen} & &  &\\[-0.95em]
\rowcolor{Linen}         $X^{\text{SERVED}}_{f1}$ & 8779 & 7409 & 4501\\
& &  &\\[-0.95em]
         $X^{\text{SERVED}}_{f2}$ & 15534  & 16869 & 19789 \\
\rowcolor{Linen} & &  &\\[-0.95em]
\rowcolor{Linen}         $X^{\text{NOT-SERVED}}$~[\%] & 0.01 & 0.16 & 0.11 \\
& &  &\\[-0.95em]
         $T^{\text{AVG}}$~[Mbps] &  {429.99} &  {411.74} &  {374.42} \\
\rowcolor{Linen} & &  &\\[-0.95em]
\rowcolor{Linen}         $E^{\text{AVG}}$~[V/m] & 1.18 & 1.62 &2.06\\
\hline
\end{tabular}
\end{table}

 {The goal of this part is to study the impact of} adding the pre-5G exposure term on the 5G planning. As detailed in Appendix~\ref{app:emf_measurements}, the actual electric field strength in \ac{TMC} hardly exceeds 1~[V/m], even for the pixels that are at the shortest distance and \ac{LOS} conditions w.r.t. the serving \ac{gNB}. On the other hand, the electric field rapidly decreases to negligible values (below 1~[V/m]) as the distance between the pixel and the radiating \acp{gNB} increases. However, in order to introduce a set of conservative (and worst case) scenarios, we assume: \textit{i}) three different settings of pre-5G exposure, namely 1~[V/m], 1.5~[V/m] and 2~[V/m], and \textit{ii}) a uniform term of pre-5G exposure for all the pixels in the \ac{TMC} scenario.\footnote{The application of a background exposure of 2~[V/m] is equivalent to the case in which the limit $L^{\text{RES}}$ equal to 4~[V/m], a value currently in use in many Swiss cantons and in Monaco. Alternatively, the background exposure can be also seen as a margin that is left for the deployment of post-5G networks.} 
As a consequence, the cumulative pre-5G power density is set as $\sum_{f \in \mathcal{F}}P^{\text{BASE}}_{(p,f)}=\{0.00265, 0.00597, 0.0106\}~[\text{W/m}^2], \quad \forall p \in \mathcal{P}$, respectively. In addition, since the same restrictive limit is applied for all the pixels of \ac{TMC}, Eq.~(\ref{eq:power_density_limit_res}) is rewritten as:
\begin{equation}
\label{eq:power_density_limit_res_revisited}
(1-w_p)\cdot\underbrace{\sum_{f \in \mathcal{F}} P^{\text{BASE}}_{(p,f)}}_{\text{pre-5G Exposure Term}} + \underbrace{\sum_{f \in \mathcal{F}} P^{\text{ADD-TS}}_{(p,f)}}_{\text{5G Exposure Term}}  \leq L^{\text{RES}}, \quad \forall p \in  \mathcal{P}^{\text{RES}} 
\end{equation}
where $L^{\text{RES}}=0.1$~[W/m$^2$] (in accordance to \textit{R6} of Tab.~\ref{tab:regulation_comparison}). Intuitively, the introduction of the pre-5G exposure term may limit the amount of 5G \acp{gNB} that are installed over the territory, since it is more challenging to ensure Eq.~(\ref{eq:power_density_limit_res_revisited}) compared to the case in which the pre-5G technologies are dismissed.

Tab.~\ref{tab:compare_results_costs_background} reports the performance metrics of \textsc{PLATEA} (averaged over 10 runs) vs. the different values of pre-5G exposure. When the pre-5G exposure is increased, we can note: \textit{i}) a reduction in the number of $f1$ \acp{gNB}, and consequently of total costs, \textit{ii}) an increase in the number of pixels served by $f2$ \acp{gNB},  \textit{iii}) a throughput decrease, and \textit{iv}) an \ac{EMF} increase, mainly due to the pre-5G exposure term. Overall, these results prove that the performance metrics are impacted by the level of background exposure. However, \textsc{PLATEA} is always able to retrieve a feasible planning, with a percentage of unserved pixels at most equal to 0.16\%. 



\subsection{ {Impact of frequency reuse scheme}} 
\label{app:freq_reuse}

 {Finally, we have considered the case in which the frequency reuse is different than unity. More formally, the frequency reuse factor can be easily changed by varying the value of $\epsilon^{\text{F-REUSE}}_f$, which appear in Eq.~(\ref{eq:capacity}), and therefore in the throughput metrics reported in Tab.~\ref{tab:metrics}. In the following, we consider the following variation of frequency reuse scheme: $\epsilon^{\text{F-REUSE}}_f=\{1,3,7\}$ (for both $f1$ and $f2$). In addition, we adopt the conservative assumption that the change of $\epsilon^{\text{F-REUSE}}_f$ does not affect the \ac{SIR} computation in Eq.~(\ref{eq:sir}).}\footnote{ {Intuitively, the SIR may be improved when the frequency reuse scheme is increased, because a lower number of interfering \acp{gNB} are counted in the \ac{SIR} denominator of Eq.~(\ref{eq:sir}). We refer the reader to} \cite{marzetta}  {for a thorough analysis of this aspect. In this work, we consider a worst-case scenario, in which the increase of $\epsilon^{\text{F-REUSE}}_f$ reduces the available bandwidth of Eq.~(\ref{eq:capacity}), while the SIR in Eq.~(\ref{eq:sir}) is not improved.}} 

\begin{table}[t]
\centering
\caption{ {Impact of the frequency reuse reuse scheme $\epsilon^{\text{F-REUSE}}_f$ on \textsc{PLATEA}.}}
\label{tab:frequency_reuse}
\begin{tabular}{|c|c|c|c|}
\hline
\rowcolor{Coral} & & &\\[-0.95em]
\rowcolor{Coral} \textbf{Metric} & $\epsilon^{\text{F-REUSE}}_f=1$ & $\epsilon^{\text{F-REUSE}}_f=3$ & $\epsilon^{\text{F-REUSE}}_f=7$  \\
\rowcolor{Coral} & & &\\[-0.95em]
\hline
\rowcolor{Linen} & & &\\[-0.95em]
\rowcolor{Linen}         $C^{\text{TOT}}$~[k\euro] & 386.1 & 354.3 & 320.2 \\
         $N_{f1}$ & 10.7 & 8.9 & 7 \\
\rowcolor{Linen}         $N_{f2}$ &  3 &  3 & 3 \\[0.05em]
 & & &\\[-0.95em]
        $X^{\text{SERVED}}_{f1}$ & 9103 & 7465.3 & 5999.5 \\[0.05em]
& & &\\[-0.95em]
\rowcolor{Linen}        $X^{\text{SERVED}}_{f2}$ & 15307 & 16845.3 & 18305.3 \\[0.05em]
 & & &\\[-0.95em]
         $X^{\text{NOT-SERVED}}$~[\%] &  0.03 & 0.03 & 0.05 \\
 & & &\\[-0.95em]
\rowcolor{Linen}         $T^{\text{AVG}}$~[Mbps]  & 428.3 & 144.1 & 60.4 \\
 & & &\\[-0.95em]
        $E^{\text{AVG}}$~[V/m]  &  0.57 & 0.53 & 0.48 \\
\hline
\end{tabular}
\end{table}

 {Tab.~\ref{tab:frequency_reuse} reports the main metrics when \text{PLATEA} solves the TMC scenario under 1, 3, and 7 frequency reuse schemes (with values averaged over 10 independent runs). As expected, the increase of $\epsilon^{\text{F-REUSE}}_f$ notably reduce the throughput $T^{\text{AVG}}$. In particular, since it is more challenging to ensure the minimum traffic constraint for the micro gNBs, \textsc{PLATEA} tends to associate more pixels to the macro gNBs. Consequently, the number of micro gNBs is decreased, and therefore the overall costs  $C^{\text{TOT}}$ are reduced. This is in turn beneficial for the average EMF levels $E^{\text{AVG}}$, which pass from $0.57$~[V/m] when $\epsilon^{\text{F-REUSE}}_f=1$ to 0.48~[V/m] when $\epsilon^{\text{F-REUSE}}_f=7$.}

\end{document}